\documentclass[letterpaper,11pt]{article}
\usepackage[utf8]{inputenc}
\usepackage[margin=1in]{geometry}
\usepackage{authblk}

\usepackage{amsmath}
\usepackage{amsthm}
\usepackage{amssymb}
\usepackage{thmtools}
\usepackage{thm-restate}
\usepackage{enumitem}
\usepackage{dsfont}
\usepackage{xspace}
\usepackage{comment}
\usepackage{xparse}
\usepackage{tikz}
\usepackage{fixme}
\usepackage[colorlinks,urlcolor=red]{hyperref}
\usepackage[capitalize]{cleveref}
\usepackage{subcaption}

\usepackage{caption}

\setlist[enumerate]{nosep}
\setlist[itemize]{nosep}

\newtheorem{problem}{Problem}[]
\newtheorem{theorem}{Theorem}[section]
\newtheorem{question}{Question}[]
\newtheorem{fact}{Fact}[]
\newtheorem{note}{Note}[]
\newtheorem{definition}{Definition}[]
\newtheorem{remark}{Remark}[]
\newtheorem{corollary}{Corollary}[]
\newtheorem{lemma}{Lemma}[section]
\newtheorem{observation}{Observation}[]
\usepackage{algorithm}
\usepackage{algpseudocode}
\usepackage{algorithmicx}  
\usepackage{float}
\usepackage{microtype}
\usepackage{algorithm}
\usepackage{algpseudocode}
\usepackage{algorithmicx}   
\usepackage{floatflt}
\usepackage{graphics}
\usepackage{amsthm}

\begin{document}

\begin{titlepage}
\title{Vital Edges for $(s,t)$-mincut: Efficient Algorithms, Compact Structures, and Optimal Sensitivity Oracle}

\renewcommand\Affilfont{\normalsize}

\author[1]{Surender Baswana}
\author[2]{Koustav Bhanja}

\affil[1,2]{Indian Institute of Technology Kanpur, India}
\affil[ 1]{\texttt{sbaswana@cse.iitk.ac.in}}
%\affil[2]{Indian Institute of Technology Kanpur, India}
\affil[ 2]{\texttt{kbhanja@cse.iitk.ac.in}}
\date{}

\maketitle

\begin{abstract}
 Let $G$ be a directed weighted graph on $n$ vertices and $m$ edges with designated source and sink vertices $s$ and $t$. An edge in $G$ is vital if its removal reduces the capacity of $(s,t)$-mincut. Since the seminal work of Ford and Fulkerson [CJM 1956], a long line of work has been done on computing the \textit{most vital edge} and \textit{all vital edges} of $G$. However, even after $60$ years, the existing results are for either undirected or unweighted graphs. We present the following result for \textit{directed weighted graphs} that also solves an open problem by Ausiello, Franciosa, Lari, and Ribichini [NETWORKS 2019].
 
\noindent
\textbf{1. Algorithmic Results:} There is an algorithm that computes all vital edges as well as the most vital edge of $G$ using ${\mathcal O}(n)$ maximum $(s,t)$-flow computations. \\%[-5pt]

\noindent
Vital edges play a crucial role in the design of \textit{sensitivity oracle} for $(s,t)$-mincut -- a compact data
structure for reporting $(s,t)$-mincut after insertion/failure of any edge. For directed graphs, the only existing sensitivity oracle is for unweighted graphs by Picard and Queyranne [MPS 1982]. We present the first and optimal sensitivity oracle for \textit{directed weighted graphs} as follows.

\noindent
\textbf{2.} \textbf{Sensitivity Oracles:} \textbf{(a)} There is an optimal ${\mathcal O}(n^2)$ space data structure that can report an $(s,t)$-mincut $C$ in ${\mathcal O}(|C|)$ time after the failure/insertion of any edge.  \\
\textbf{(b)} There is an ${\mathcal O}(n)$ space data structure that can report the capacity of $(s,t)$-mincut after failure or insertion of any edge $e$ in ${\mathcal O}(1)$ time if the capacity of edge $e$ is known. \\%[-5pt]

\noindent
A \textit{mincut for a vital edge} $e$ is an $(s,t)$-cut of the least capacity in which edge $e$ is outgoing. For unweighted graphs, in a classical work, Picard and Queyranne [MPS 1982] designed an ${\mathcal O}(m)$ space directed acyclic graph (DAG) that stores and characterizes all mincuts for all vital edges. Conversely, there is a set containing at most $n-1$ $(s,t)$-cuts such that at least one mincut for every vital edge belongs to the set. We generalize these results for \textit{directed weighted graphs} as follows.

\noindent
\textbf{3.} \textbf{Structural \& Combinatorial Results: (a)} There is a set ${\mathcal M}$ containing at most $n-1$ $(s,t)$-cuts such that at least one mincut for every vital edge belongs to the set. This bound is tight as well. We also show that set ${\mathcal M}$ can be computed using ${\mathcal O}(n)$ maximum $(s,t)$-flow computations.\\
\textbf{(b)} We design two compact structures for storing and characterizing all mincuts for all vital edges -- (i) an ${\mathcal O}(m)$ space DAG for \textit{partial} and (ii) an ${\mathcal O}(mn)$ space structure for \textit{complete} characterization.\\%[-5pt]

\noindent
To arrive at our results, we develop new techniques, especially a generalization of \textit{maxflow-mincut} Theorem by Ford and Fulkerson [CJM 1956], which might be of independent interest.
\end{abstract}

\end{titlepage}

%%%%%%%%%%%%%%%%%SECTION 1%%%%%%%%%%%%%%%%%%%%%%%%%%%%%
\section{Introduction} \label{sec : introduction}
For any graph problem, there are edges whose removal affects the solution of the given problem. These edges are called {\em vital} edges for the problem. There has been extensive research on the vital edges for various fundamental problems -- shortest paths/distance \cite{roditty2012replacement, williams2022algorithms, hershberger2002erratum, nardelli2001faster, DBLP:journals/talg/0001W20}, minimum spanning trees \cite{frederickson1999increasing, DBLP:journals/ipl/LinC93, tsen1994finding}, strongly connected components (SCC) \cite{georgiadis2020strong}. The concept of vital edge for $(s,t)$-mincut has existed ever since the seminal work of Ford and Fulkerson \cite{ford_fulkerson_1956} on maximum $(s,t)$-flow. In this article, for directed weighted graphs, we present the following two main results --~(1) an efficient \textit{algorithm} for computing all vital edges,~
and (2) optimal {\em sensitivity oracles} for $(s,t)$-mincut. 
The algorithm in (1) is the first nontrivial algorithm for computing all vital edges in directed weighted graphs, which also answers an open question in \cite{DBLP:journals/networks/AusielloFLR19}. Our optimal sensitivity oracle in (2) is the first sensitivity oracle for directed weighted graphs in the area of minimum cuts. 
In order to arrive at these results, we present
interesting \textit{structural} \& optimal \textit{combinatorial} results on mincuts for vital edges. These results provide a generalization of the classical work of Picard and Queyranne \cite{DBLP:journals/mp/PicardQ80} and the recent work of Baswana, Bhanja, and Pandey \cite{baswana2023minimum+}. 

Let $G=(V,E)$ be a directed graph on $n=|V|$ vertices and $m=|E|$ edges. 
Each edge $e\in E$ has a capacity, denoted by $w(e)$, which is a positive real number. Let $s$ be a designated source vertex and $t$ be a designated sink vertex in $G$. 

A set $C\subset V$ is said to be a cut if $C\not= \emptyset$. The outgoing edges from $C$ are called \textit{contributing} edges of $C$. 
The \textit{capacity} of cut $C$, denoted by $c(C)$, is defined as the sum of the capacities of all contributing edges of $C$. 
A cut $C$ is said to be an $(s,t)$-cut if $s\in C$ and $t\in \overline{C}=V\setminus C$. 
An $(s,t)$-cut $C$ with the least capacity is called an \textit{$(s,t)$-mincut}. We denote the capacity of $(s,t)$-mincut by $f^*$. The $(s,t)$-mincut of a graph is a fundamental concept in graph theory and is used to design efficient algorithms for numerous real-world problems \cite{DBLP:books/daglib/0069809}. 
Now we formally define the vital edges for $(s,t)$-micuts.
\begin{definition}[vital edge]
An edge $e\in E$ is said to be a vital edge if the removal of $e$ decreases the capacity of $(s,t)$-mincut in $G$. $E_{vit}$ denotes the set of all vital edges in $G$.
\label{def:relevant-edges}
\end{definition}
 Observe that each edge that contributes to an $(s,t)$-mincut is definitely a vital edge. However, a vital edge might not necessarily contribute to any $(s,t)$-mincut (e.g., edge $(v_1,v_4)$ in Figure \ref{fig : each mincut has a nonvital edge}). An edge that is not vital is called a \textit{nonvital} edge. At first glance, it may appear that we may remove all nonvital edges from the graph without affecting the $(s,t)$-mincut. But it is not true, as stated in the following note.
\begin{note} \label{note : removal of gamma edges}
    Although the removal of any single nonvital edge does not decrease the capacity of $(s,t)$-mincut, the removal of a set of nonvital edges might lead to the reduction in the capacity of $(s,t)$-mincut (e.g., edges $(v_2,v_6)$ and $(v_3,v_5)$ in Figure \ref{fig : each mincut has a nonvital edge}). 
\end{note}
\begin{figure}[ht]
  \begin{center}
    \includegraphics[width=0.4\textwidth]{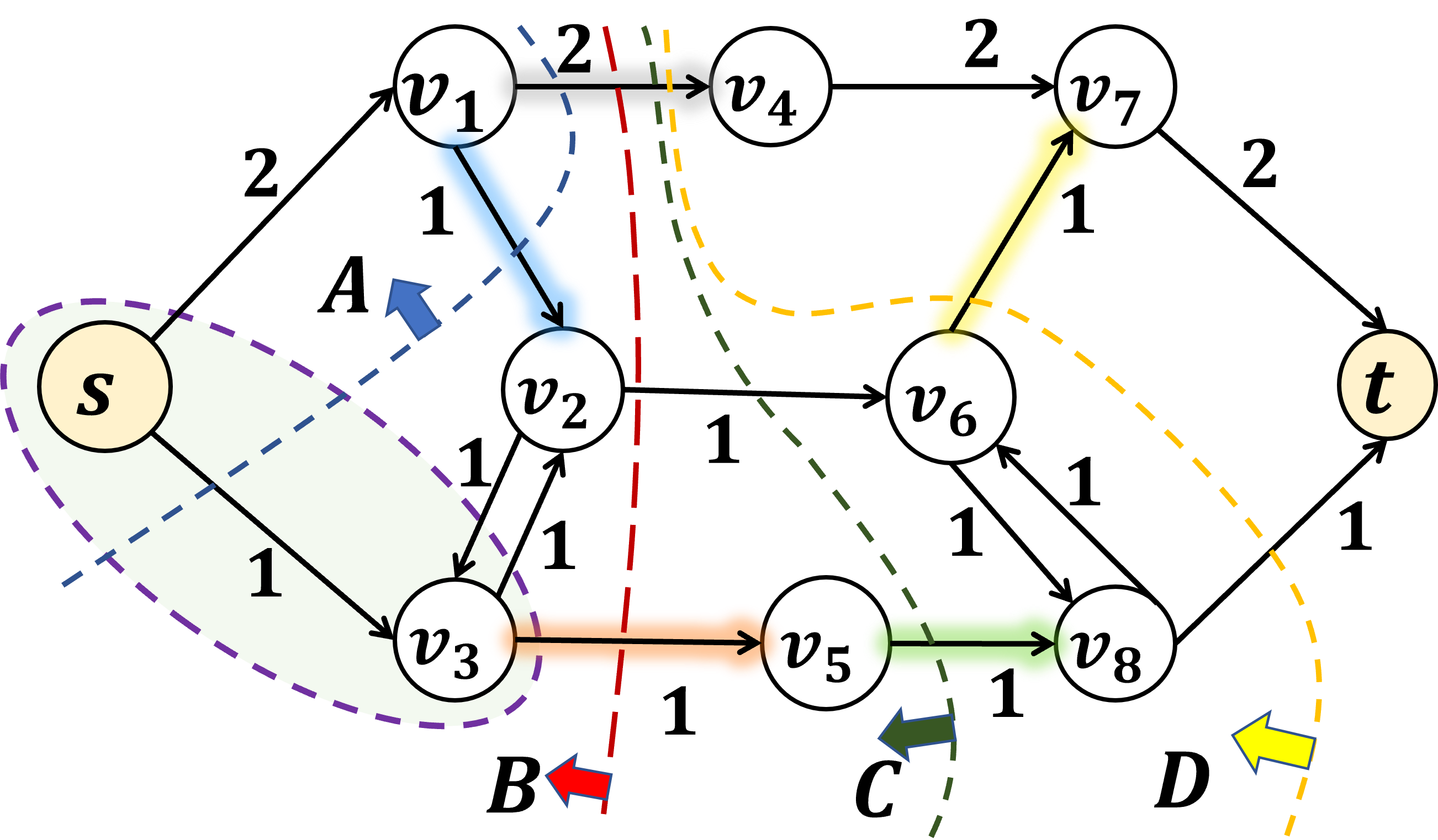}
  \end{center}
  \caption{Each mincut for vital edge $(v_1,v_4)$ contains a nonvital edge (shown in the same color).} 
  \label{fig : each mincut has a nonvital edge}
\end{figure}
The design of efficient algorithms for various problems \cite{ratliff1975finding, wollmer1963some, lubore1971determining, phillips1993network} related to vital edges started just a few years after the seminal work of Ford and Fulkerson \cite{ford_fulkerson_1956}. The \textit{vitality of an edge} $e$ is the reduction in the capacity of $(s,t)$-mincut after the removal of edge $e$. Observe that one maximum $(s,t)$-flow computation is sufficient for computing the vitality of any edge. So, for computing all vital edges and their vitality, there is a trivial algorithm that requires ${\mathcal O}(m)$ maximum $(s,t)$-flow computations. 
An edge having the maximum vitality is said to be the \textit{most vital edge}. 
Aneja, Chandrasekaran, and Nair 
\cite{DBLP:journals/networks/AnejaCN01} designed an algorithm that performs ${\mathcal O}(n)$ maximum $(s,t)$-flow computations to compute the most vital edge in an undirected graph.  
Ausiello, Franciosa, Lari, and Ribichini 
\cite{DBLP:journals/networks/AusielloFLR19}  showed that ${\mathcal O}(n)$ maximum $(s,t)$-flow computations are sufficient even for computing all vital edges and their vitality in an undirected graph. 
 Unfortunately, even after 60 years, the following question has remained unanswered, which is also posed as an open problem in \cite{DBLP:journals/networks/AusielloFLR19}.
\begin{question} \label{ques : algorithm}
     For directed weighted graph $G$, does there exist an algorithm that can compute all vital edges along with their vitality using ${\mathcal O}(n)$ maximum $(s,t)$-flow computations? 
\end{question}   
We can generalize the notion of $(s,t)$-mincut to {\em `mincut for an edge'} as follows.
\begin{definition} [mincut for an edge] \label{def : relevant and mincut for an edge}
    Let $e$ be a contributing edge of an $(s,t)$-cut $C$. $C$ is a mincut for edge $e$ if $c(C)\le c(C')$ for each $(s,t)$-cut $C'$ in which edge $e$ appears as a contributing edge. 
\end{definition}
Not only the study of mincuts for vital edges is important from a graph theoretic perspective but it also plays a crucial role in designing sensitivity oracle for $(s,t)$-mincuts -- a compact data structure that efficiently reports an $(s,t)$-mincut after the failure/insertion of any edge.

 The minimum cuts in a graph can be quite large in number -- $\Omega(n^2)$ global mincuts \cite{dinitz1976structure}, $\Omega(2^n)$ $(s,t)$-mincuts \cite{DBLP:journals/mp/PicardQ80}. 
Interestingly, compact structures have been invented that compactly store and {\em characterize} various types of minimum cuts and cuts of capacity near minimum \cite{dinitz2000general, dinitz19952, 
 DBLP:journals/mp/PicardQ80, dinitz1976structure, baswana2023minimum+}. A compact structure $G'$ is said to characterize a set ${\mathcal C}$ of cuts using a property ${\mathcal P}$ if the following holds. A cut $C$ belongs to ${\mathcal C}$ if and only if cut $C$ satisfies property ${\mathcal P}$ in graph $G'$. 
 Specifically for $(s,t)$-mincuts in any directed weighted graph $G$, 
Picard and Queyranne \cite{DBLP:journals/mp/PicardQ80} showed that there exists a directed acyclic graph (DAG), denoted by ${\mathcal D}_{PQ}(G)$, occupying ${\mathcal O}(m)$ space that compactly stores all $(s,t)$-mincuts. In addition, it provides the following characterization for each $(s,t)$-mincut. \\ 
\centerline{\em An $(s,t)$-cut $C$ is an $(s,t)$-mincut in $G$ if and only if $C$ is a $1$-transversal cut in ${\mathcal D}_{PQ}(G)$.}\\
An $(s,t)$-cut is said to be $1$-transversal if its edges intersect any simple path at most {\em once}; e.g., cut $A$ is $1$-transversal but cut $\{s,v_3\}$ is not $1$-transversal in Figure \ref{fig : each mincut has a nonvital edge}. The $1$-transversality property also plays a crucial role in designing sensitivity oracles for $(s,t)$-mincuts in unweighted graphs \cite{DBLP:journals/mp/PicardQ80, baswana2023minimum+}. Observe that ${\mathcal D}_{PQ}(G)$ stores all mincuts for all edges that contribute to $(s,t)$-mincuts. 
However, edges contributing to $(s,t)$-mincuts may constitute a small subset of the set of all vital edges.
Therefore,  ${\mathcal D}_{PQ}(G)$ may fail to preserve all mincuts for all vital edges. This raises the following question.

\begin{question} \label{question 0}
    Does there exist a compact graph structure that stores and characterizes all mincuts for all vital edges in directed weighted graph $G$? 
\end{question}
 Constructing the smallest set storing a minimum cut for every edge or every pair of vertices has been addressed extensively \cite{hassin2007flow, dinitz1976structure, DBLP:journals/anor/ChengH91, GH61, DBLP:journals/dam/GranotH86, gusfield1991efficient} since the remarkable work of Gomory and Hu \cite{GH61}. 
The smallest set of $(s,t)$-cuts that has at least one mincut for every edge is called a \textit{mincut cover}. 
 For directed unweighted graphs, Baswana, Bhanja, and Pandey \cite{baswana2023minimum+}, exploiting the DAG structure in \cite{DBLP:journals/mp/PicardQ80}, showed that there is a mincut cover of cardinality at most $n-1$ for all vital edges. 
 For undirected weighted graphs, 
 it is shown in \cite{DBLP:journals/networks/AusielloFLR19} that there is a mincut cover of cardinality at most $n-1$ for all edges, and hence for all vital edges. 
 Unfortunately, for directed weighted graphs, we establish that it is not possible to have a mincut cover of cardinality $o(n^2)$ for all edges (Theorem \ref{thm : su vt lower bound}). 
Therefore, the following question naturally arises. 
\begin{question}   \label{question 1}
    Does there exist a mincut cover of cardinality $o(n^2)$ for all vital edges in directed weighted graph $G$?
\end{question}
Design of sensitivity oracles have been studied quite extensively for various fundamental problems in both unweighted and weighted graphs --  shortest paths/distances \cite{bilo2023approximate, DBLP:journals/talg/0001W20}, 
reachability \cite{italiano2021planar, choudhary2016optimal}, 
strongly connected components \cite{baswana2019efficient, georgiadis2020strong}, all-pairs mincuts \cite{baswana2022sensitivity}. In
weighted graphs, the concept of failures/insertions of edges is, interestingly, more generic. Here, the aim is to report the solution to the given problem given that the capacities of a \textit{small} set of edges are decreased/increased by an amount $\Delta>0$. 
 
 For $(s,t)$-mincuts in directed graphs, the existing sensitivity oracles are only for unweighted graphs and completely based on DAG ${\mathcal D}_{PQ}(G)$ \cite{DBLP:journals/mp/PicardQ80}, which dates back to 1982. Firstly, there is an ${\mathcal O}(m)$ space sensitivity oracle given in
 \cite{DBLP:journals/mp/PicardQ80}. After the failure/insertion of an edge, the oracle takes ${\mathcal O}(1)$  time for reporting the capacity and ${\mathcal O}(m)$ time for reporting the corresponding $(s,t)$-mincut. 
Assuming the edge of the query exists in the graph, an ${\mathcal O}(n)$ space data structure can be designed \cite{baswana2023minimum+}.
This data structure can report the capacity of $(s,t)$-mincut in ${\mathcal O}(1)$ time and an $(s,t)$-mincut $C$ in ${\mathcal O}(|C|)$ time after the failure/insertion of an edge.
For various problems related to sensitivity analysis, 
 it is also important to efficiently report a compact structure that stores and characterizes \textit{all} $(s,t)$-mincuts after the failure of any edge, as shown in \cite{DBLP:journals/mp/PicardQ80}. For unweighted graphs, DAG ${\mathcal D}_{PQ}$ for the resulting graph can be reported in ${\mathcal O}(m)$ time \cite{DBLP:journals/mp/PicardQ80}. All these results crucially exploit the property that, in unweighted graphs, a mincut for a vital edge is also an $(s,t)$-mincut; hence, every contributing edge is also vital. However, for weighted graphs, this property no longer holds. In fact, nonvital edges may contribute to all mincuts for a vital edge (e.g., all mincuts $\{A,B,C,D\}$ for edge $(v_1,v_4)$ in Figure \ref{fig : each mincut has a nonvital edge}). So, the following question arises. 

\begin{question}\label{question 3}
    For directed weighted graph $G$, does there exist 
    \begin{enumerate}
        \item an ${\mathcal O}(n)$ space data structure that can report the capacity of $(s,t)$-mincut in ${\mathcal O}(1)$ time,
        \item a compact data structure that can report an $(s,t)$-mincut $C$ in ${\mathcal O}(|C|)$ time, and 
        \item a compact data structure that can report DAG ${\mathcal D}_{PQ}$  in ${\mathcal O}(m)$ time
    \end{enumerate}
        % 1. an ${\mathcal O}(n)$ space data structure that can report the capacity of $(s,t)$-mincut in ${\mathcal O}(1)$ time, \\
        % 2. a compact data structure that can report an $(s,t)$-mincut $C$ in ${\mathcal O}(|C|)$ time, and \\
        % 3. a compact data structure that can report DAG ${\mathcal D}_{PQ}$  in ${\mathcal O}(m)$ time \\ 
for the resulting graph after increasing/decreasing the capacity of any given edge? 
\end{question}

\begin{note} It is a simple exercise, using the result of \cite{DBLP:journals/mp/PicardQ80}, to design an ${\mathcal O}(n^2)$ space data structure and an
${\mathcal O}(n)$ space data structure that, after insertion of any given edge, can report an $(s,t)$-mincut $C$ in ${\mathcal O}(|C|)$ time and its capacity in ${\mathcal O}(1)$ time, respectively. Henceforth, while addressing sensitivity oracles, we focus only on handling the decrease in the capacity of an edge. 
\end{note}

\subsection{Our results}
We provide an affirmative answer to all the four questions raised above. To arrive at our results, as one of our key technical contributions, we present a generalization of the well-known maxflow-mincut theorem \cite{ford_fulkerson_1956} (stated in Theorem \ref{thm : a special assignment of flow}), which might be of independent interest.
\paragraph*{A. Mincut Cover}
Our first result, in the following theorem, provides a bound on the cardinality of mincut cover for all vital edges, %present our first result in the following theorem
which answers Question \ref{question 1} in the affirmative. In addition, the insights used to establish this result act as the foundation for our sensitivity oracles. 
\begin{theorem}[Mincut Cover] \label{thm : n-1 cuts} 
    For any directed weighted graph $G$ on $n$ vertices with a designated source vertex $s$ and a designated sink vertex $t$, there exists a set ${\cal C}_{min}$ containing at most $n-1$ $(s,t)$-cuts such that, for any vital edge $e$ in $G$, at least one mincut for edge $e$ is present in set ${\cal C}_{min}$.
\end{theorem}

  Let $H$ be a directed path from $s$ to $t$ consisting of $n-1$ edges. Observe that each edge $e$ in $H$ is a vital edge and there is a unique mincut for $e$.  
  Hence, the bound of $n-1$ on the cardinality of mincut cover, as mentioned in Theorem \ref{thm : n-1 cuts}, is tight. In addition, our bound of $n-1$ cuts also matches with the existing best-known bounds for both directed unweighted \cite{DBLP:journals/mp/PicardQ80, baswana2023minimum+} and undirected weighted graphs \cite{DBLP:journals/networks/AusielloFLR19}.

\paragraph*{B. Sensitivity Oracles for $(s,t)$-mincut}
We present a pair of data structures that act as sensitivity oracles for $(s,t)$-mincut. Our first data structure is stated in the following theorem, %is used to report both $(s,t)$-mincuts and its capacity, %in the worst case optimal time and can also (implicitly) report all $(s,t)$-mincuts achieving a nontrivial query time. This 
 which answers Question \ref{question 3}(2,3) in the affirmative. 
\begin{theorem} [Sensitivity Oracle] \label{thm : main result}
     Let $G$ be a directed weighted graph on $n$ vertices with a designated source vertex $s$ and a designated sink vertex $t$. There is an ${\mathcal O}(n^2)$ space data structure that, given any edge $e$ and any value $\Delta\ge 0$, can report 
     \begin{enumerate}
     \item the capacity of $(s,t)$-mincut in ${\mathcal O}(1)$ time, 
     \item an $(s,t)$-mincut $C$ in ${\mathcal O}(|C|)$ time for the resulting graph, and 
     \item the DAG ${\mathcal D}_{PQ}$ occupying ${\mathcal O}(m)$ space in ${\mathcal O}(m)$ time
     \end{enumerate}
     after reducing the capacity of the edge $e$ by $\Delta$.
\end{theorem}
% \begin{theorem} [Sensitivity Oracle] \label{thm : main result}
%      Let $G$ be a directed weighted graph on $n$ vertices with a designated source vertex $s$ and a designated sink vertex $t$. There is an ${\mathcal O}(n^2)$ space data structure that, given any edge $e$ and any  value $\Delta$ (positive/negative), can report (1) the capacity of $(s,t)$-mincut in ${\mathcal O}(1)$ time, (2) an $(s,t)$-mincut $C$ in ${\mathcal O}(|C|)$ time, and (3) the DAG ${\mathcal D}_{PQ}$ occupying ${\mathcal O}(m)$ space in ${\mathcal O}(m)$ time for the resulting graph after changing the capacity of the edge $e$ from $w(e)$ to $w(e)+\Delta$.
% \end{theorem}
Suppose the aim is to report only the capacity of $(s,t)$-mincut. Interestingly, in this case, we design a significantly compact data structure, which answers Question \ref{question 3}(1) in the affirmative. 
\begin{theorem} [Sensitivity Oracle for Reporting Capacity] \label{thm : reporting value}
   Let $G$ be a directed weighted graph on $n$ vertices with a designated source vertex $s$ and a designated sink vertex $t$. There is an ${\mathcal O}(n)$ space data structure that, given any edge $e\in E$ and a value $\Delta$ satisfying 
   $0\le \Delta\le w(e)$, can report in ${\mathcal O}(1)$ time the capacity of $(s,t)$-mincut after reducing the capacity of $e$ by $\Delta$.  
\end{theorem}
% \begin{note} \label{note : vitality of an edge}
%     Given any edge $e$ and its capacity $w(e)$, the data structure in Theorem \ref{thm : reporting value} can be used to report vitality of $e$ in ${\mathcal O}(1)$ time by assigning $\Delta=-w(e)$. This leads to Theorem \ref{thm : data structure for vitality}.
% \end{note}
\noindent
\textbf{Lower Bound:} We complement the result in Theorem \ref{thm : main result} by a matching lower bound, which holds even for undirected graphs as follows.
% \begin{theorem} \label{thm : lower bound}
%     Let $D$ be any data structure for reporting the capacity of $(s,t)$-mincut after the failure of an edge in an (un)directed graph on $n$ vertices with positive edge capacities must require $\Omega (n^2\log~n)$ bits of space in the worst case, irrespective of the query time, for a designated source vertex $s$ and a designated sink vertex $t$.
% \end{theorem}
\begin{theorem} \label{thm : lower bound}
    Any data structure for reporting the capacity of $(s,t)$-mincut after the failure of an edge in an (un)directed graph on $n$ vertices with positive edge capacities must require $\Omega (n^2\log{n})$ bits of space in the worst case, irrespective of the query time, for a designated source vertex $s$ and a designated sink vertex $t$.
\end{theorem}
\begin{remark}
The ${\mathcal O}(n)$ space upper bound in Theorem \ref{thm : reporting value} does not violate the $\Omega(n^2)$ space lower bound in Theorem \ref{thm : lower bound}. This is because Theorem \ref{thm : reporting value} assumes that the query edge $e$ belongs to $E$ and 
the reduction in capacity of $e$ is at most $w(e)$. 
However, this assumption seems practically justified since, in real world, the capacity of an edge can decrease only if the edge {\em actually} exists, and furthermore, it can decrease by an amount at most the capacity of the edge. 
\end{remark}
% \paragraph*{Labeling Scheme} Let $f$ be any function defined on the vertex/edge set of the graph. A labeling scheme of $f$
%  assigns small labels to each vertex/edge of the graph in such a way that, for any given subset
%  $A$ of edges/vertices, $f(A)$ can be computed only by using the labels of vertices/edges in $A$. 
\noindent
 We show that, using data structure in Theorem \ref{thm : reporting value}, a labeling scheme of ${\mathcal O}(\log^2{n}+\log{n}\log{W})$ bits per vertex can be designed for reporting capacity of $(s,t)$-mincut after decreasing capacity of any edge $e$ in $G$ by a value $\Delta$ satisfying $0\le \Delta \le w(e)$.  Here $W$ is the maximum capacity of any edge in $G$.
 %using Theorem \ref{thm : reporting value} and a labeling scheme for \textsc{lca} query \cite{DBLP:journals/siamcomp/KatzKKP04}. Here $W$ is the maximum
% capacity of any edge in $G$. 
% \begin{remark}
% The ${\mathcal O}(n)$ space upper bound in Theorem \ref{thm : reporting value} does not violate the $\Omega(n^2)$ space lower bound in Theorem \ref{thm : lower bound}. This is because Theorem \ref{thm : reporting value} assumes that the query edge $e$ belongs to $E$ and the change in capacity ($\Delta$) provided with the query is at least $-w(e)$. 
% However, this assumption seems practically justified since, in real world, the capacity of an edge can decrease only if the edge {\em actually} exists, and furthermore, it can reduce by an amount at most the capacity of the edge. 
% \end{remark}
% \noindent
% A \textit{labeling scheme} of ${\mathcal O}(\log^2~n+\log~n\log~W)$ bits for reporting capacity of $(s,t)$-mincut after decreasing capacity of any edge is designed in Appendix \ref{sec : sensitivity oracle}.
\paragraph*{C. Algorithm for Computing All Vital Edges}  
Our algorithmic result on computing all vital edges is stated in the following theorem, which answers Question \ref{ques : algorithm} in the affirmative. 
\begin{theorem}[Computing All Vital Edges] \label{thm : computing all vital edges}
    For any directed weighted graph on $n$ vertices with a designated source vertex $s$ and designated sink vertex $t$, 
    there is an algorithm that computes all vital edges and their vitality using ${\mathcal O}(n)$ maximum $(s,t)$-flow computations. % to compute all vital edges along with their vitality.  
\end{theorem}
Note that the running time of our algorithm also matches with the best-known result for undirected weighted graphs \cite{DBLP:journals/networks/AusielloFLR19}.
\begin{note}
    Let $G\setminus e$ be the graph obtained from $G$ after the removal of edge $e$. Computing an $(s,t)$-mincut in graph $G\setminus e$ for every edge $e$ can be accomplished easily using $m$ maximum $(s,t)$-flow computations. However, we show, using Theorem \ref{thm : computing all vital edges}, that ${\mathcal O}(n)$ maximum $(s,t)$-flow computations are sufficient. 
    %(refer to Note \ref{note : mincut in G minus e}). 
\end{note}
 
 % Not only the classification of vital edges plays crucial role in their efficient computation, but it is also used in designing a compact structure for all vital edges (Section 1.1.3). We believe that this classification might find other applications. Hence, it is also interesting to efficiently determine, given any edge, whether it is tight or loose. Given the value of maximum $(s,t)$-flow $f^*$, we show that just one maximum $(s,t)$-flow computation on graph $G$ after {\em small} modification is sufficient for this task (Theorem \ref{thm : maxflow at an edge}). 
 %we provide an algorithm that, given any vital edge, can determine whether the edge is a tight or loose using one maximum $(s,t)$-flow computation (refer to Theorem \ref{thm : maxflow at an edge}). 

Observe that there is a trivial data structure of ${\mathcal O}(m)$ space 
that can report the vitality of any given edge in ${\mathcal O}(1)$ time. Can we have a data structure that is even more compact and still achieves ${\mathcal O}(1)$ query time? The following theorem answers this question in the affirmative %Refer to Section 1.1.4 for the overview of the construction. %{\color{red}
(Refer to Note \ref{note : vitality of an edge}).

\begin{theorem} \label{thm : data structure for vitality}
    For any directed weighted graph $G$ on $n$ vertices, there is an ${\mathcal O}(n)$ space data structure that, given any edge $e$ in $G$ and its capacity, can report the vitality of edge $e$ in ${\mathcal O}(1)$ time. 
\end{theorem}

We show, using Theorem \ref{thm : computing all vital edges}, that the data structure in Theorem \ref{thm : data structure for vitality} can be built using ${\mathcal O}(n)$ maximum $(s,t)$-flow computations. As a byproduct, this computes a mincut cover for all vital edges. %(refer to Theorem \ref{thm : t vit G computation}).

% Not only the classification of vital edges plays crucial role in their efficient computation, but it is also used in designing a compact structure for all vital edges (Section 1.1.3). We believe that this classification might find other applications. Hence, it is also interesting to efficiently determine, given any edge, whether it is tight or loose. Given the value of maximum $(s,t)$-flow $f^*$, we show that just one maximum $(s,t)$-flow computation on graph $G$ after {\em small} modification is sufficient for this task (Theorem \ref{thm : maxflow at an edge}). 

\paragraph*{D. Compact Structures for All Mincuts for All Vital Edges}
We present a pair of compact structures for storing and characterizing all mincuts for all vital edges -- one is a single DAG that provides a \textit{partial} characterization and the other consists of a set of ${\mathcal O}(n)$ DAGs that provides a complete characterization. These two results gives an affirmative answer to Question \ref{question 0}.
\begin{theorem}[Partial Characterization] \label{thm : dag and 1 transversal}
    For a directed weighted graph $G$ on $m$ edges with a designated source vertex $s$ and a designated sink vertex $t$, there is an ${\mathcal O}(m)$ space directed acyclic graph ${\mathcal D}_{vit}(G)$ that preserves the capacity of $(s,t)$-mincut, and for each vital edge $e$ in $G$, 
    \begin{enumerate}
        \item  Every mincut for edge $e$ in $G$ is a $1$-transversal cut in ${\mathcal D}_{vit}(G)$ and
        \item  A mincut $C$ for $e$ in $G$ appears as a  \textit{relevant cut} in ${\mathcal D}_{vit}(G)$, that is, $c(C)-w(e)<f^*$ in ${\mathcal D}_{vit}(G)$.  
    \end{enumerate}
\end{theorem}
% \begin{note}
    Existing DAG structures for $(s,t)$-mincuts \cite{DBLP:journals/mp/PicardQ80, baswana2023minimum+} turn out to be just a special case of DAG ${\mathcal D}_{vit}(G)$ in Theorem \ref{thm : dag and 1 transversal}. Moreover, ${\mathcal D}_{vit}(G)$ additionally guarantees the property (2) in Theorem \ref{thm : dag and 1 transversal}, which does not hold in the existing DAGs \cite{DBLP:journals/mp/PicardQ80, baswana2023minimum+}.
% \end{note}
\begin{theorem}[Complete Characterization] \label{thm : vital complete characterization}
    For any directed weighted graph $G$ on $n$ vertices and $m$ edges with a designated source vertex $s$ and a designated sink vertex $t$, there is an ${\mathcal O}(mn)$ space structure ${\mathcal S}_{vit}(G)$ consisting of ${\mathcal O}(n)$ DAGs such that for any vital edge $e$,
    \begin{enumerate}
        \item There exists at most two DAGs, ${\mathcal D}^s(e)$ and ${\mathcal D}^t(e)$, that store all mincuts for edge $e$.
        \item An $(s,t)$-cut $C$ in $G$ is a mincut for edge $e$ if and only if $C$ is a $1$-transversal cut in either ${\mathcal D}^s(e)$ or ${\mathcal D}^t(e)$  and $e$ is a contributing edge of $C$.
    \end{enumerate}
\end{theorem}
 An important application of the structure in Theorem \ref{thm : vital complete characterization} is the construction of sensitivity oracle for $(s,t)$-mincut stated in Theorem \ref{thm : main result}(3). In addition, it provides a complete characterization for all $(s,t)$-cuts of the least capacity separating pair of vertices $u,v$, for every pair  $(u,v)\in V\times V$. 
 
 All the results of this manuscript are compactly represented in Table \ref{tab : Algorithmic results}, Table \ref{tab : Sensitivity Oracle}, and Table \ref{tab : structural results}.

\begin{table}[H]
    \centering
    \small
    \begin{tabular}{|c|c|c|c|}
        \hline
        \textbf{Algorithms for} & Unweighted & Weighted & {\color{blue}\textbf{Weighted}}\\
        \textbf{Vital Edges} & Directed \cite{DBLP:journals/mp/PicardQ80} & Undirected \cite{DBLP:journals/anor/ChengH91, DBLP:journals/networks/AusielloFLR19, DBLP:journals/networks/AnejaCN01} & {\color{blue}\textbf{Directed (New)}} \\
        \hline 
        & & &\\
        \textsc{All Vital Edges} & ${\mathcal O}(T_{\textsc{mf}})$ & ${\mathcal O}(nT_{\textsc{mf}})$ & {\color{blue}$\mathbf{{\mathcal O}(nT_{\textsc{mf}})}$}\\
        %& & &\\
        \hline
        & & &\\
        \textsc{Most Vital edge} & ${\mathcal O}(T_{\textsc{mf}})$ & ${\mathcal O}(nT_{\textsc{mf}})$ & {\color{blue}$\mathbf{{\mathcal O}(nT_{\textsc{mf}})}$}\\
        \hline
    \end{tabular}    
    \caption{ This table represents a comparison between existing and new algorithms for computing vital edges. Here, $T_{\textsc{mf}}$ denotes the time taken for computing a maximum $(s,t)$-flow.}
    \label{tab : Algorithmic results}
\end{table}

\begin{table}[H]
    \centering
    \small
    \begin{tabular}{|c|c|c|c|c|c|}
        \hline
         \textbf{Sensitivity Oracles} & Reporting & Reporting & Lower & Reporting & PreProcessing \\
         \textbf{for} & Capacity & Cut & Bound & ${\mathcal D}_{PQ}(G_{e,\Delta})$ & Time \\
         \textbf{$(s,t)$-mincuts} & (Space, Time) & (Space, Time) & (Space) & (Space, Time) & \\
         \hline
         Unweighted &   &   &  & & \\ 
         Directed \cite{DBLP:journals/mp/PicardQ80, baswana2023minimum+} & ${\mathcal O}(n)$, ${\mathcal O}(1)$ & ${\mathcal O}(n)$, ${\mathcal O}(n)$ & ${{\Omega}(m)}$ & ${\mathcal O}(m)$, ${\mathcal O}(m)$ & ${\mathcal O(T_{\textsc{mf}})}$  \\
         \hline
        Weighted &  &  &  Not & Not & \\
         Undirected \cite{DBLP:journals/anor/ChengH91} & ${\mathcal O}(m)$, ${\mathcal O}(1)$ &  ${\mathcal O}(n^2)$, ${\mathcal O}(n)$ & Addressed & Addressed &${\mathcal O(nT_{\textsc{mf}})}$ \\
         \hline
         {\color{blue}\textbf{Weighted}} &   &   &  & & \\ 
        \textbf{ {\color{blue}Directed (New)}} & {\color{blue}$\mathbf{{\mathcal O}(n)}$}, {\color{blue}$\mathbf{{\mathcal O}(1)}$ }&{\color{blue} $\mathbf{{\mathcal O}(n^2)}$},{\color{blue} $\mathbf{{\mathcal O}(n)}$} & {\color{blue}  ${\mathbf{{\Omega}(n^2)}}$} &{\color{blue} $\mathbf{{\mathcal O}(n^2),~{\mathcal O}(m)}$ }& ${\color{blue}\mathbf{{\mathcal O(nT_{\textsc{mf}})}}}$\\ 
         \hline
    \end{tabular}    
    \caption{ ${\mathcal O}(n)$ space data structures for reporting capacity assume that query edge $e$ is present in $G$ and, for weighted graphs, change in capacity $\Delta$, satisfying $0\le \Delta \le w(e)$, is provided with the query. Here $G_{e,\Delta}$ denote the graph obtained from $G$ after increasing/decreasing the capacity of edge $e$ by $\Delta$. $T_{\textsc{mf}}$ denotes the time taken for computing a maximum $(s,t)$-flow. The preprocessing time is for the data structures for reporting capacity (Theorem \ref{thm : reporting value}) and reporting cut (Theorem \ref{thm : main result}).}

    \label{tab : Sensitivity Oracle}
\end{table}

\begin{table}[H]
    \centering
    \small
    \begin{tabular}{|c|c|c|c|c|}
        \hline
         \textbf{Structural \&  Com-} & Structure & Space & Characterization & Mincut \\
         \textbf{binatorial Results} &  &  &  & Cover \\
         %\textbf{$(s,t)$-mincuts} & (Space, Time) & (Space, Time) & (Space, Time) & (Space, Time) \\
         \hline
          Unweighted   &  &  &  &  \\
         Directed \cite{DBLP:journals/mp/PicardQ80, baswana2023minimum+} & Single DAG &  ${\mathcal O}(m)$ & Complete: $1$-transversal & $ n-1$ \\
        
         \hline
        Weighted &  Not &  Not &  Not & \\
         Undirected \cite{DBLP:journals/anor/ChengH91, DBLP:journals/networks/AusielloFLR19} & Addressed & Addressed & Addressed & $n-1$ \\
         \hline
        {\color{blue}\textbf{Weighted}} & {\color{blue}\textbf{Single DAG}} & {\color{blue}$\mathbf{{\mathcal O}(m)}$} & {\color{blue}\textbf{Partial: $1$-transversal}} &   \\
        {\color{blue}\textbf{ Directed (New)}} & {\color{blue}$\mathbf{{\mathcal O}(n)}$ \textbf{DAGs}} & {\color{blue}$\mathbf{{\mathcal O}(mn)}$ } & {\color{blue}\textbf{Complete: $1$-transversal}} & {\color{blue}$\mathbf{n-1}$}  \\
         \hline
    \end{tabular}    
    \caption{ This table represents a comparison between existing and new structural and combinatorial results for mincuts for vital edges.}
    \label{tab : structural results}
\end{table}

\subsection{Related works} \label{sec : related works}
We present, related to our results, the state-of-the-art algorithmic, structural, and combinatorial results on various minimum cuts of a graph as well as existing sensitivity oracles for them.

\paragraph*{Mincut Cover:} In undirected weighted graphs, \cite{GH61} showed that there is a set containing at most $n-1$ cuts such that at least one cut of the least capacity separating every pair of vertices is present in the set. 
% They arrived at this result by designing a tree, widely known as \textit{Gomory-Hu tree}, using $n-1$ maximum flow computations.  
The result of \cite{GH61} has been extended with the matching bounds for undirected graphs with vertex capacities \cite{DBLP:journals/dam/GranotH86, gusfield1991efficient}. There also exist similar results for vertex cuts in undirected graphs with both vertex and edge capacities \cite{hassin2007flow, DBLP:journals/dam/GranotH86, gusfield1991efficient}. 
%{\color{red} Among all the results mentioned above, only for Gomory-Hu tree, the recent result of \cite{DBLP:conf/focs/AbboudK0PST22} could break the cubic barrier and designed an $\tilde{{\mathcal O}}(n^2)$ time algorithm for computing the Gomory-Hu tree.}

%  For vertex and edge capacitated undirected graphs, \cite{DBLP:journals/dam/HassinL07} showed that there is a set containing at most $n-1$ vertex cuts such that a cut of the least capacity separating every pair of vertices is present in the set. Moreover, \cite{DBLP:journals/dam/HassinL07} designed an algorithm for computing this set using only ${\mathcal O}(n)$ cut computations.  
% For global mincuts in undirected weighted graphs, the result of \cite{DBLP:journals/anor/ChengH91} can be used to establish that there is a set of at most $n-1$ cuts such that a mincut for every vital edge is present in the set. Hence, for global mincut, using the algorithm of \cite{DBLP:journals/anor/ChengH91}, a mincut cover for all vital edges can be computed using ${\mathcal O}(n)$ maximum flow computations. Among all the results mentioned above, only for Gomory-Hu tree, the recent result of \cite{DBLP:conf/focs/AbboudK0PST22} could break the cubic barrier and designed an $\tilde{{\mathcal O}}(n^2)$ time algorithm for computing the Gomory-Hu tree.   

\paragraph*{Algorithmic Results for Vital Edges:} 
There exist algorithms related to vital edges for various graph problems, such as computing the most vital edge in matching \cite{hung1993most}, and in spanning trees \cite{tsen1994finding, DBLP:conf/soda/FredericksonS96, DBLP:journals/ipl/LinC93}, computing replacement paths in directed graphs between designated source vertex $s$ and designated sink vertex $t$ for every possible edge \cite{DBLP:journals/talg/0001W20} and every possible subset of edges of cardinality two or more \cite{williams2022algorithms}, computing total number of strongly connected components in $G\setminus \{e\}$ for every possible edge $e$ \cite{georgiadis2020strong}.  

\paragraph*{Cardinality of Cuts:} In undirected graphs, for any fixed $k$, the number of $k^{th}$ minimum cuts is  ${\mathcal O}(n^{3k-1})$ \cite{vazirani1992suboptimal}. However, in directed graphs, it can be exponential. In directed weighted graphs, there can be $\Omega(n^2)$ different capacities of all-pairs mincuts \cite{frank1971communication} and minimum vertex cuts.  In undirected weighted graphs, Hassin \cite{DBLP:journals/mor/Hassin88} showed that, for any $k\in [n]$, there can be $\binom{n-1}{k-1}$ different capacities of minimum multi-terminal $k$-cuts. %Hassin and Levin \cite{DBLP:journals/dam/HassinL07} showed that there exists directed weighted graphs with $\Omega(n^2)$ distinct capacities of minimum vertex cuts.

\paragraph*{Compact Structures for (Minimum) Cuts:} There exist elegant structures for storing and characterizing various minimum cuts of a graph, such as, cactus graph for all global mincuts \cite{dinitz1976structure}, \textit{connectivity carcass} for all Steiner mincuts \cite{dinitz2000general}. There also exist compact structures for storing and characterizing near minimum cuts, namely, 2-level cactus for all minimum+1 global cuts \cite{dinitz19952}, 2-level dag for minimum+1 $(s,t)$-cuts \cite{baswana2023minimum+}, cuts of capacity within $\frac{6}{5}\lambda_g$ \cite{DBLP:conf/focs/Benczur95}, within $\frac{3}{2}\lambda_g$ \cite{DBLP:journals/jacm/Karger00}, within $\frac{3}{2}\lambda_g$ in multigraphs \cite{DBLP:journals/jacm/KawarabayashiT19}, where $\lambda_g$ is the capacity of global mincut.

\paragraph*{Sensitivity Oracles:} For directed multi-graphs, there is an ${\mathcal O}(n^2)$ space dual edge Sensitivity Oracle for $(s,t)$-mincut \cite{baswana2023minimum+} that, upon failure/insertion of any pair of edges, can report the capacity of $(s,t)$-mincut in ${\mathcal O}(1)$ time and a resulting $(s,t)$-mincut in ${\mathcal O}(n)$ time. It is also established that the space occupied by the oracle is optimal if Reachability Hypothesis \cite{goldstein2017conditional} holds. For unweighted graphs, there also exists a sensitivity data structure that can handle a single edge insertion for single source mincuts \cite{DBLP:journals/algorithmica/BaswanaGK22}.

\subsection{Organization of the manuscript}
% We provide an affirmative answer to all the four questions raised above.
We present basic notations and terminologies in the following section. Thereafter, we present a generalization of \textit{maxflow-mincut} Theorem by Ford and Fulkerson \cite{ford_fulkerson_1956} in Section \ref{sec : flowcut}.
%This property, which might be of independent interest, plays a key role in establishing various results in this article.
Mincut cover for all vital edges and sensitivity oracles for $(s,t)$-mincut
are presented in Section \ref{sec: mincut-cover} and \ref{sec: sensitivity-oracles} respectively.
The algorithm for computing all vital edges is presented in Section \ref{sec: algorithm}. Compact structures for storing and characterizing all mincuts for all vital edges are described in Section \ref{sec: compact-structures}. An application of compact structures in sensitivity oracle is given in Section \ref{sec : application sensitivity oracle}. Finally, We present various lower bounds in Section \ref{sec : lower bounds}.

%%%%%%%%%%%%%%%%%%%%%%SECTION 2%%%%%%%%%%%%%%%%%%%%%%%%%%%%%%
\section{Preliminaries} \label{sec : basic preliminaries}
For any directed weighted graph $H$ with a designated source vertex $s$ and a designated sink vertex $t$, we define the following notations to be used throughout the article.
\begin{itemize}
\item $H\setminus\{e\}$: Graph after the failure of an edge $e$ in $H$.
\item value($f,H$): value of an $(s,t)$-flow $f$ in graph $H$. 
% {\color{red}We denote value($f,G$) by $f^*$.}
\item $f(e,H)$: value of $(s,t)$-flow $f$ along edge $e$ in $H$.  
\item $f_{in}(C,H)$: For any $(s,t)$-flow $f$ and an $(s,t)$-cut $C$ in $H$, $f_{in}(C,H)$ is the sum of the flow through all edges that leave $\overline{C}$ and enter $C$.
\item $f_{out}(C,H)$: For any $(s,t)$-flow $f$ and an $(s,t)$-cut $C$ in $H$, $f_{out}(C,H)$ is the sum of the flow through all edges that leave $C$ and enter $\overline{C}$.
\end{itemize}
Without causing any ambiguity, we use $f_{in}(C)$ and $f_{out}(C)$ to denote $f_{in}(C,G)$ and $f_{out}(C,G)$ respectively. 

The following lemma can be proved easily using the conservation of an $(s,t)$-flow.  
\begin{lemma} \label{lem : conservation of flow}
    For any $(s,t)$-cut $C$ in $H$,
    $f_{out}(C,H)-f_{in}(C,H)=\text{value}(f,H)$.
\end{lemma}
For a set $U\subseteq V$, $E_{vit}(U)$ denotes the set of all vital edges whose both endpoints belong to $U$. 
We now introduce a notation that quantitatively captures the {\em vitality} of an edge. 
\begin{itemize}
    \item $w_{min}(e)$: the reduction in the capacity of $(s,t)$-mincut in $G$ after the failure of edge $e$.
\end{itemize}
The following observation is immediate from Definition \ref{def:relevant-edges}.
\begin{observation} \label{obs: min-capacity-of-an-edge}
For any vital edge $e$, $w_{min}(e)>0$. 
\end{observation}
\begin{definition}[A mincut cover for a set of edges]
    A set ${\mathcal A}$ consisting of $(s,t)$-cuts is said to be a mincut cover for a set of edges $E'$ if, for each edge $e\in E'$, at least one mincut for $e$ is present in ${\mathcal A}$, and 
    $|{\mathcal A}|$ is the smallest. 
\end{definition}
The following lemma provides a maximum $(s,t)$-flow based characterization for a vital edge. % evident from the strong duality between maximum $(s,t)$-flow and $(s,t)$-mincut.
\begin{lemma} \label{lem : vital edge in every maximum flow}
    An edge $e$ is vital if and only if $f(e)>0$ in every maximum $(s,t)$-flow $f$ in $G$. 
\end{lemma}
Lemma \ref{lem : vital edge in every maximum flow} can be proved easily using the strong duality between maximum $(s,t)$-flow and $(s,t)$-mincut \cite{ford_fulkerson_1956}.
\begin{definition}[Edge-set of a cut]
    A cut $C$ is said to {\em separate} a pair of vertices $u,v\in V$ if either $u\in C$ and $v\in \overline{C}$ or $v\in C$ and $u\in \overline{C}$. The edge-set of $C$ is the set of all those edges whose endpoints are separated by $C$. 
    %An edge $e=(u,v)$ belongs to edge-set of $C$ if $C$ separates $u$ and $v$. 
    \label{def: edge-set}
\end{definition}
For an undirected graph, the set of contributing edges of a cut is the edge-set of the cut. In a seminal work, Gomory and Hu \cite{GH61} established the following two results for an undirected graph. (1)~There is a set of $n-1$ cuts such that the cut of the least capacity separating each pair of vertices is present in this set. (2)~Each of these $n-1$ cuts appears as a cut in a spanning tree on the vertex set $V$, which can be computed using ${\mathcal O}(n)$ maximum flow computations. This tree is widely-known as Gomory-Hu tree. 

The construction of Gomory-Hu tree crucially exploits that the capacity of each cut in graph is defined as the sum of the capacities of all contributing edges of the cut. 
Suppose capacity of each cut is defined by any arbitrary real-valued function $F$. Even for this generic setting, Cheng and Hu \cite{DBLP:journals/anor/ChengH91} showed that there exists a set consisting of $n-1$ cuts such that a cut of the least capacity ($F$-value) separating any pair of vertices is present in the set. Furthermore,
Cheng and Hu \cite{DBLP:journals/anor/ChengH91} showed that all these cuts can be stored in a suitably augmented rooted full binary tree, called {\em ancestor tree} as follows. Every vertex of the given graph is mapped to a unique leaf node in the ancestor tree, and the cut of the least capacity separating any pair of vertices is stored at their lowest common ancestor ({\textsc{lca}}) in the tree. The ancestor tree occupies ${\mathcal O}(n^2)$ space and, Cheng and Hu \cite{DBLP:journals/anor/ChengH91} also designed an efficient algorithm to construct the tree --
It performs only ${\mathcal O}(n)$ calls to an algorithm that, given any pair of vertices, can report a cut of the least capacity separating them.

Ausiello, Franciosa, Lari, and Ribichini \cite{DBLP:journals/networks/AusielloFLR19} showed that 
the ancestor tree of Cheng and Hu \cite{DBLP:journals/anor/ChengH91} can be constructed for $(s,t)$-cuts as well if function $F$ is defined as follows. 
\begin{equation} \label{eq : s,t cuts}
\text{For a set $C\subset V$,~}F(C)=\begin{cases}
            c(C), \text{ if $s\in C$ and $t\in \overline{C}$} \\
            \infty, \quad \text{otherwise.}
     \end{cases}
\end{equation}
This insight played a crucial role in the computation of all vital edges in an undirected graph using ${\cal O}(n)$ maximum $(s,t)$-flow computations \cite{DBLP:journals/networks/AusielloFLR19}.

%%%%%%%%%%%%%%%%%%%%%SECTION 3%%%%%%%%%%%%%%%%%%%%%%%%%%%%%%%
\section{A Generalization of \textsc{FlowCut} Property} \label{sec : flowcut}
Let $E_{min}\subseteq E$ be the set of all edges contributing to any $(s,t)$-mincut. Ford and Fulkerson \cite{ford_fulkerson_1956} established a strong duality between $(s,t)$-mincut and maximum $(s,t)$-flow. 
Exploiting this, the following property 
provides a maximum $(s,t)$-flow based characterization for mincut for any edge $e\in E_{min}$.
\newline

\noindent
\textsc{FlowCut}:~{\em Let $C$ be an $(s,t)$-cut and $e\in E_{min}$ is a contributing edge of $C$. $C$ is a mincut for edge $e$ if and only if for any maximum $(s,t)$-flow, each contributing edge of $C$ is fully saturated and each incoming edge of $C$ carries no flow.}
\newline

% \begin{itemize}
%     \item [] \textsc{FlowCut}:~{\em Let $C$ be an $(s,t)$-cut and $e\in E_{min}$ is a contributing edge of $C$. $C$ is a mincut for edge $e$ if and only if for any maximum $(s,t)$-flow, each contributing edge of $C$ is fully saturated and each incoming edge of $C$ carries no flow.}
% \end{itemize}

%
\noindent
For directed weighted graphs, we know that $E_{min}$ 
maybe only a proper subset of the set of all vital edges.
We now present a generalization of $\textsc{FlowCut}$ property that provides a maximum $(s,t)$-flow based characterization for each mincut $C(e)$, $\forall e \in E_{vit}$. 
\begin{theorem}[\textsc{GenFlowCut}] \label{thm : a special assignment of flow}
    Let $C$ be an $(s,t)$-cut and a vital edge $e=(u,v)$ contributes to $C$ in $G$. $C$ is a mincut for $e$ if and only if there is a maximum $(s,t)$-flow $f$ such that 
      \begin{enumerate}[nosep]
        \item $f_{in}(C)=0$.
        \item $e$ carries exactly $w_{min}(e)$ amount of $(s,t)$-flow and every other contributing edge $e'$ in $C$ is fully saturated, that is, $f(e')=w(e')$.
    \end{enumerate}  
\end{theorem}
\begin{proof}
    Suppose $C$ is a mincut for vital edge $e$ in $G$. After the removal of edge $e$, the capacity of $(s,t)$-cut $C$ is related to $f^*$ as follows.
\begin{equation}\label{eq : relevant edge use}
    \text{In graph } G\setminus\{e\},~ c(C)=f^*-w_{min}(e)    
\end{equation}
 Let $G'$ be the graph obtained from $G$ after reducing the capacity of edge $e$ from $w(e)$ to $w_{min}(e)$. 
 Therefore, using Equation \ref{eq : relevant edge use}, the capacity of  $C$ in $G'$ is $(f^*-w_{min}(e))+w_{min}(e)=f^*$. 
 Let ${\cal C}_{u,v}$ be the set of all $(s,t)$-cuts that keep $u$ on the side of $s$ and $v$ on the side of $t$. The capacity of each cut in ${\cal C}_{u,v}$ gets reduced by the same amount due to reduction in the capacity of $e$. Therefore, $C$ is a mincut for edge $e$ in $G'$ as well. 
 Hence the capacity of every cut belonging to ${\cal C}_{u,v}$ in $G'$ is at least $f^*$. 
 Capacity of any ($s,t$)-cut in $G'$, that does not belong to ${\cal C}_{u,v}$, is at least $f^*$ since it remains unaffected by the reduction in the capacity of edge $e$. 
 These facts imply that $f^*$ is the capacity of $(s,t)$-mincut in $G'$. Thus, using the strong duality between maximum $(s,t)$-flow and $(s,t)$-mincut, it follows that there exists a maximum $(s,t)$-flow $f$ in $G'$ such that value$(f,G')=f^*$. Using Lemma \ref{lem : conservation of flow} for $C$ in $G'$, we get the following equality.
 \begin{equation}\label{eq : conservation}
         f_{out}(C,G')-f_{in}(C,G')=f^*
\end{equation}
It follows from the capacity constraint that $f_{out}(C, G')\le c(C)$. So, Equation \ref{eq : conservation} implies that $f_{in}(C,G')\le c(C)-f^*$. Since $c(C)=f^*$ in $G'$ as shown above, we arrive at the following.
    \begin{equation}\label{eq : outflow f* and inflow 0}
        f_{in}(C,G')=0 ~~~~~\text{ and }~~~~~ f_{out}(C,G')=f^*
    \end{equation}
    This implies that, in $G'$, each outgoing edge of $C$ is fully saturated and each incoming edge does not carry any amount of flow. $f$ is also a valid maximum $(s,t)$-flow for graph $G$ because $f^*$ is the value of maximum $(s,t)$-flow in $G$ and $f((u,v))\le w((u,v))$ in $G$. 
    Therefore, it follows from Equation \ref{eq : outflow f* and inflow 0} that in graph $G$, $f_{in}(C)=0$, each outgoing edge of $C$ is fully saturated, and edge $e$ carries exactly $w_{min}(e)$ amount of $(s,t)$-flow.

     We now prove the converse part. Suppose there is a maximum $(s,t)$-flow in $G$ such that each outgoing edge of $C$, except $e$, is fully saturated and each incoming edge carries no flow. Therefore, $c(C)$ in $G$ is $f^*-w_{min}(e)+w(e)$. So, in graph $G\setminus \{e\}$, capacity of $C$ is $f^*-w_{min}(e)$. Since $e$ is a vital edge, it follows from definition that $f^*-w_{min}(e)$ is the capacity of $(s,t)$-mincut in $G\setminus \{e\}$. Therefore, $C$ is an $(s,t)$-mincut in $G\setminus \{e\}$. Hence, each $(s,t)$-cut in $G$ that keeps $u$ on the side of $s$ and $v$ on the side of $t$ has a capacity at least $f^*-w_{min}(e)+w(e)$. This implies that $C$ is a mincut for edge $e$.
\end{proof}

The following lemma can be seen as an immediate corollary of Theorem \ref{thm : a special assignment of flow}($2$).
% {\color{red}\begin{lemma}
% %\label{lem:vital-edge-remains-vital}
% Let $e$ be a vital edge and let $C(e)$
% be a mincut for $e$. $e$ remains a vital edge and $C(e)$ continues to be a mincut for edge $e$ even after {\color{red} we remove or reduce} the capacities of any number of edges from $C(e)$ provided the capacity of $(s,t)$-mincut remains the same.
% \end{lemma}}

\begin{lemma}
\label{lem:vital-edge-remains-vital}
Let $e$ be a vital edge and let $C(e)$ be a mincut for $e$. If the capacity of $(s,t)$-mincut remains the same after removing or reducing the capacities of any number of edges from the edge-set of $C(e)$, then $e$ remains a vital edge and $C(e)$ continues to be a mincut for edge $e$.
\end{lemma}

% {\color{red}
% We now state an important corollary of Theorem \ref{thm : a special assignment of flow}($1$).
% \begin{corollary}
% \label{cor:removing-all-incoming-edges-justified}
%  Let $C(e)$ be a mincut for a vital edge $e$. The capacity of $(s,t)$-mincut remains unchanged even after we remove all incoming edges of $C(e)$.
%  \end{corollary}
% }

\begin{remark}
Theorem \ref{thm : a special assignment of flow} crucially exploits Equation \ref{eq : relevant edge use} which holds for vital edges only. Note that $w_{min}(e)=0$ for a nonvital edge $e$. So, if $C$ is a mincut for $e$, $C$ might have capacity 
$>f^*$ even after removing nonvital edge $e$ from $G$. 
\end{remark}

%%%%%%%%%%%%%%%%%%%%%%%SECTION 4%%%%%%%%%%%%%%%%%%%%%%%%%%%%
\section{A Mincut Cover for All Vital Edges} \label{sec: mincut-cover}
For any subset ${\mathcal E}$ of vital edges, let $V({\mathcal E})$ denote the smallest set of vertices such that for each edge $(u,v)\in {\mathcal E}$, both $u$ and $v$ belong to $V({\mathcal E})$. We establish an upper bound on the cardinality of the mincut cover of ${\cal E}$ in terms of $|V({\mathcal E})|$.

Let $e$ be a vital edge, and let $C$ be a mincut for $e$. 
In undirected graphs, recall that each edge belonging to the edge-set of $C$ is a contributing edge of $C$. However, in directed graphs, there may exist edges incoming to $C$. Any such incoming edge, say $e'$, is not a contributing edge of $C$, and certainly, any mincut for edge $e'$ is different from $C$. Can $e'$ be a vital edge? 
The following lemma answers this question in negation. It crucially exploits the properties of maximum $(s,t)$-flow across $C$ as stated in \textsc{GenFlowCut} Property (Theorem \ref{thm : a special assignment of flow}). 
\begin{lemma}
If $C$ is a mincut for a vital edge, each incoming edge of $C$ must be a nonvital edge.
\label{lem:incoming-edge-irrelevant}
\end{lemma}
\begin{proof}
    Let $e'$ be an incoming edge of $C$. It follows from Theorem \ref{thm : a special assignment of flow}(1) that there is a maximum $(s,t)$-flow $f$ such that $f_{in}(C)=0$. So $e'$ does not carry any flow in maximum $(s,t)$-flow $f$. Hence, it follows from Lemma \ref{lem : vital edge in every maximum flow} that $e'$ is not a vital edge in $G$. 
\end{proof}
%An $(s,t)$-cut, by definition, partitions the set of vertices into two sets -- one containing vertex $s$ and another containing $t$. 
The following lemma, which can be seen as a corollary of Lemma \ref{lem:incoming-edge-irrelevant}, establishes that a mincut for a vital edge
partitions all the vital edges into three sets.

\begin{lemma}
  Let $G=(V,E)$ be a directed weighted graph with a designated source vertex $s$ and a designated sink vertex $t$. Let ${\cal E}$ be a subset of vital edges. For any edge $e\in {\mathcal E}$, let $C$ be a mincut for $e$. $C$ partitions the entire set ${\cal E}$ into three subsets -- $(i)$ ${\mathcal E}_{C}$: edges that are contributing to $C$, $(ii)$ ${\mathcal E}_{L}$: edges whose both endpoints belong to $C$, and $(iii)$ ${\mathcal E}_{R}$: edges whose both endpoints belong to $\overline{C}$. 
\label{lem:3-partitions}
\end{lemma}
Let ${\mathcal M}({\mathcal E})$ denote the mincut cover for a subset ${\mathcal E}$ of vital edges.
Suppose Lemma  \ref{lem:3-partitions} additionally guarantees the following property.

{\em $\mathbf{{\mathcal P}}:$~For each vital edge $e'\in {\mathcal E}$ that contributes to $C$, $C$ is a mincut for $e'$ as well. }

% %
 Property ${\mathcal P}$ ensures that $C$ is a mincut for all vital edges that belong to ${\mathcal E}_{C}$. This implies that ${\mathcal M}({\mathcal E}_L)\cup {\mathcal M}({\mathcal E}_R) \cup \{C\}$ is a mincut cover for ${\mathcal E}$ as well. Therefore, the cardinality of mincut cover for set ${\mathcal E}$ can be bounded as follows.
\begin{equation} \label{eq : equation size}
    |{\mathcal M}({\mathcal E})| \le |{\mathcal M}({\mathcal E}_L)|+|{\mathcal M}({\mathcal E}_R)|+1
\end{equation}
Let ${\mathcal N}(\mu)$ denote the cardinality of a mincut cover for any set ${\mathcal E}$ of vital edges with ${\mu=|{V({\mathcal E})}|}$. Note that ${\mathcal E}=\emptyset$ implies $\mu=0$, and ${\mathcal E}\not= \emptyset$ implies $\mu\ge 2$. 
Equation \ref{eq : equation size}, Lemma \ref{lem:3-partitions}, and Property ${\cal P}$ 
lead to the following recurrence for ${\mathcal N}(\mu)$. \\

%\begin{center}
%\fbox{\parbox{\textwidth}{
Base case:~~~~~~~~~
${\mathcal N}(0)=0$.
\begin{equation} \text{For any}~\mu\ge 2,~~~~~~ {\mathcal N}(\mu) \le  1 + {\mathcal N}(\mu_1) + {\mathcal N}(\mu_2), ~~~\mbox{where}~ \mu_1=|V({\mathcal E}_L)|,~ \mu_2=|V({\mathcal E}_R)|
\label{eq:recurrence-for-mincutcover-size}
\end{equation}
%}}
%\end{center}

% \begin{center}
% \fbox{\parbox{\textwidth}{
% Base case:~~~~~~~~~
% ${\mathcal N}(0)=0$.
% \begin{equation} \text{For any}~\mu\ge 2, {\mathcal N}(\mu) \le  1 + {\mathcal N}(\mu_1) + {\mathcal N}(\mu_2), ~~~\mbox{where}~ \mu_1=|V({\mathcal E}_L)|,~ \mu_2=|V({\mathcal E}_R)|.
% \label{eq:recurrence-for-mincutcover-size}
% \end{equation}
% }}
% \end{center}
% %
Note that $\mu_1,\mu_2<\mu$, and $\mu_1+\mu_2\le \mu$. Using induction on $\mu$, it is a simple exercise to show that ${\mathcal N}(\mu)\le \mu-1$ for all $\mu\ge 2$.

Though property ${\mathcal P}$ is crucially used in establishing an upper bound on the mincut cover of a set of vital edges, Lemma \ref{lem:3-partitions}, in its current form, does not guarantee it. However, we can enforce ${\mathcal P}$ easily as follows. 
%
%Instead of selecting edge $e$ arbitrarily from ${\mathcal E}$, we select the edge $e$ from ${\mathcal E}$ such that 
%
In Lemma \ref{lem:3-partitions}, the mincut for edge $e$ is of the least capacity among mincuts for all edges in ${\mathcal E}$. 

 Using Recurrence \ref{eq:recurrence-for-mincutcover-size} for the entire set of vital edges, we get an upper bound of $n-1$ on the cardinality of the mincut cover for all vital edges in $G$. This leads to Theorem \ref{thm : n-1 cuts}. %to the following theorem that answers Question \ref{question 1} in the affirmative.
% \begin{theorem}[Mincut Cover] \label{thm : n-1 cuts} 
%     For any directed weighted graph $G$ on $n$ vertices with a designated source vertex $s$ and a designated sink vertex $t$, there exists a set ${\cal C}_{min}$ containing at most $n-1$ $(s,t)$-cuts such that, for any vital edge $e$ in $G$, at least one mincut for edge $e$ is present in set ${\cal C}_{min}$.
% \end{theorem}
The following note emphasizes the need of property ${\mathcal P}$ in establishing the upper bound derived above.
\begin{figure}[ht]
  \begin{center}
    \includegraphics[width=0.2\textwidth]{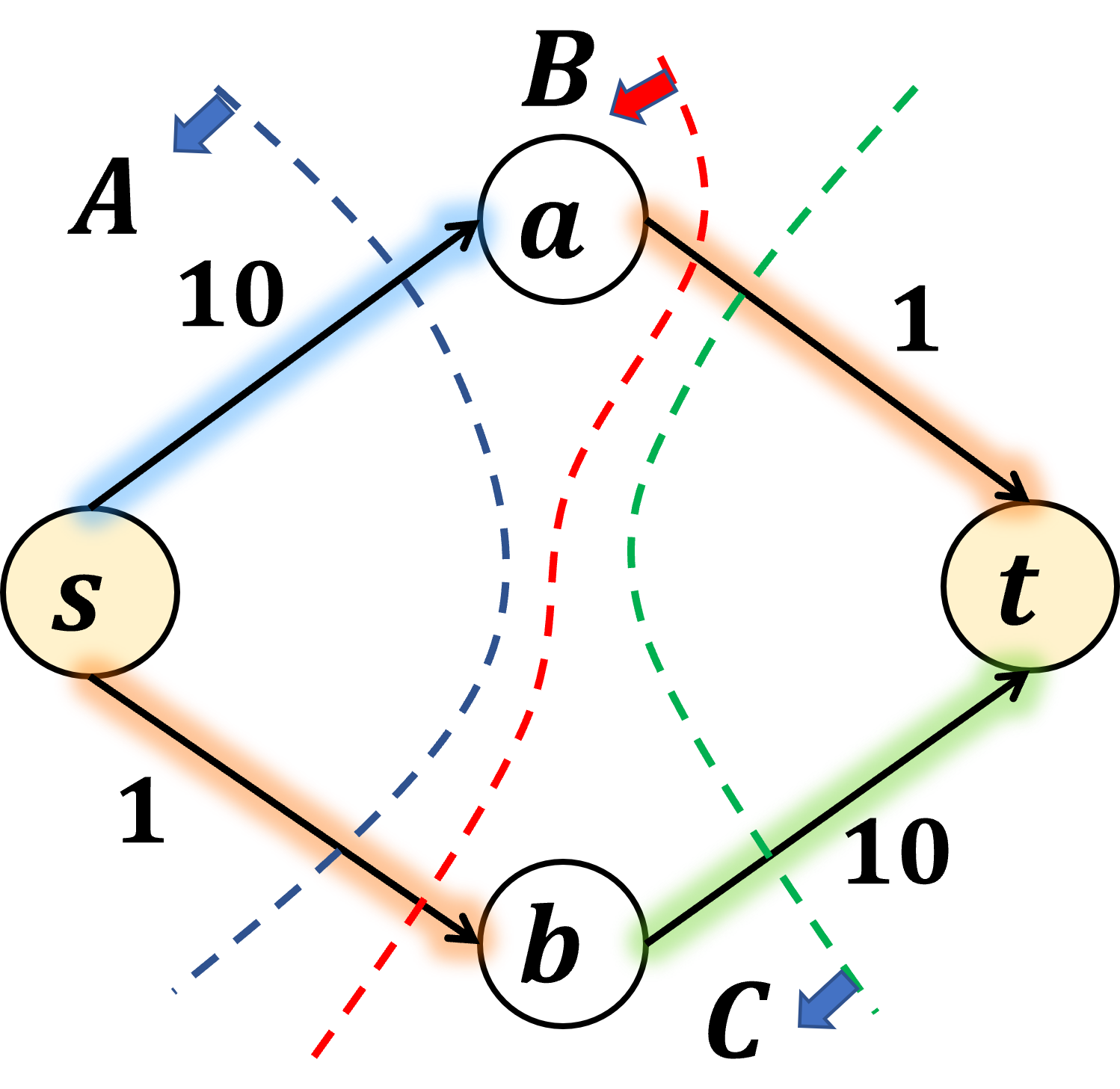}
  \end{center}
  \caption{A mincut cover fails to include mincut $B$ for edges $(a,t)$ and $(s,b)$ if property ${\mathcal P}$ is not ensured in the construction. An edge and the mincut for the edge are shown in the same color.} 
  \label{fig : mincut cover}
\end{figure}
\begin{note}  \label{note : arbitrary edge does not work}
Selecting any arbitrary edge $e\in {\cal E}$ might skip some mincuts for vital edges. Refer to Figure \ref{fig : mincut cover}. The set of vital edges, denoted by ${\mathcal E}$, contains four vital edges $(s,a)$, $(s,b)$, $(a,t)$, and $(b,t)$. Suppose we first select vital edge $(s,a)$ and mincut $A=\{s\}$ for edge $(s,a)$. It follows from Lemma \ref{lem:3-partitions} that ${\mathcal E}$ is partitioned into ${\mathcal E}_L$, ${\mathcal E}_R$ and ${\mathcal E}_A$. Observe that ${\mathcal E}_L$ does not contain any vital edge, vital edges $(a,t)$ and $(b,t)$ belong to ${\mathcal E}_R$, and $(s,b)$ belongs to ${\mathcal E}_A$. We now recurse on set ${\mathcal E}={\mathcal E}_R$ and select the vital edge $(b,t)$. Mincut $C=\{s,a,b\}$ partitions ${\mathcal E}$ into ${\mathcal E}_L$, ${\mathcal E}_R$, and ${\mathcal E}_{C}$. Observe that both ${\mathcal E}_L$ and ${\mathcal E}_R$ do not contain any vital edge, and vital edge $(a,t)$ belongs to ${\mathcal E}_C$. The process terminates as $V({\mathcal E}_L)$ and $V({\mathcal E}_R)$ are empty. We get a pair of $(s,t)$-cuts $A$ and $C$ as mincut cover for ${\mathcal E}$. Unfortunately, it does not contain mincut $B=\{s,a\}$ for edges $(a,t)$ and $(s,b)$. 
\end{note}

%%%%%%%%%%%%%%%%%%%%%%%%%%%%%%%%SECTION 5%%%%%%%%%%%%%%%%%%%%%%%%%%%
\section{Optimal Sensitivity Oracles for (s,t)-mincut} \label{sec: sensitivity-oracles}
 Let $e=(u,v)$ be an edge in $G$, and we wish to determine the impact of the failure of $e$ on $(s,t)$-mincut. We begin with a brief overview of ${\mathcal O}(n)$ space sensitivity oracle \cite{baswana2023minimum+} for $(s,t)$-mincuts in unweighted graphs. By construction of ${\mathcal D}_{PQ}(G)$ \cite{DBLP:journals/mp/PicardQ80}, there is a mapping from the vertices of $G$ to the nodes of ${\mathcal D}_{PQ}(G)$ such that the following assertion holds -- $u$ and $v$ are mapped to the same node of ${\mathcal D}_{PQ}(G)$ if and only if there is no $(s,t)$-mincut in $G$ that separates $u$ and $v$. In ${\mathcal D}_{PQ}(G)$, each $1$-transversal cut is an $(s,t)$-mincut in $G$. Exploiting this property, it is shown in \cite{baswana2023minimum+} that there is an $(s,t)$-mincut to which edge $(u,v)$ contributes if and only if the node containing $u$ succeeds the node containing $v$ in any topological ordering of ${\mathcal D}_{PQ}(G)$. 
 If $u$ and $v$ are mapped to the same node of ${\mathcal D}_{PQ}(G)$, it follows that the mincut for edge $(u,v)$ has a capacity strictly larger than that of the $(s,t)$-mincut. Therefore, if $G$ is unweighted, the failure of edge $(u,v)$ does not reduce the capacity of the $(s,t)$-mincut; so $(u,v)$ is not a vital edge. Thus, the mapping from $V$ to the nodes of ${\mathcal D}_{PQ}(G)$ along with its topological ordering serves as an ${\mathcal O}(n)$ space sensitivity oracle in unweighted graphs.

If $G$ is weighted, it is still possible that though $u$ and $v$ mapped to the same node of ${\mathcal D}_{PQ}(G)$, yet the failure of edge $(u,v)$ reduces $(s,t)$-mincut. In other words, if $G$ is weighted, there may be many vital edges internal to the nodes of ${\mathcal D}_{PQ}(G)$ (refer to Figure \ref{fig : G and D_pq(G)}). Thus ${\mathcal D}_{PQ}(G)$ fails to serve as sensitivity oracle for a weighted graph $G$; hence we need to explore the {\em structure} of cuts {\em internal} to a node of ${\mathcal D}_{PQ}(G)$. 
\begin{figure}[ht]
 \centering
    \includegraphics[width=\textwidth]{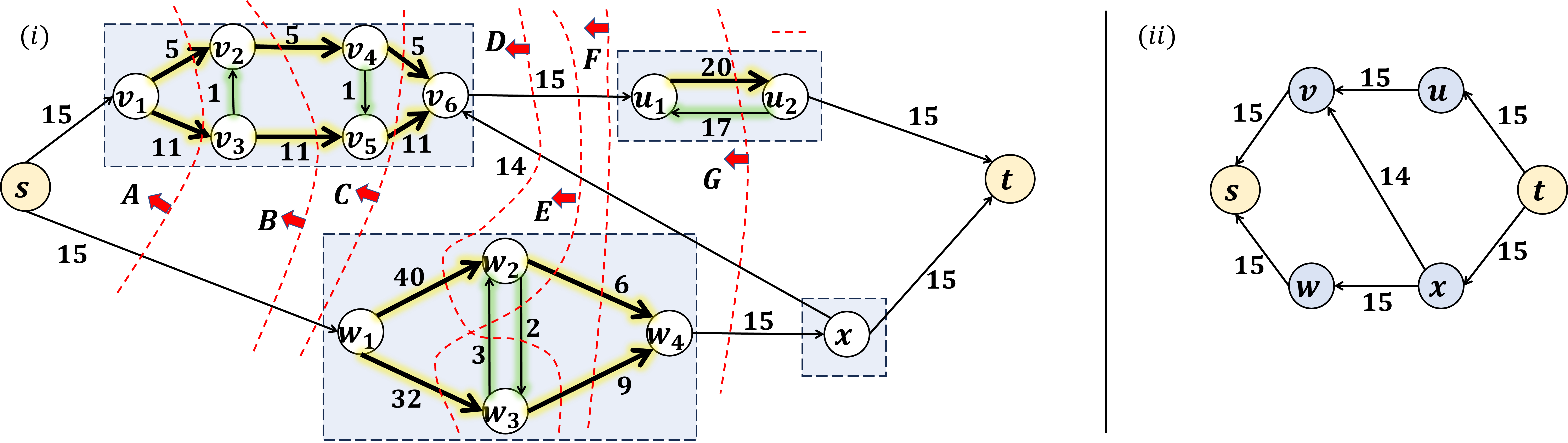} 
   \caption{$(i)$ A graph $H$ and $(ii)$ ${\mathcal D}_{PQ}(H)$. Thick edges in $(i)$ represent the vital edges of $H$ that are internal to the nodes of ${\mathcal D}_{PQ}(H)$. A mincut for them is represented by dashed curves.}
  \label{fig : G and D_pq(G)}
\end{figure}

Let $\textsc{cap}(e,\Delta)$ denote the query for reporting the capacity of $(s,t)$-mincut after reducing the capacity of $e\in E$ by $\Delta$ such that $0\le \Delta \le w(e)$. We now address the design of a compact data structure that can efficiently answer query $\textsc{cap}(e,\Delta)$ for $G$.

Let us first consider only the set of vital edges.
There may exist $\Omega(n^2)$ vital edges in a graph (refer to Figure \ref{fig : n2 values}($ii$) on Page 33). Therefore, explicitly storing the capacity of a mincut for each vital edge would occupy $\Omega(n^2)$ space. 
To design an ${\mathcal O}(n)$ space data structure, we use the following lemma for vital edges. This lemma follows from Lemma \ref{lem:3-partitions} and the discussion we used for bounding the size of mincut cover for vital edges.
\begin{lemma}
Let $U\subseteq V$ and an edge $e\in E_{vit}(U)$. If mincut for $e$, say $C$, has the least capacity among the mincuts for all vital edges from $E_{vit}(U)$, then, for any edge $e'\in E_{vit}(U)$, either $C$ is a mincut for $e'$ or both endpoints of $e'$ belong to $U\cap C$ or $U\cap \overline{C}$. 
\label{lem: divide-and-conquer-for-T_{s,t}} 
\end{lemma}
Using Lemma \ref{lem: divide-and-conquer-for-T_{s,t}}, we design a divide and conquer based algorithm to build a data structure that compactly stores the capacity of mincut for each vital edge as follows.

Let $e^*$ be an edge in $E_{vit}(V)$ such that the capacity of mincut for $e^*$ is less than or equal to the capacity of mincut for each vital edge in $G$.  Let $C$ be a mincut for $e^*$. $C$ serves as the mincut for each vital edge contributing to $C$. For the mincuts of the remaining vital edges, we recursively process $E_{vit}(V\cap C)$ and $E_{vit}(V\cap \overline{C})$. Refer to Algorithm \ref{alg : hierarchy tree} for the pseudocode of this recursive algorithm.

It can be observed that the data structure resulting from Algorithm \ref{alg : hierarchy tree} is a rooted binary tree. We refer to it as ${\mathcal T}_{vit}(G)$ henceforth.  
Its leaves are associated with disjoint subsets of the entire vertex set $V$ -- for each vertex $v\in V$, ${\mathcal L}(v)$ stores the pointer to the leaf to which $v$ is mapped. Each internal node $\nu$ 
in ${\mathcal T}_{vit}(G)$ has the following three fields. 
\begin{itemize}
    \item 
    $\nu.cap$ is the capacity of mincut for the vital edge selected during the recursive call in which  node $\nu$ is created.
    \item $\nu.left$ is the left child of $\nu$ and $\nu.right$ is the right child of $\nu$.
 \end{itemize} 
\begin{algorithm}[ht]
\caption{:~
\textsc{TreeConstruction}($V$) computes ${\mathcal T}_{vit}(G)$.}
\label{alg : hierarchy tree}
\begin{algorithmic}[1]
\Procedure{\textsc{TreeConstruction}$(U)$}{}
    \State Create a node $\nu$; 
    \If {$E_{vit}(U)=\emptyset$}{ 
        \For{each $x\in U$} {${\mathcal L}(x)\gets \nu$};
        \EndFor}
    \Else
        \State Let $C(e)$ denote a mincut for edge $e$;
        \State Select an edge $e^*\in E_{vit}(U)$ such that $c(C(e^*))\le c(C(e'))$ $\forall {e'}\in E_{vit}(U)$;
        \State Assign $\nu.cap\gets c(C(e^*))$; 
        \State $\nu.left \gets$ \textsc{TreeConstruction}($U\cap C(e^*)$);
        \State $\nu.right \gets$ \textsc{TreeConstruction}($U\cap \overline{C(e^*)}$);
    \EndIf
    \State \Return $\nu$;
\EndProcedure
\end{algorithmic}
\end{algorithm}
The following observation is immediate from the construction of ${\mathcal T}_{vit}(G)$.
\begin{observation} \label{obs : left and right child}
    Let $\nu$ be an internal node in ${\mathcal T}_{vit}(G)$ and $C$ be the $(s,t)$-cut associated with node $\nu$ (chosen in Step $8$ of Algorithm \ref{alg : hierarchy tree}).  Let $U$ be the set of vertices belonging to the subtree rooted at $\nu$.
    For any vertex $u\in U$, ${\cal L}(u)$ belongs to the subtree rooted at $\nu.left$ if and only if $u\in C$.
\end{observation}
Observe that ${\mathcal T}_{vit}(G)$ is a full binary tree with at most $n$ leaves. So the number of internal nodes is at most $n-1$. Moreover, every internal node of ${\mathcal T}_{vit}(G)$ stores ${\mathcal O}(1)$ information. Hence, we can state the following theorem.
\begin{theorem}\label{thm : reporting capacity for a relevant edge}
Let $G=(V,E)$ be a directed weighted graph on $n=|V|$ vertices with designated source vertex $s$ and designated sink vertex $t$. There is a rooted full binary tree occupying ${\mathcal O}(n)$ space that stores the capacity of mincut for each vital edge.
\end{theorem}
Let $e=(x,y)$ be a query edge. If $e$ is a vital edge, it follows from the construction of tree ${\mathcal T}_{vit}(G)$ in Theorem \ref{thm : reporting capacity for a relevant edge} that the \textsc{lca} of ${\mathcal L}(x)$ and ${\mathcal L}(y)$ stores the capacity of mincut for $(x,y)$. So, using the efficient data structure for \textsc{lca} \cite{bender2005lowest}, it takes ${\mathcal O}(1)$ time to answer query $\textsc{cap}(e, \Delta)$ for any vital edge $e$. However, what if the query edge is a nonvital edge? Although there is no impact on $(s,t)$-mincut capacity due to the reduction in the capacity of a nonvital edge, tree ${\mathcal T}_{vit}(G)$ stores no information explicitly about nonvital edges. Therefore, in order to answer query $\textsc{cap}(e, \Delta)$ for any edge, it seems a classification of all edges as vital or nonvital is required. 
Unfortunately, any such explicit classification would require ${\Omega}(m)$ space, which is also trivial for answering query $\textsc{cap}$ as discussed above. 

Interestingly, ${\mathcal T}_{vit}(G)$ can answer $\textsc{cap}(e, \Delta)$ for any edge $e$ without any such classification. This is because, as shown in the following lemma, it is already capable of classifying an edge to be vital or nonvital.

 \begin{lemma} \label{lem : condition for irrelevant}
     Let $e=(x,y)$ be any edge in $G$ such that ${\mathcal L}(x)\ne {\mathcal L}(y)$. Let $\nu$ be the \textsc{lca} of ${\mathcal L}(x)$ and ${\mathcal L}(y)$ in ${\mathcal T}_{vit}(G)$, and let $C$ be the $(s,t)$-cut associated with node $\nu$ in  ${\mathcal T}_{vit}(G)$. Edge $(x,y)$ is a nonvital edge if and only if exactly one of the following two conditions is satisfied.
     \begin{enumerate}
         \item ${\mathcal L}(x)\in \nu.right$ and ${\mathcal L}(y)\in \nu.left$.
         \item ${\mathcal L}(x)\in \nu.left$, ${\mathcal L}(y)\in \nu.right$, and $c(C)-w(e)\ge f^*$. 
     \end{enumerate}
 \end{lemma}
 \begin{proof} 
     Suppose edge $(x,y)$ is a nonvital edge. 
     Observe that ${\mathcal L}(x)$ and ${\mathcal L}(y)$ belong to the subtree rooted at $\nu$ in ${\mathcal T}_{vit}(G)$ because $\nu$ is the \textsc{lca} of ${\mathcal L}(x)$ and ${\mathcal L}(y)$. It follows from the construction of ${\mathcal T}_{vit}(G)$ (Algorithm \ref{alg : hierarchy tree}) that edge $(x,y)$ either contributes to $C$ or is an incoming edge to $C$. Suppose $(x,y)$ is a contributing edge to $C$. So using Observation \ref{obs : left and right child}, ${\mathcal L}(x)\in \nu.left$ and ${\mathcal L}(y)\in \nu.right$.
     If $c(C)$ falls below $f^*$ upon removal of edge $(x,y)$, $(x,y)$ would be a vital edge by Definition \ref{def:relevant-edges}. Hence,  $c(C)-w(e)\ge f^*$. So condition (2) is satisfied. 
     Suppose $(x,y)$ is an incoming edge to $C$, that is, $x\notin C$ and $y\in C$.  
     %Observe that ${\mathcal L}(x)$ and ${\mathcal L}(y)$ belong to the subtree rooted at $\nu$ in ${\mathcal T}_{vit}(G)$ because $\nu$ is the \textsc{lca} of ${\mathcal L}(x)$ and ${\mathcal L}(y)$. 
     It follows from Observation \ref{obs : left and right child} that ${\mathcal L}(x)$ must belong to the right subtree and ${\mathcal L}(y)$ must belong to the left subtree of $\nu$. So condition (1) is satisfied.

     We now prove the converse part. Suppose condition (1) is satisfied, that is, ${\mathcal L}(x)\in \nu.right$ and ${\mathcal L}(y)\in \nu.left$.  Observation \ref{obs : left and right child} implies that $x\in \overline{C}$ and $y\in C$. Hence, edge $(x,y)$ is an incoming edge of $C$.
      It follows from the construction of ${\mathcal T}_{vit}(G)$ that $C$ is a mincut for some vital edge. So, by Theorem \ref{thm : a special assignment of flow}(1), there exists a maximum ($s,t$)-flow $f$ such that $f_{in}(C)=0$. Hence $f(x,y)=0$. So it follows from Lemma \ref{lem : vital edge in every maximum flow} that $(x,y)$ is a nonvital edge.
      Suppose condition (2) is satisfied. Hence $(x,y)$ is an outgoing edge of $C$. Assume to the contrary that $e$ is a vital edge. It follows from the construction of ${\mathcal T}_{vit}(G)$ that $C$ is a mincut for $e$. Since $e$ is a vital edge, by Definition \ref{def:relevant-edges}, $c(C)-w(e)< f^*$, a contradiction.
 \end{proof}
 For any edge $e\in E$ and $0\le \Delta \le w(e)$, Algorithm \ref{alg : reporting} verifies conditions of Lemma \ref{lem : condition for irrelevant} to answer query $\textsc{cap}(e,\Delta)$ in ${\mathcal O}(1)$ time  using ${\mathcal T}_{vit}(G)$. So it leads to Theorem \ref{thm : reporting value}. %we can state the following Theorem.
\begin{note} \label{note : vitality of an edge}
    Given any edge $e$ and its capacity $w(e)$, the data structure in Theorem \ref{thm : reporting value} can be used to report the vitality of $e$ in ${\mathcal O}(1)$ time by assigning $\Delta=w(e)$. 
\end{note}

\begin{algorithm}[ht]
\caption{Reporting the capacity of $(s,t)$-mincut after decreasing the capacity of an edge }\label{alg : fault-tolerant query}
\begin{algorithmic}[1]
\Procedure{$\textsc{cap}((x,y),\Delta)$}{}
    \If{${\mathcal L}(x)$==${\mathcal L}(y)$}
        \State \Return $f^*$
    \Else
        \State $\nu \gets$ \textsc{lca} of ${\mathcal L}(x)$ and ${\mathcal L}(y)$ in ${\mathcal T}_{vit}(G)$.
    \EndIf
     \State $NewCapacity \gets \nu.cap-\Delta$
    \If {(${\mathcal L}(x)\in \nu.left ~ \land ~ {\mathcal L}(y)\in \nu.right~ ) ~\land ~NewCapacity<f^*$}
         \State \Return $NewCapacity$. 
    \Else
        \State \Return $f^*$.
    \EndIf
\EndProcedure
\end{algorithmic}
\label{alg : reporting}
\end{algorithm} 
% \begin{theorem} [Sensitivity Oracle for Reporting Capacity] \label{thm : reporting value}
%    Let $G$ be a directed weighted graph on $n$ vertices with a designated source vertex $s$ and a designated sink vertex $t$. There is an ${\mathcal O}(n)$ space data structure that, given any edge $e\in E$ and a value $\Delta$ satisfying 
%    $0\le \Delta\le w(e)$, can report in ${\mathcal O}(1)$ time the capacity of $(s,t)$-mincut after reducing the capacity of $e$ by $\Delta$.  
% \end{theorem}
%
\paragraph*{Reporting an $(s,t)$-mincut:} 
 An internal node, say $\nu$, of ${\mathcal T}_{vit}(G)$ stores only the capacity of the mincut for a vital edge. We augment $\nu$ with a mincut for the edge that was picked by Algorithm \ref{alg : hierarchy tree} in Step $8$.  
The size of the resulting data structure is ${\mathcal O}(n^2)$. 
Observe that the condition $\Delta \le w(e)$ can be verified for any edge $e$ in ${\mathcal O}(1)$ time if we store the capacity of all edges of $G$. This will require only additional ${\cal O}(m)={\cal O}(n^2)$ space. 
So this completes the proof of Theorem \ref{thm : main result}. %we can state the following theorem. 

\subsection*{Labeling Scheme} 
%A labeling scheme for Query $\textsc{cap}(e,\Delta)$ assigns each vertex of the graph a label of a polylogarithmic size such that, given any edge $(u,v)$ present in $G$ and a value $0\le \Delta\le w((u,v))$, Query $\textsc{cap}(e,\Delta)$ can be answered in polylogarithmic time by processing only the labels of $u$ and $v$. 
Let $f$ be any function defined on the vertex/edge set of the graph. A labeling scheme of $f$ assigns small labels to each vertex/edge of the graph in such a way that, for any given subset $A$ of edges/vertices, $f(A)$ can be computed only by using the labels of vertices/edges in $A$.
 Katz, Katz, Korman, and Peleg \cite{DBLP:journals/siamcomp/KatzKKP04} are the first to design a labeling scheme for flow and connectivity. In particular, they designed a labeling scheme that can report the capacity of 
the cut of the least capacity separating any pair of vertices in an undirected weighted graph. In order to accomplish this task, \cite{DBLP:journals/siamcomp/KatzKKP04} also design and use the following labeling scheme for LCA in a rooted tree. We denote the \textsc{lca} of a pair of nodes $\{\mu,\nu\}$ in a tree by $\textsc{lca}(\mu,\nu)$. 
\begin{lemma} [\cite{DBLP:journals/siamcomp/KatzKKP04}] \label{lem : katz labeling scheme}
     Let $T$ be any rooted tree on $n$ nodes with a maximum edge capacity $W$. Each node $\mu$ of $T$ stores a value $\mu.\textsc{val}\ge0$. Tree $T$ can be preprocessed to compute a label of size ${\mathcal O}(\log^2{n}+\log{n}\log{W})$ bits for each node such that, given any pair of nodes $\mu$ and $\nu$ from $T$, the value $\textsc{lca}(\mu, \nu).\textsc{val}$ can be reported in polylogarithmic time by processing only the labels of $\mu$ and $\nu$.
\end{lemma}
We present a labeling scheme for reporting the capacity of $(s,t)$-mincut after changing the capacity of any given edge in $G$.  The labeling scheme is inspired by Algorithm \ref{alg : reporting} that uses tree ${\mathcal T}_{vit}(G)$. For answering Query $\textsc{cap}((x,y),\Delta)$, it follows from Algorithm \ref{alg : reporting} that if $(x,y)$ is a vital edge, then the new capacity of $(s,t)$-mincut is $\textsc{lca}({\mathcal L}(x),{\mathcal L}(y)).cap-\Delta$; otherwise the new capacity remains $f^*$. In order to determine the capacity stored at \textsc{lca} of ${\mathcal L}(x)$ and ${\mathcal L}(y)$ in ${\mathcal T}_{vit}(G)$, we use labeling scheme for \textsc{lca} on ${\mathcal T}_{vit}(G)$ as given in Lemma \ref{lem : katz labeling scheme}. The label size for each leaf node of ${\mathcal T}_{vit}(G)$ is ${\mathcal O}(\log^2{n}+\log{n}\log{W})$ bits, where $W$ is the maximum capacity of any edge in $G$. %capacity of any mincut for a vital edge in $G$. The only remaining task is to determine whether given query edge $(x,y)$ is vital or not.   
We now design a labeling scheme for the task of identifying whether query edge is vital based on Observation \ref{obs : left and right child}. At each leaf node ${\nu}$ of ${\mathcal T}_{vit}(G)$, we store a positive integer value, denoted by $\textsc{ord}(\nu)$, according to the following order. \\
    {\em For a pair of leaf nodes $\mu$ and $\nu$, $\mu\in \textsc{lca}(\mu,\nu).left$ and $\nu\in \textsc{lca}(\mu,\nu).right$ if and only if $\textsc{ord}(\mu)<\textsc{ord}(\nu)$.}    \\    
%{\color{red} This ordering exists because of the following transitivity property. For any three leaf vertices $x$,$y$, and $z$, if $\textsc{ord}(x)<\textsc{ord}(y)$ and  $\textsc{ord}(y)<\textsc{ord}(z)$, then $\textsc{ord}(x)<\textsc{ord}(z)$. $\textsc{ord}(x)>\textsc{ord}(z)$, then it implies that } 
This information $\textsc{ord}$ helps in classifying an edge to be vital or nonvital as follows. 
\begin{lemma} \label{lem : vital nonvital labeling}
    An edge $(x,y)$ is a vital edge if and only if $\textsc{ord}({\mathcal L}(x))<\textsc{ord}({\mathcal L}(y))$ and \\$\textsc{lca}({\mathcal L}(x),{\mathcal L}(y)).cap-w((x,y))<f^*$. 
\end{lemma}
The proof of Lemma \ref{lem : vital nonvital labeling} is along similar lines as Lemma \ref{lem : condition for irrelevant}. 

The labeling scheme for answering Query $\textsc{cap}(e,\Delta)$ stores the following labels at each leaf node $\mu$ of the tree ${\mathcal T}_{vit}(G)$ -- $(i)$ ${\mathcal O}(\log^2{n}+\log{n}\log{W})$ bits label given in \cite{DBLP:journals/siamcomp/KatzKKP04}, $(ii)$ the value $\textsc{ord}(\mu)$, and $(iii)$ the capacity of $(s,t)$-mincut, that is $f^*$. Finally, for each vertex $x$ in graph $G$, we store the label associated with ${\mathcal L}(x)$. Observe that, along the similar procedure in Algorithm \ref{alg : reporting}, we can answer Query $\textsc{cap}((x,y),\Delta))$ in polylogarithmic time by using only the labels of $x$ and $y$. 

This completes the proof of the following theorme.
% Theorem \ref{thm : labeling scheme}. %(Labeling scheme for increasing capacity of any edge is given in Appendix \ref{app : insertion}).
 
\begin{theorem} \label{thm : labeling scheme}
    Any directed weighted graph $G$ on $n$ vertices, with a designated source vertex $s$ and a designated sink vertex $t$, can be preprocessed to compute labels of ${\mathcal O}(\log^2{n}+\log{n}\log{W})$ bits for each vertex such that, given any edge $e=(x,y)$ in $G$ and any value $\Delta$ satisfying $0\le \Delta \le w(e)$, the capacity of $(s,t)$-mincut after reducing the capacity of $e$ by $\Delta$ can be determined in polylogarithmic time by processing only the labels of $x$ and $y$. Here $W$ is the maximum capacity of any edge in $G$. 
\end{theorem}

% }

%%%%%%%%%%%%%%%%%%%%%%%%%%%%%%%%%%%%%%%%%%%%%%%%%%%%%%%%%%%%
\section{An Algorithm for Computing All Vital Edges} 
%%%%%%%%%%%%%%%%%%%%%%%%%%%%%%%%%%%%%%%%%%%%%%%%%%%%%%%%%%%%
\label{sec: algorithm}
 The existing algorithms \cite{DBLP:journals/networks/AnejaCN01, DBLP:journals/networks/AusielloFLR19} that compute the most vital edge or all vital edges in undirected graphs take \textit{cut-based} approaches as follows.  For computing the most vital edge in an undirected graph, Aneja, Chandrasekaran, and Nair \cite{DBLP:journals/networks/AnejaCN01} establish that there exists an $(s,t)$-cut such that the most vital edge is the maximum capacity edge among all the edges that contribute to the cut. Ausiello, Franciosa, Lari, and Ribichini \cite{DBLP:journals/networks/AusielloFLR19} define function $F$ on cuts as in Equation \ref{eq : s,t cuts} in order to use the ancestor tree of Cheng and Hu \cite{DBLP:journals/anor/ChengH91}.
 To compute all vital edges efficiently in a directed weighted graph, we take a \textit{flow-based} approach that analyzes flow along a set of edges in a maximum $(s,t)$-flow. We first classify all vital edges into two types based on the maximum flow that an edge can carry in any maximum $(s,t)$-flow.

\paragraph*{A classification of vital edges:}
In graph $G$, 
there may exist multiple $(s,t)$-flows attaining value $f^*$.  The flow along an edge may be different in these multiple maximum $(s,t)$-flows.
A vital edge is a \textit{tight} edge if it carries flow equal to its capacity in at least one maximum $(s,t)$-flow; otherwise, it is a \textit{loose} edge.\\ 

Our algorithm for computing all vital edges makes use of the classification of vital edges, and employs two results that were derived for solving two seemingly {\em unrelated} problems.
 
To compute all the loose edges, we crucially exploit the following result on maximum $(s,t)$-flow, which has been used for solving \textit{mincost maximum $(s,t)$-flow problem}.
\begin{lemma} [Theorem 11.1 and Theorem 11.2 in \cite{DBLP:books/daglib/0069809}] \label{lem : spanning tree property} 
    For any directed weighted graph $G$, there exists a maximum $(s,t)$-flow $f^\#$ such that the edges that carry nonzero flow but not fully saturated are at most $n-1$.
\end{lemma}
For any maximum $(s,t)$-flow, by definition, the set of all loose edges is a subset of all edges that carry nonzero flow but are not fully saturated. So, by Lemma \ref{lem : spanning tree property}, the number of loose edges can be at most $n-1$. Moreover, there is an algorithm that, given any maximum $(s,t)$-flow, can compute $f^\#$ of Lemma \ref{lem : spanning tree property} in ${\mathcal O}(mn)$ time \cite{DBLP:books/daglib/0069809}. As a result, all the loose edges of $G$ can be computed using ${\mathcal O}(n)$ maximum $(s,t)$-flow computations.

Observe that the efficient computation of the set of loose edges crucially exploits the fact that there can be at most $n-1$ loose edges.
However, there can be $\Omega(n^2)$ tight edges in a graph (discussed in Section \ref{sec : cardinality of tight edges and loose edges}). %refer to Figure \ref{fig : n2 values}($ii$) on page 33).
Hence, to compute all the vital edges, the main challenge arises in computing all the tight edges of $G$. We now state a property of tight edges that plays a crucial role in their efficient computation.
\subsection*{A property satisfied by tight edges}
Let $E_{min}$ be the set of all edges that contribute to $(s,t)$-mincuts. Set $E_{min}$ is a subset of all the tight edges since all edges in $E_{min}$ are fully saturated in every maximum $(s,t)$-flow. Exploiting the strong duality between maximum $(s,t)$-flow \& $(s,t)$-mincut and Lemma \ref{lem : conservation of flow}, it can be observed that each edge $(u,v)\in E_{min}$ satisfies the following property. 
If $C$ is any $(s,t)$-cut such that $v\in C$ and $u\in \overline{C}$, then $c(C)$ is strictly greater than $f^*$. The following lemma shows that this property is satisfied by each tight edge as well.   

\begin{lemma} \label{lem : tight and separation}
    Let $e=(u,v)$ be a tight edge. Let $C$ be a mincut for edge $e$ and $C_{v,u}$ be an $(s,t)$-cut of the least capacity that keeps $v$ on the side of $s$ and $u$ on the side of $t$. Then, $c(C)<c(C_{v,u})$.
\end{lemma}
\begin{proof}

      Edge $(u,v)$ is a tight edge, so by definition, there is a maximum $(s,t)$-flow $f$ such that $f(e)$ is equal to $w(e)$. This would imply that
     $f_{in}(C_{v,u})$ is at least $w(e)$ since
     edge $(u,v)$ is an incoming edge for cut $C_{v,u}$.
     It follows from Lemma \ref{lem : conservation of flow} that $f_{out}(C_{v,u})-f_{in}(C_{v,u})=f^*$.  %for the $(s,t)$-cut $C_{v,u}$. 
     Therefore, $f_{out}(C_{v,u})-w(e)\ge f^*$. Since the capacity of an $(s,t)$-cut is an upper bound on the value of any outgoing flow of the $(s,t)$-cut, therefore,  $f_{out}(C_{v,u})\le c(C_{v,u})$. Hence we arrive at the following inequality.
    \begin{equation} \label{eq : 1}
        c(C_{v,u})\ge f^*+w(e)
    \end{equation}
    It follows from Theorem \ref{thm : a special assignment of flow}$(2)$ that $w_{min}(e)$ is the minimum amount of flow that must pass through edge $(u,v)$ to get a maximum flow of value $f^*$. Let $f'$ be the value of maximum $(s,t)$-flow in graph $G\setminus\{e\}$. Thus we arrive at the following equality.
    \begin{equation} \label{eq : 2}
        f^*=f^{'}+w_{min}(e)
    \end{equation}
    We now add $(w(e)-w_{min}(e))$ on both sides of Equation \ref{eq : 2} and get the following equation.
    \begin{equation} \label{eq : 11}
        f^*+w(e)-w_{min}(e)=f^{'}+w(e)
    \end{equation}
   It follows from Theorem \ref{thm : a special assignment of flow}(2) that $f^{'}$ is also equal to the sum of capacities of all edges contributing to $C$ except edge $(u,v)$. Therefore, $f^{'}+w(e)$ is equal to the capacity of $C$ in $G$. So it follows from Equation \ref{eq : 11} that $f^*+w(e)-w_{min}(e)=c(C)$. Since $(u,v)$ is a vital edge, therefore, $w_{min}(e)>0$ using Observation \ref{obs: min-capacity-of-an-edge}. Hence,
    \begin{equation}\label{eq : 3}
        f^*+w(e)>c(C)
    \end{equation} 
    It follows from Inequality \ref{eq : 1} and Inequality \ref{eq : 3} that $c(C_{v,u})>c(C)$.
\end{proof}
For any tight edge $(u,v)$, let $C$ be an $(s,t)$-cut of the least capacity that separates vertex $u$ and vertex $v$. It follows from Lemma \ref{lem : tight and separation} that $C$ is also a mincut for edge $(u,v)$. So we can state the following theorem.
\begin{theorem}\label{thm : tight edge property}
    Let $G$ be a directed weighted graph with a designated source vertex $s$ and a designated sink vertex $t$. 
    For any tight edge $e=(u,v)$ in $G$, a mincut for edge $e$ is identical to an {$(s,t)$-cut of the least capacity} that separates vertex $u$ and vertex $v$. 
\end{theorem}
\subsection*{Algorithm for computing tight edges}
For the given directed weighted graph $G$, we define function $F$ as in Equation \ref{eq : s,t cuts} and then compute the Ancestor tree of Cheng and Hu \cite{DBLP:journals/anor/ChengH91} for $(s,t)$-cuts.  
 For each pair of vertices, ${\mathcal T}_{(s,t)}$ stores an $(s,t)$-cut of the least capacity separating them at their \textsc{lca}. This tree, denoted by ${\mathcal T}_{(s,t)}$, can be built using ${\mathcal O}(n)$ maximum $(s,t)$-flow computations \cite{DBLP:journals/anor/ChengH91}. 
We augment ${\mathcal T}_{(s,t)}$ with a data structure for \textsc{lca} queries \cite{bender2005lowest}. 

To compute all tight edges, we process ${\mathcal T}_{(s,t)}$ and the edges of $G$ as follows. Let $(u,v)$ be an edge in $G$. 
We perform \textsc{lca} query on ${\mathcal T}_{(s,t)}$ for $u$ and $v$ to get the $(s,t)$-cut, say $C$, of the least capacity that separates $u$ and $v$. 
We determine whether $C$ satisfies the following two conditions. 
\begin{enumerate}
    \item Edge $(u,v)$ contributes to $C$.
    \item $c(C)-w((u,v))<f^*$
\end{enumerate}
If both conditions are satisfied, by Definition \ref{def:relevant-edges}, $(u,v)$ is a vital edge. Observe that there may be vital edges that do not satisfy Condition (1); refer to Figure \ref{fig : gamma non zero}. However, it follows from Theorem \ref{thm : tight edge property} that each tight edge does satisfy these two conditions. After processing all edges of $G$, let $S$ be the resulting set of vital edges that satisfy both these conditions. 
We eliminate from $S$ all the loose edges to get all the tight edges.
Thus we can state the following theorem.  
\begin{theorem} \label{thm : tight edge computation}
    For any directed weighted graph $G$ on $n$ vertices with a designated source vertex $s$ and a designated sink vertex $t$, there is an algorithm that computes all tight edges of $G$ using ${\mathcal O}(n)$ maximum $(s,t)$-flow computations. 
\end{theorem}
Theorem \ref{thm : tight edge computation} and the discussion on computing the loose edges lead to Theorem \ref{thm : computing all vital edges}.
\begin{remark}
    In an undirected graph, observe that a mincut for an edge $(u,v)$ is always an $(s,t)$-cut of the least capacity that separates $u$ and $v$. Therefore, computing all vital edges in an undirected graph amounts to just verifying Condition $(2)$ for each edge as shown in %by Ausiello, Lari, Franciosa, and Ribichini 
    \cite{DBLP:journals/networks/AusielloFLR19}.
\end{remark}

\begin{figure}
  \begin{center}
    \includegraphics[width=0.3\textwidth]{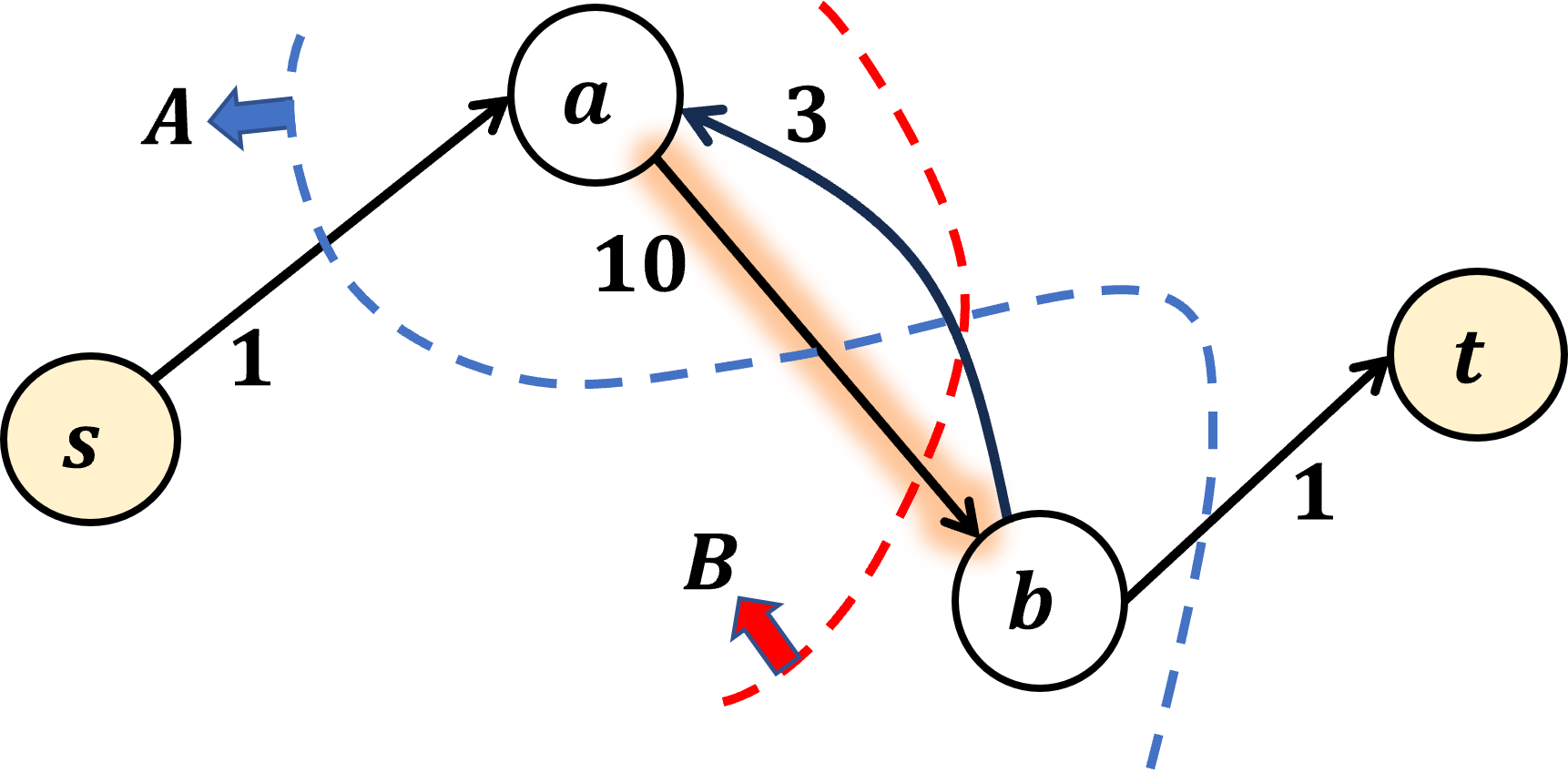}
  \end{center}
  \caption{$A$ is the $(s,t)$-cut of least capacity that separates $a$ and $b$. Edge $(a,b)$ is a vital edge, and it is incoming to $A$.}
 % $(ii)$ ${\mathcal D}_{vit}(G)$=$G$.  $\{s,a\}$ is a $1$-transversal cut, but not a mincut for any vital edge. 
  % $(iii)$ mincuts $A$ and $B$ for vital edges $(c,t)$ and $(d,t)$, respectively, are not closed under union. Moreover, edge $(b,a)$ (likewise edge $(c,d)$) contributes to $A$ (likewise to $B$) and is incoming to $B$ (likewise to $A$).
  \label{fig : gamma non zero}
\end{figure}

\subsection{Tight and Loose Edges: Cardinality and Computation} \label{sec : cardinality of vital edges}
Not only the classification of vital edges plays a crucial role in their efficient computation, but it is also used in designing a compact structure for storing and characterizing all mincuts for all vital edges (refer to Section \ref{sec: compact-structures}). We believe that this classification might find other applications. Hence, it is also important to provide a bound on the number of edges belonging to each type and to design an algorithm that, given any edge, can efficiently determine whether it is tight or loose.

\subsubsection{Cardinality of Tight and Loose Edges} \label{sec : cardinality of tight edges and loose edges}
In this section, we address the cardinality of tight and loose edges. Let us consider a complete bipartite graph with two sets of vertices $L$ and $R$ such that all edges have finite capacities and are directed from $L$ to $R$. We add a source vertex $s$ and a sink vertex $t$ to this bipartite graph and define the edges as follows. From source $s$, we add an edge of infinite capacity to each vertex in $L$, and similarly, we add an edge of infinite capacity from each vertex in $R$ to $t$ (refer to Figure \ref{fig : n2 values}($ii$) on Page 33). Observe that, in the resulting graph, each edge from $L$ to $R$ is a tight edge. This shows the existence of graphs on $n$ vertices where the number of tight edges can be ${\Omega}(n^2)$. 

%While tight edges satisfy property ${\mathcal P}_1$, the loose edges may potentially fail to satisfy property ${\mathcal P}_1$.
%We show that the edges that fail to satisfy property ${\mathcal P}_1$ and also carry a nonzero amount of flow in a maximum $(s,t)$-flow can be $\Omega(n^2)$ (Theorem \ref{thm : n2 unclassified edges} in Appendix \ref{sec : n2 unclassified edges with flow}). 
%However, it follows from Lemma \ref{lem : n-1 edges graph} that the number of loose edges can be at most $n-1$.
It follows from Lemma \ref{lem : spanning tree property} that the number of loose edges can be at most $n-1$.
Observe that a graph with a simple directed path from $s$ to $t$  with different capacities on edges contains exactly $n-2$ loose edges. We present here an alternate proof using Mincut Cover (Theorem \ref{thm : n-1 cuts}) to show that the number of loose edges cannot be more than $n-2$; hence, the bound of $n-2$ is tight as well. 
To establish this claim, we first state the following lemma. 
\begin{lemma} \label{lem : at most one loose edge in a cut}
    For a vital edge $e\in E_{vit}$, there is at most one loose edge out of all contributing edges of any mincut $C(e)$ for edge $e$. 
\end{lemma}
\begin{proof}
    Suppose there is an edge $e'$ ($\ne e$) that is contributing to $C(e)$ and is a loose edge. It follows from Theorem \ref{thm : a special assignment of flow}$(2)$ that there is a maximum flow $f$ such that each edge contributing to $C(e)$, except edge $e$, has a flow equal to its capacity. Since $e'$ is a contributing edge and not the edge $e$, therefore $f((u,v))=w((u,v))$. This violates the definition of a loose edge, a contradiction. Therefore, only edge $e$ can be a loose edge out of all contributing edges of $C(e)$. 
\end{proof}

We now use Lemma \ref{lem : at most one loose edge in a cut} and Theorem \ref{thm : n-1 cuts} (Mincut Cover) to establish an upper bound on the number of loose edges.
\begin{lemma} \label{lem : n-1 loose edges}
    The number of loose edges in $G$ is at most $n-2$.
\end{lemma}
\begin{proof}
 Theorem \ref{thm : n-1 cuts} states that there is a set ${\mathcal M}$ containing at most $n-1$ $(s,t)$-cuts such that at least one mincut for every vital edge of $G$ is present in ${\mathcal M}$. It follows from the construction that each $(s,t)$-cut $C$ present in ${\mathcal M}$ is a mincut for at least one vital edge. Suppose the number of loose edges is strictly greater than $n-1$. Using pigeon hole principle, it follows that there exists at least one $(s,t)$-cut $C\in {\mathcal M}$ such that at least two loose edges are contributing to $C$, which contradicts Lemma \ref{lem : at most one loose edge in a cut}. It follows from the strong duality between maximum $(s,t)$-flow and $(s,t)$-mincut that each edge that contributes to $(s,t)$-mincut is fully saturated. Observe that there is an $(s,t)$-cut $C$ in ${\mathcal M}$ such that $C$ is an $(s,t)$-mincut. Therefore, all the contributing edges of $C$ are fully saturated and hence cannot be a loose edge. This establishes that there are at most $n-2$ loose edges in $G$.
\end{proof}
This completes the proof of the following theorem. %\ref{thm : n-1 loose edges}. 
\begin{theorem}  \label{thm : n-1 loose edges}
    Let $G$ be any directed weighted graph on $n$ vertices with a designated source vertex $s$ and a designated sink vertex $t$. There can be $\Omega(n^2)$ vital edges in $G$, but the loose edges can never be more than $n-2$.
\end{theorem}

% This completes the proof of Theorem \ref{thm : n-1 loose edges}.

\subsubsection{An Algorithm for Computing the Maximum Flow along an Edge} \label{app : maxflow at an edge}
%In this section, we design an algorithm that, given an edge $e$, efficiently computes a maximum $(s,t)$-flow $f_M$ such that $f_M(e)$ is maximum among all possible maximum $(s,t)$-flow in $G$. 

% {\color{blue}%Not only the classification of vital edges plays crucial role in their efficient computation, but it is also used in designing a compact structure for all vital edges (Section 1.1.3). We believe that this classification might find other applications. Hence, it is also interesting to efficiently determine, given any edge, whether it is tight or loose. 
In this section, we show that, given the value of maximum $(s,t)$-flow $f^*$, just one maximum $(s,t)$-flow computation on graph $G$ after {\em small} modification is sufficient for determining whether an edge is tight or loose. %this task. %(Refer to Theorem \ref{thm : maxflow at an edge} in Appendix ??). }

Let $e=(u,v)$ be any edge in $G$. To compute the maximum flow that edge $e$ can carry in any maximum $(s,t)$-flow,  we first construct a graph $G'$ from $G$ as follows. 

\paragraph*{Construction of graph $G'$:} Add vertices $s'$ and $t'$. These will be the source and sink vertices in $G'$. Add an edge $e_s$ from $s'$ to $s$ and an edge $e_t$ from $t$ to $t'$ each having capacity $f^*$. %$\text{value}(f,G)=f^*$.
we remove edge $e$ and add two edges of capacity $w(e)$ each -- one edge is $e_v=(s',v)$ and the other edge is $e_u=(u,t')$. \\

%Observe that, in graph $G'$, the number of vertices is $n+2$ and the number of edges is $m+3$, which implies 
It is easy to establish that $G'$ has  ${\mathcal O}(m)$ edges and ${\mathcal O}(n)$ vertices. The following lemma states a crucial property about the maximum ($s',t'$)-flow in $G'$. 

\begin{lemma}
There exists a maximum $(s',t')$-flow in $G'$ in which edges $e_s$ and $e_t$ carry flow of value $f^*$. 
\label{lem:flow-inG'-es-et-f*}
\end{lemma}
\begin{proof}
     Let $f$ be any maximum $(s,t)$-flow in $G$. We define a flow $f'$ in $G'$ as follows. For each edge $\hat{e}$ in $G'$, that is also present in $G$, $f'(\hat{e})=f(\hat{e})$. For edges $e_s$ and $e_t$, we define $f'(e_s)=f'(e_t)=f^*$. 
     For edges $e_u$ and $e_v$, we define $f'(e_u)=f'(e_v)=f(e)$. It is easy to verify that $f'$ is a valid flow in $G'$. This also establishes that the value of the maximum $(s',t')$-flow in $G'$ is at least $f^*$. 
      We now use any algorithm based on augmenting paths (e.g. the Edmond-Karp algorithm \cite{10.1145/321694.321699} or Dinitz's algorithm \cite{dinitz1970algorithm}) to compute the maximum $(s',t')$-flow in $G'$ with $f'$ as the initial flow. 
      Note that $e_s$ and $e_t$ are fully saturated in $f'$ and, therefore, there is no augmenting path from $s'$ to $t'$ in $G'_{f'}$ that passes through $e_s$ or $e_t$. So each augmenting path used by the algorithm to send additional flow from $s'$ to $t'$ will not alter the flow along $e_s$ and $e_t$. Hence, the resulting maximum ($s',t'$)-flow will have both $e_s$ and $e_t$ carrying flow of value $f^*$. 
\end{proof}

%Let $f^*+\alpha$ be the maximum $(s',t')$-flow in $G'$ where $\alpha\ge 0$. 
% The following lemma establishes that $\alpha$ is an upper bound on the flow along edge $e$ in each maximum $(s,t)$-flow.
% \begin{lemma} \label{lem : max flow is alpha}
%     For any maximum $(s,t)$-flow $f_1$ in $G$,
%     $f_1(e)\le \alpha$. 
% \end{lemma}
% \begin{proof}
%     We give a proof by contradiction. Let us assume to the contrary that there is a maximum $(s,t)$-flow $f_1$ in $G$ such that $f_1(e)>\alpha$. We now obtain graph $G'$ and then assign flow on edges of $G'$ as follows. Each edge $\hat{e}$ in $G'$ which also belongs to $G$ has flow $f_1(\hat{e})$. Edges $e_s$ and $e_t$ are assigned flow $f^*$, and edges $e_u$ and $e_v$ are assigned flow $f_1(e)$. Since $f_1(e)\le w(e)$, the obtained flow $f_2$ is a valid $(s',t')$-flow in graph $G'$. Observe that the value of $f_2$ is $f^*+f_1(e)$.
%     Since $\alpha<f_1(e)$, we have a valid $(s',t')$-flow in graph $G'$ whose value is more than the value of the maximum $(s',t')$-flow in $G'$, a contradiction.     
% \end{proof}
%Exploiting Lemma \ref{lem : max flow is alpha}, 

In the following 2 lemmas, we establish the relation between the value of maximum $(s',t')$-flow in $G'$ and the value of maximum flow that edge $e$ can carry in any maximum $(s,t)$-flow in $G$. %show that there exists a maximum $(s,t)$-flow in $G$ such that edge $e$ carries a flow of value exactly $\alpha$.
\begin{lemma} \label{lem: part-1}
    If $f^*+\alpha$ is the maximum $(s',t')$-flow in graph $G'$, the maximum flow that $e$ can carry in any maximum $(s,t)$-flow is at least $\alpha$. 
\end{lemma}
\begin{proof}
    Using Lemma \ref{lem:flow-inG'-es-et-f*}, let $f'$ be a maximum $(s',t')$-flow in $G'$ such that $e_s$ and $e_t$ carry flow $f^*$. We now define a  $(s,t)$-flow $f$ in $G$ using $f'$ as follows. For each edge $\hat{e}$ of $G$, that is also present in $G'$, we define $f(\hat{e})=f'(\hat{e})$. We assign flow $f(e)=\alpha$. It is easy to verify that $f$ satisfies the conservation of flow and the capacity constraint. Hence $f$ is a valid ($s,t$)-flow in $G$. Moreover, since $f'(e_s)=f'(e_t)=f^*$, and these edges are not present in $G$, it follows that the value of flow leaving $s$
    (likewise entering $t$) in $G$ is $f^*$. Hence, $f$ is a maximum $(s,t)$-flow in $G$ such that $f(e)=\alpha$. So, the maximum flow that $e$ can carry in any maximum $(s,t)$-flow is at least $\alpha$.
\end{proof} 

\begin{lemma} \label{lem: part-2}
    If the maximum flow that $e$ can carry in any maximum $(s,t)$-flow is $\alpha$, the maximum $(s',t')$-flow in graph $G'$ is at least $f^*+\alpha$.
\end{lemma}
\begin{proof}    
    Let $\alpha$ be the maximum flow that $e$ can carry in any maximum $(s,t)$-flow in $G$. Let $f$ be one such maximum $(s,t)$-flow. We shall now demonstrate a $(s',t')$-flow of value $f^*+\alpha$ in $G'$ as follows. 
    For each edge $\hat{e}$ in $G'$, that is also present in $G$, $f'(\hat{e})=f(\hat{e})$. For edges $e_s$ and $e_t$, we define $f'(e_s)=f'(e_t)=f^*$. 
    For edges $e_u$ and $e_v$, we define $f'(e_u)=f'(v)=f(e)$. It is easy to verify that $f'$ satisfies the conservation of flow and the capacity constraint. Hence $f'$ is a valid $(s',t')$-flow in $G'$, and its value if $f^*+\alpha$. Hence the maximum $(s',t')$-flow has value at least $f^*+\alpha$. 
\end{proof}

Lemmas \ref{lem: part-1} and \ref{lem: part-2} together lead to the following lemma immediately. 
\begin{lemma} \label{lem: short - one maxflow computations for value}
    The value of maximum $(s',t')$-flow in graph $G'$ is $f^*+\alpha$ if and only if the maximum flow that 
    $e$ can carry in any maximum $(s,t)$-flow is $\alpha$. 
\end{lemma}
   
It requires ${\mathcal O}(1)$ time to construct graph $G'$ from $G$. Observe that computing a maximum $(s',t')$-flow in $G'$ is asymptotically the same as computing a maximum $(s,t)$-flow in $G$. Therefore, it follows from Lemma \ref{lem: short - one maxflow computations for value} that, given any maximum $(s,t)$-flow in $G$, we can compute the maximum flow that any edge can carry in $G$ using just one maximum $(s,t)$-flow computation. This completes the proof of the following theorem.  %to report only the value of the maximum $(s,t)$-flow that the edge can carry in any maximum $(s,ṭ)$-flow can be reported using two maximum $(s,t)$-flow computations -- one on graph $G$ and another one on graph $G'$. 
%This completes the proof of Theorem \ref{thm : maxflow at an edge}.

\begin{theorem} \label{thm : maxflow at an edge}
     For any directed weighted graph $G$ on $n$ vertices and $m$ edges with a designated source vertex $s$ and a designated sink vertex $t$, there is an algorithm that, given any maximum $(s,t)$-flow in $G$ and an edge $e$, can compute the maximum flow that edge $e$ can carry in any maximum $(s,t)$-flow using one maximum $(s,t)$-flow computation.
    %computations to compute a maximum $(s,t)$-flow $f$ such that $f(e)\ge f'(e)$ for every maximum $(s,t)$-flow $f'$ in $G$. 
\end{theorem}
\subsection{Applications of our algorithm}
 We present the following important applications of our algorithm stated in Theorem \ref{thm : computing all vital edges}.
    \subsubsection{Computing mincut cover and Preprocessing time for ${\mathcal T}_{vit}(G)$} 
We construct data structure ${\mathcal T}_{vit}(G)$ from Theorem \ref{thm : reporting value} and, as a byproduct, we show that it also computes a mincut cover for all vital edges. 
%To construct ${\mathcal T}_{vit}(G)$, we need to design a compact structure that stores a mincut for all vital edges. Theorem \ref{thm : computing all vital edges} can compute all the loose edges of $G$. Since there are at most $n-2$ loose edges (Theorem \ref{thm : n-1 loose edges}), we can compute a mincut for each loose edge and its capacity using ${\mathcal O}(n)$ maximum $(s,t)$-flow computations. The remaining vital edges are all tight edges. For tight edges, exploiting Theorem \ref{thm : tight edge property}, we use the \textit{Ancestor tree} given by Cheng and Hu \cite{DBLP:journals/anor/ChengH91} to compute a mincut for each tight edge and its capacity. It requires ${\mathcal O}(n)$ maximum $(s,t)$-flow computations (discussed before Theorem \ref{thm : tight edge computation}). 
It follows from the design of algorithm in Theorem \ref{thm : computing all vital edges} that it can compute a mincut and its capacity for every vital edge of $G$.
Now, given a mincut for all vital edges, we sort all the vital edges in the increasing order of the capacities of mincuts for them. Let ${\mathbb S}$ be the sorted list of vital edges.

We now present an iterative algorithm to compute tree ${\mathcal T}_{vit}(G)$ 
using the sorted list ${\mathbb S}$ of all vital edges. Similar to Algorithm \ref{alg : hierarchy tree}, we begin by creating a root node $r$ and assign vertex set $V$ to $r$. Let us consider the first edge $e_1$ from the sorted list ${\mathbb S}$, which must be contributing to an $(s,t)$-mincut. The mincut $C_1$ for edge $e_1$ partitions set $V$ into two subsets -- $C_1\cap V$ and $\overline{C_1}\cap V$. We create two children of the root node $r$ and assign vertex set $C_1\cap V$ to the left child and vertex set $\overline{C_1}\cap V$ to the right child. At any $i^{th}$ step, we select an edge $e_i$ from the sorted list ${\mathbb S}$. If both endpoints of edge $e_i$ belong to two different leaf nodes in the tree, then it follows from the construction that a mincut for edge $e_i$ is already considered. %, and hence we do nothing.
Suppose both endpoints of $e_i$ belong to the same leaf node, say $\mu$, with vertex set $U$ associated with $\mu$. We create two children of $\mu$. Then associate vertex set $C_i\cap U$ to the left child and $\overline{C_i}\cap U$ to the right child of $\mu$, where $C_i$ is a mincut for edge $e_i$. The process is repeated for all edges in the sorted list ${\mathbb S}$.

It is shown that the hierarchy tree ${\mathcal T}_{vit}(G)$ has at most $n-1$ internal nodes. Moreover, at every step in the construction, we compute two intersections of vertex sets if a new internal node is created. Therefore, given a mincut and its capacity for every vital edge, the process of constructing ${\mathcal T}_{vit}(G)$ requires ${\mathcal O}(m\log{n}+n^2)$ time. Since the best-known algorithm for computing a maximum $(s,t)$-flow is ${\mathcal O}(mn\log_{m/n \log{n}}{n})$ \cite{DBLP:journals/jal/KingRT94}, it leads to the following theorem. 
\begin{theorem} \label{thm : t vit G computation}
    For any directed weighted graph $G$ on $n$ vertices and $m$ edges with a designated source vertex $s$ and a designated sink vertex $t$, there is an algorithm that can construct tree ${\mathcal T}_{vit}(G)$  using ${\mathcal O}(n)$ maximum $(s,t)$-flow computations. It also computes a mincut cover for all vital edges.
\end{theorem}
%Observe that, using Theorem \ref{thm : t vit G computation}, we can compute a mincut cover for all vital edges. 

\subsubsection{Computing $k$ most vital edges}
Given tree ${\mathcal T}_{vit}(G)$, the vitality of each edge can be computed in ${\mathcal O}(m)$ time. Let us consider a \textsc{MaxHeap} data structure ${\mathcal H}$ of size ${\mathcal O}(m)$ that stores the vitality of every edge. We can construct the \textsc{MaxHeap} ${\mathcal H}$ in ${\mathcal O}(m)$ time. Finally, to report the $k$ most vital edges, we remove the maximum element of the heap and update the heap $k$ times. Therefore, given the vitality of every edge, the process of reporting the $k$ most vital edges takes ${\mathcal O}(m+k\log{k})$ time. This completes the proof of the following theorem.
\begin{theorem}
    For any directed weighted graph $G$ on $n$ vertices and $m$ edges with a designated source vertex $s$ and a designated sink vertex $t$, there is an algorithm that, given any integer $k\in [m]$, can compute $k$ most vital edge of $G$ using ${\mathcal O}(n)$ maximum $(s,t)$-flow computations. 
\end{theorem}
\section{Compact Structures for All Mincuts for All Vital Edges}
\label{sec: compact-structures}
The set containing all $(s,t)$-mincuts satisfies two important properties -- $(i)$ it is closed under both intersection and union and $(ii)$ an edge contributing to an $(s,t)$-mincut can never be an incoming edge to another $(s,t)$-mincut. These two properties are exploited crucially in the design of DAG structures for compactly storing and characterizing all $(s,t)$-mincuts \cite{DBLP:journals/mp/PicardQ80, baswana2023minimum+} using $1$-transversal cuts. 
Unfortunately, none of these properties holds  for the set of all mincuts for all vital edges
(refer to Figure \ref{fig : gamma}$(ii)$).   
This makes the problem of designing compact structures for mincuts for all vital edges challenging.  
We present two structures for storing and characterizing all mincuts for all vital edges -- one is a single DAG that provides a \textit{partial} characterization and the other consists of a set of ${\mathcal O}(n)$ DAGs that provides a complete characterization.

\subsection{An O(m) Space DAG for Partial Characterization} \label{sec: partial characterization}
% {\color{purple}
%     We first show that the approaches taken in \cite{DBLP:journals/mp/PicardQ80, baswana2023minimum+} are not sufficient for designing a compact structure for all mincuts for all vital edges. In particular, the resulting graph, say $Q'(G)$, obtained from $G$ after applying the techniques from \cite{DBLP:journals/mp/PicardQ80, baswana2023minimum+},
% is still not acyclic. Moreover, a mincut for a vital edge can have {\em unbounded transversality} (refer to the full version) in $Q'(G)$.
% It turns out that the source of this problem is the set of all edges that are contributing to a mincut for a vital edge as well as incoming to a mincut for another vital edge. The set of such edges is denoted by ${\Gamma}${\em-edges}. Exploiting GenFlowCut property crucially, we show that all edges in $\Gamma$-edges are nonvital. Thereafter, counter intuitive to Note \ref{note : removal of gamma edges},
% we show that the graph obtained after the removal of all $\Gamma$-edges from $G$ still preserves all $(s,t)$-mincuts of $G$,
% and provides a partial characterization for all mincuts for all vital edges as follows. }

In this section, we present a DAG structure that provides a partial characterization using $1$-transversal cuts.
%Old
    % We begin by defining the following relation on the vertex set of $G$ for any given set of cuts ${\mathcal A}$ of graph $G$.
%     \begin{definition}[Relation $R_{\mathcal A}$]
% Any two vertices $x,y\in V$ are said to be related by $R_{\mathcal A}$ if and only if $x$ and $y$ are not separated by any cut in ${\mathcal A}$. %$(s,t)$-cut from ${\mathcal C}(E_{vit})$.
% \end{definition}
  We begin by defining the following relation on any given vertex set $S\subseteq V$ and any given set of cuts ${\mathcal A}$ of graph $G$.
\begin{definition}[Relation ${\mathcal R}_S({\mathcal A})$] \label{def : relation R}
Let $S\subseteq V$ and ${\mathcal A}$ be any set of cuts. Any two vertices $x,y\in S$ are said to be related by ${\mathcal R}_S({\mathcal A})$ if and only if $x$ and $y$ are not separated by any cut in ${\mathcal A}$. %$(s,t)$-cut from ${\mathcal C}(E_{vit})$.
\end{definition}
% It is easy to show that ${\mathcal R}_S({\mathcal A})$ defines an equivalence relation on the vertex set; and the corresponding equivalence classes form a disjoint partition of the vertex set. We construct a quotient graph $Q(G)$ from $G$ by contracting each equivalence class into a single node.  Without causing any ambiguity, in $Q(G)$, we denote the node containing source vertex $s$ by $s$ and the node containing sink vertex $t$ by $t$. The following theorem is immediate from the construction of $Q(G)$.
It is easy to show that ${\mathcal R}_S({\mathcal A})$ defines an equivalence relation on the vertex set $S$; and the corresponding equivalence classes form a disjoint partition of $S$. 
For $S=V$, we construct a quotient graph $Q(G)$ from $G$ by contracting each equivalence class of relation ${\mathcal R}_V({\mathcal A})$ into a single node.  Without causing any ambiguity, in $Q(G)$, we denote the node containing source vertex $s$ by $s$ and the node containing sink vertex $t$ by $t$. The following theorem is immediate from the construction of $Q(G)$.
\begin{theorem}\label{thm : quotient graph for cuts}
Let $G=(V,E)$ be a directed weighted graph on $m=|E|$ edges, %with a designated source vertex $s$ and a designated sink vertex $t$. 
and let ${\mathcal A}$ be any set of cuts in $G$. There exists an ${\mathcal O}(m)$ space quotient graph $Q(G)$ of $G$ %that preserves all cuts from ${\mathcal A}$. 
such that $C\in {\mathcal A}$  if and only if $C$ is a cut in $Q(G)$ with the same capacity.
\end{theorem}
For the set of all mincuts for all vital edges of $G$, the following theorem immediately follows from Theorem \ref{thm : quotient graph for cuts}.
\begin{theorem}\label{thm : quotient graph}
For a directed weighted graph $G=(V,E)$ on $m=|E|$ edges with a designated source vertex $s$ and a designated sink vertex $t$, there exists an ${\mathcal O}(m)$ space quotient graph $Q_{vit}(G)$ that preserves all mincuts for each vital edge $e$ in $G$ such that 
$C$ is a mincut for $e$ in $G$ if and only if $C$ is a mincut for $e$ in $Q_{vit}(G)$.
\end{theorem}

% We begin by defining the following relation on the vertex set of $G$ for the given set of edges $E_{vit}$.
% \begin{definition}[Relation $R$]
% Any two vertices $x,y\in V$ are said to be related by $R$ if and only if $x$ and $y$ are not separated by any mincut for any vital edge. %$(s,t)$-cut from ${\mathcal C}(E_{vit})$.
% \end{definition}

% It is easy to show that $R$ defines an equivalence relation on the vertex set; and the corresponding equivalence classes  form a disjoint partition of the vertex set. We construct a quotient graph $Q(G)$ from $G$ by contracting each equivalence class into a single node.  Without causing any ambiguity, in $Q(G)$, we denote the node containing source vertex $s$ by $s$ and the node containing sink vertex $t$ by $t$. The following theorem is immediate from the construction of $Q(G)$.

% % The node of $Q(G)$ containing source vertex $s$ is denoted by $S$ and the node containing sink vertex $t$ is denoted by $T$. We call an ($S,T$)-cut of $Q(G)$ as ($s,t$)-cut without causing any ambiguity. The following theorem is immediate from the construction of $Q(G)$. %Hence, the following theorem is immediate.
% \begin{theorem}\label{thm : quotient graph}
% For a directed weighted graph $G=(V,E)$ on $m=|E|$ edges with a designated source vertex $s$ and a designated sink vertex $t$, there exists an ${\mathcal O}(m)$ space quotient graph $Q(G)$ that preserves all mincuts for all edges in $E_{vit}$ such that $C$ is a mincut for edge $e\in E_{vit}$ in $G$ if and only if $C$ is a mincut for edge $e$ in $Q(G)$.
% \end{theorem}

Let $E_{min}$ be the set of all edges that contribute to an $(s,t)$-mincut.
  It is shown in \cite{baswana2023minimum+} that a compact structure (alternative to ${\mathcal D}_{PQ}(G)$ in \cite{DBLP:journals/mp/PicardQ80}) for all mincuts for all edges from $E_{min}$  can be designed as follows.\\

%It is shown in \cite{DBLP:conf/icalp/BaswanaBP22} that the structure ${\mathcal D}_{PQ}(G)$ can be obtained by reversing all the edges of the following graph. 

%For $(s,t)$-mincuts, it is shown in \cite{DBLP:conf/icalp/BaswanaBP22} that the compact structures \cite{DBLP:journals/mp/PicardQ80, DBLP:conf/icalp/BaswanaBP22} can be constructed as follows. 
\noindent
${\mathcal D}_{\lambda}(G)$: {\em In the quotient graph} {\text(Theorem \ref{thm : quotient graph for cuts})} {\em obtained from $G$ for all mincuts for all edges from set $E_{min}$, flip the direction of each edge that is not contributing to an $(s,t)$-mincut.}\\

The technique of flipping each noncontributing edge ensures that ${\mathcal D}_{\lambda}(G)$ is acyclic and each $(s,t)$-mincut is a $1$-transversal cut in ${\mathcal D}_{\lambda}(G)$.
% The structures in \cite{DBLP:journals/mp/PicardQ80, DBLP:conf/icalp/BaswanaBP22} are directed acyclic graphs, and each $(s,t)$-mincut is characterized as a $1$-transversal cut. In order to design a compact structure for storing and characterizing all mincuts for all vital edges, it is a promising approach to first build the quotient graph $Q(G)$ (Theorem \ref{thm : quotient graph}) from graph $G$ based on the set $E_{vit}$. This is because the quotient graph $Q(G)$ keeps all the vital edges and their mincuts intact. 
Applying the same technique in graph $Q_{vit}(G)$ from Theorem \ref{thm : quotient graph}, %for $E_{vit}$, 
%, similar to the case of $(s,t)$-mincuts, of flipping all the edges that do not contribute to any mincut for any vital edge, 
we obtain a graph $Q'(G)$ from $Q_{vit}(G)$. However, unlike ${\mathcal D}_{\lambda}(G)$, apparently $Q'(G)$ does not serve our objective as shown in the following theorem  %It is observed that some properties that hold for $(s,t)$-mincuts no longer hold for mincuts for vital edges, as explained in Section \ref{sec : Section 1 structure}. However, we are still considering only the mincuts for each vital edge. Thus, it seems that it might be possible to establish $\ell$-transversality, for some constant $\ell$, of each mincut for every vital edge. {\color{red} Thus, one would expect that by exploiting the fact that these are $(s,t)$-cuts of the least capacities for vital edges, it might be possible to establish at least $\ell$-transversality, for some constant $\ell>0$, in $Q'(G)$ {\color{red} for each mincut for every vital edge}}.
%Unfortunately, it turns out that there are graphs $H$ such that {\color{blue}, $Q'(H)$ is {\color{blue} still non-acyclic and does not provide bounded transversality for characterizing all mincuts for all vital edges in $E_{vit}$.  %} there exists a mincut for a vital edge in $G$ that appears as a $\Omega(n)$-transversal cut in $Q'(G)$, } {\color{red} 
 %the edge-set of a mincut for an edge intersects a path in $Q'(H)$ at least $n$ times as stated in the following theorem.
 (Proof is in Appendix \ref{app : unbounded transversality}). 
 \begin{theorem} \label{thm : n transversality}
    There exists a graph $H$, on $n$ vertices with a designated source vertex $s$ and a designated sink vertex $t$, such that $Q'(H)$ contains cycles 
    and there is a mincut for a vital edge in $H$ that appears as an $\Omega(n)$-transversal cut in $Q'(H)$. 
\end{theorem}
 %{\color{red} There may exist a non-zero number of edges that are contributing to a mincut for a vital edge as well as incoming to a mincut for another vital edge.  We denote the set of such edges by $\Gamma$-edges.} 
 We observe that the source of $\Omega(n)$-transversality in Theorem \ref{thm : n transversality} is the presence of edges that are contributing to a mincut for a vital edge as well as incoming to another mincut for a vital edge, which are called \textit{$\Gamma$-edges}. %(defined in Section \ref{sec : Section 1 structure}). 
 So, can we remove all the $\Gamma$-edges? 
  At first glance, it seems that their removal might lead to a reduction in the capacity of $(s,t)$-mincut as stated in Note \ref{note : removal of gamma edges}. Interestingly, we establish the following lemma to ensure that the capacity of $(s,t)$-mincut remains unchanged even after the removal of all $\Gamma$ edges. %{\color{red} Interestingly, we show that we can indeed remove them while preserving the capacity of $(s,t)$-mincut as stated in the following lemma.}   

\begin{lemma}\label{lem : gamma is irrelevant}
    Let $C_1$ and $C_2$ be mincuts for vital edges $e_1$ and $e_2$, respectively. Let $\Gamma_1$ and $\Gamma_2$ denote the set of edges that are from $C_2\setminus C_1$ to $C_1\setminus C_2$ and from $C_1\setminus C_2$ to $C_2\setminus C_1$ respectively. Then, the removal of all edges belonging to $\Gamma_1$ and $\Gamma_2$ does not change the capacity of $(s,t)$-mincut. 
\end{lemma}

\begin{proof}
\begin{figure}
 \centering
    \includegraphics[width=0.7\textwidth]{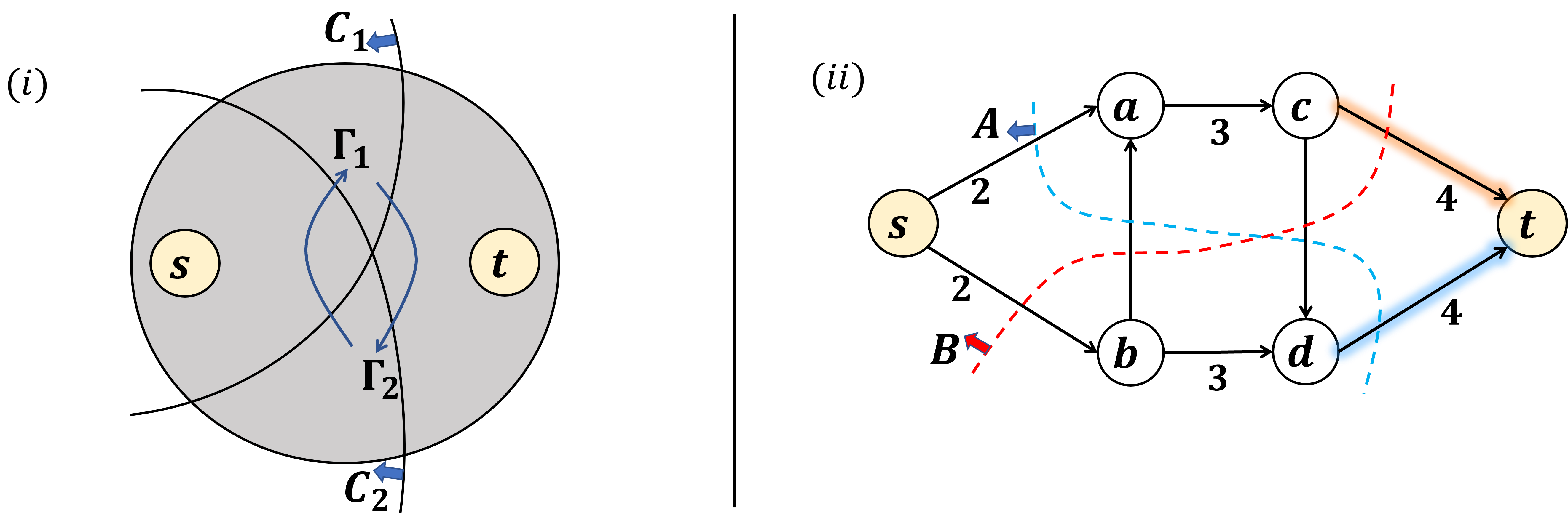} 
  \caption{($i$) $\Gamma$-edges between a pair of mincuts $C_1$ and $C_2$ for a pair of vital edges. ($ii$) mincuts $A$ and $B$ for vital edges $(c,t)$ and $(d,t)$, respectively, are not closed under union. Moreover, edge $(b,a)$ (likewise edge $(c,d)$) contributes to $A$ (likewise to $B$) and is incoming to $B$ (likewise to $A$).} \label{fig : gamma} 
\end{figure}
   If one of $C_1$ and $C_2$ is a proper subset of the other, then $\Gamma_1=\Gamma_2=\emptyset$. Therefore, we need to consider only the case when $C_1$ and $C_2$ are crossing, that is, $C_1\cap C_2$, $\overline{C_1\cup C_2}$, $C_1\setminus C_2$, and $C_2\setminus C_1$ are nonempty. 

      First we show that $e_1\notin \Gamma_2$ and $e_2\notin \Gamma_1$. Without loss of generality, assume that $e_1\in \Gamma_2$. Observe that $e_1$ is an incoming edge of $C_2$. Since $C_2$ is a mincut for vital edge $e_2$, therefore, it follows from Lemma \ref{lem:incoming-edge-irrelevant} that $e_1$ is a nonvital edge, a contradiction.   
      
      Suppose exactly one of $\Gamma_1$ and $\Gamma_2$ is nonempty. Without loss of generality, assume that $\Gamma_1=\emptyset$. Since $C_2$ is a mincut for vital edge $e_2$, therefore, it follows from Lemma \ref{lem:incoming-edge-irrelevant} that removal of all edges belonging to $\Gamma_2$ does not decrease the capacity of $(s,t)$-mincut.  
      
      Suppose both $\Gamma_1$ and $\Gamma_2$ are nonempty. We refer to Figure \ref{fig : gamma}($i$) for illustration.
     Let $G_1$ be the graph obtained after removing all edges from $\Gamma_2$. Since all edges of $\Gamma_2$ are incoming to $C_2$, it follows from %{\color{red} Theorem \ref{thm : a special assignment of flow}($1$)} {\color{blue} 
     Lemma \ref{lem:incoming-edge-irrelevant} that the capacity of $(s,t)$-mincut in $G_1$ is the same as that in $G$. Therefore, for each $(s,t)$-cut $C$ in $G_1$, $c(C)\ge f^*$. Hence, $C_1$ in $G_1$ has a capacity at least $f^*$. 
     It follows from Lemma \ref{lem:vital-edge-remains-vital} that $e_1$ remains vital and $C_1$ continues to remain a mincut for $e_1$ in $G_1$.  However, the capacity of $(s,t)$-cut $C_1$ in $G_1$ is strictly less than the capacity of $(s,t)$-cut $C_1$ in $G$ since all edges of $\Gamma_2$ are contributing edges of $C_1$ in $G$. 
     Therefore, $w_{min}(e_1)$ in $G_1$, denoted by $w_{min}$, is strictly greater than the $w_{min}(e_1)$ in $G$. We now obtain a graph $G_2$ from $G_1$ after reducing the capacity of edge $e_1$ from $w(e_1)$ to $w_{min}$. Observe that in graph $G_2$, it follows from Theorem \ref{thm : a special assignment of flow}$(2)$ that the capacity of $(s,t)$-mincut is the same as the capacity of $(s,t)$-mincut in $G_1$, that is $f^*$. Therefore, $C_1$ is an $(s,t)$-mincut and all edges of $\Gamma_1$ are incoming to it. Hence, they are nonvital edges. It follows from Lemma \ref{lem:incoming-edge-irrelevant} that their removal does not decrease the $(s,t)$-mincut capacity.  
\end{proof}

Now exploiting Lemma \ref{lem : gamma is irrelevant}, we construct a graph ${\mathcal D}_{vit}(G)$  for storing and partially characterizing all mincuts for all vital edges of graph $G$ as follows.

\paragraph*{Construction of ${\mathcal D}_{vit}(G):$} 
 Graph ${\mathcal D}_{vit}(G)$ is obtained from graph $Q_{vit}(G)$ (Theorem \ref{thm : quotient graph}) by executing the following two steps in order.
 \begin{enumerate}
    \item For each edge $(u,v)$ that belongs to $\Gamma$-edges, remove it from $Q_{vit}(G)$. %each edge $(u,v)$ from $Q_{vit}(G)$ if it belongs to the set of $\Gamma$-edges.
%For each edge $(u,v)$ in $Q_{vit}(G)$, remove it from $Q_{vit}(G)$ if it belongs to the set of $\Gamma$-edges. % if there is a pair of  {\color{red} vital $(s,t)$-cuts} {\color{blue} mincuts $\{C,C'\}$ for a pair of vital edges} such that $u\in C\setminus C'$ and $v\in C'\setminus C$ or vice versa, then remove edge $(u,v)$ from $Q_{vit}(G)$.
\item For each edge $(u,v)$ that is incoming to a mincut for a vital edge in $Q_{vit}(G)$, replace it with an edge $(v,u)$ of capacity $0$. %if it is an incoming edge to a mincut for a vital edge. %$$e$ {\color{red}$e\in E_{vit}$}, then remove edge $(v,u)$ from $Q_{vit}(G)$.
 \end{enumerate} 

% For each edge $(u,v)$ in $Q_{vit}(G)$, remove it from $Q_{vit}(G)$ if it belongs to the set of $\Gamma$-edges or it is an incoming edge to a mincut for a vital edge. 

% {\color{red} The graph ${\mathcal D}_{vit}(G)$ is obtained from $Q_{vit}(G)$ by executing the following two steps. 
% \begin{enumerate}
%     \item 
% For each edge $(u,v)$ in $Q_{vit}(G)$, remove it from $Q_{vit}(G)$ if it belongs to the set of $\Gamma$-edges. % if there is a pair of  {\color{red} vital $(s,t)$-cuts} {\color{blue} mincuts $\{C,C'\}$ for a pair of vital edges} such that $u\in C\setminus C'$ and $v\in C'\setminus C$ or vice versa, then remove edge $(u,v)$ from $Q_{vit}(G)$.
% \item For each edge $(u,v)$ in $Q_{vit}(G)$, remove it from $Q_{vit}(G)$ if it is an incoming edge to a mincut for a vital edge. %$$e$ {\color{red}$e\in E_{vit}$}, then remove edge $(v,u)$ from $Q_{vit}(G)$.
%  \end{enumerate} }
\noindent
 The following definition of cut is to be used in establishing properties of ${\mathcal D}_{vit}(G)$.
\begin{definition} [relevant cut for an edge] \label{def : relevant for an edge}
    An $(s,t)$-cut $C$ is said to be a relevant cut for edge $e$ if $e$ contributes to $C$ and $c(C)-w(e)<f^*$.  %Moreover,  %, and $c(C)\le c(C')$ for any other $(s,t)$-cut $C'$ in which $e$ appears as a contributing edge.
\end{definition}
By Definition \ref{def : relevant and mincut for an edge} and Definition \ref{def:relevant-edges}, an edge $e$ is vital if and only if there exists a relevant cut for edge $e$. We state the following lemma %(proof in Appendix \ref{app : property of dvit}) 
about ${\mathcal D}_{vit}(G)$.
\begin{lemma} \label{lem : property of dvit}
    Graph ${\mathcal D}_{vit}(G)$ for a directed weighted graph $G$ satisfies the  following properties.
    \begin{enumerate}
        \item A vital edge of $G$ appears as a vital edge in ${\mathcal D}_{vit}(G)$.
        \item Capacity of $(s,t)$-mincut in $G$ is the same as the capacity of $(s,t)$-mincut in  ${\mathcal D}_{vit}(G)$.
        \item  Each mincut for a vital edge $e$ in $G$ appears as a relevant cut for vital edge $e$ in ${\mathcal D}_{vit}(G)$. 
    \end{enumerate}
% \begin{enumerate}
%         \item a vital edge of $G$ appears as a vital edge in ${\mathcal D}_{vit}(G)$.
%         \item capacity of $(s,t)$-mincut in $G$ is the same as the capacity of the $(s,t)$-mincut in  ${\mathcal D}_{vit}(G)$.
%         \item Each mincut for a vital edge $e$ in $G$ appears as a relevant cut for vital edge $e$ in ${\mathcal D}_{vit}(G)$. 
%     \end{enumerate}
\end{lemma}
\begin{proof}
    Lemma \ref{lem : gamma is irrelevant} states that the removal of all the $\Gamma$-edges  in Step 1
    in the construction of ${\mathcal D}_{vit}(G)$ does not decrease the capacity of $(s,t)$-mincut. Observe that, after removal of all $\Gamma$-edges, if there is an edge that is an incoming edge of a mincut for a vital edge, then it cannot be a contributing edge of a mincut for any other vital edge. Moreover, the removal of an incoming edge and adding an edge of capacity $0$ to an $(s,t)$-cut does not change its capacity. So, the capacity of $(s,t)$-mincut does not change after executing Step 2. Therefore, the capacity of $(s,t)$-mincut in $G$ is also the capacity of $(s,t)$-mincut in ${\mathcal D}_{vit}(G)$. Now, it follows from Lemma \ref{lem:vital-edge-remains-vital} and proof of Lemma \ref{lem : gamma is irrelevant} that a mincut for a vital edge in $G$ remains a relevant cut for vital edge $e$ in ${\mathcal D}_{vit}(G)$. Moreover, $Q_{vit}(G)$ preserves all vital edges and, in the construction of ${\mathcal D}_{vit}(G)$, we only removed nonvital  edges from $Q_{vit}(G)$. % that are nonvital {\color{red}(do we need to refer to construction of ${\mathcal D}_{vit}$ here ?)}, and . 
    Hence all vital edges of $G$ are also preserved in ${\mathcal D}_{vit}(G)$. 
\end{proof}

\begin{remark}
    It is easy to observe that all the properties mentioned in Lemma \ref{lem : property of dvit} are also satisfied in the quotient graph $Q_{vit}(G)$. It is interesting to see that, even after removing $\Gamma$-edges and flipping certain edges in $Q_{vit}(G)$, each of the properties is preserved in ${\mathcal D}_{vit}(G)$ as well. 
\end{remark}

We now state the following lemma that is immediate from the construction of ${\mathcal D}_{vit}$.
\begin{lemma} \label{lem : empty set gamma}
    For any edge $(u,v)$ in ${\mathcal D}_{vit}(G)$,
    \begin{enumerate}
        \item There exists a mincut $C$ for a vital edge $e$ such that $u\in C$ and $v\in \overline{C}$.
        \item There does not exist any pair of mincuts $\{C_1,C_2\}$  for any pair of vital edges $\{e_1,e_2\}$ such that edge $(u,v)$ lies between $C_1\setminus C_2$ and $C_2\setminus C_1$. 
    \end{enumerate}
\end{lemma}    
    %For any pair of mincuts $\{C_1,C_2\}$  for a pair of vital edges $\{e_1,e_2\}$ in graph $G$, there is no edge that lies between $C_1\setminus C_2$ and $C_2\setminus C_1$ in graph ${\mathcal D}_{vit}(G)$. 

We establish the acyclicity property for the structure ${\mathcal D}_{vit}(G)$ using Lemma \ref{lem : empty set gamma} as follows. 

\begin{lemma}\label{lem : acyclicity}
    ${\mathcal D}_{vit}(G)$ is a directed acyclic graph.
\end{lemma}
\begin{proof}
    Suppose there is a cycle $O=\langle v_0,v_1,\ldots,v_k,v_0\rangle$ in ${\mathcal D}_{vit}(G)$. Let $(u,v)$ be an edge in ${O}$. For the edge $(u,v)$, it follows from Lemma \ref{lem : empty set gamma}(1) that there is a mincut $C$ for an edge $e\in E_{vit}$ such that $u\in C$ and $v\in \overline{C}$. %{\color{red}(how ? which step in the construction ensures it -- the quotient graph defined by equivalence classes.)}. 
    Observe that the edge-set of $C$ must intersect cycle $O$ at an edge, say $(x,y)$, such that $(x,y)$ is an incoming edge of $C$, that is, $y\in C$ and $x\in \overline{C}$. It, again, follows from Lemma \ref{lem : empty set gamma}(1) that there is a mincut $C'$ for a vital edge such that $x\in C'$ and $y\in \overline{C'}$. Therefore, edge $(x,y)$ is a contributing edge of $C'$ and an incoming edge of $C$,  which contradicts Lemma \ref{lem : empty set gamma}(2).  
\end{proof}
For a graph $G$, not only ${\mathcal D}_{vit}(G)$ is structurally rich but also provides an interesting characteristic of each mincut for every vital edge of $G$ as follows. %{\color{red}(Proof is in Appendix \ref{app : transversality})}. 
\begin{lemma} \label{lem : transversality}
    Each mincut for a vital edge in $G$ appears as a $1$-transversal cut in ${\mathcal D}_{vit}(G)$.
\end{lemma}
\begin{proof}
    Let $C$ be a mincut for a vital edge $e$. Suppose edge-set of $C$ intersects an $(s,t)$-path $P=\langle v_0, v_1, \ldots, v_k \rangle$ at least thrice. This implies that there is at least one edge $(u,v)$ in path $P$ that is an incoming edge of $C$, that is $v\in C$ and $u\notin C$. It follows from Lemma \ref{lem : empty set gamma}(1) that there is a mincut $C'$ for a vital edge such that $u\in C'$ and $v\in \overline{C'}$, Therefore,edge $(u,v)$ is a contributing edge of $C'$ and an incoming edge of $C$,  which contradicts Lemma \ref{lem : empty set gamma}(2).  
\end{proof}
Lemma \ref{lem : property of dvit}, Lemma \ref{lem : acyclicity}, and Lemma \ref{lem : transversality} imply Theorem \ref{thm : dag and 1 transversal}.
 \subsubsection{For (s,t)-mincuts, Existing DAGs vs %DAGs in \cite{DBLP:journals/mp/PicardQ80, baswana2023minimum+}  vs 
 the DAG in Theorem \ref{thm : dag and 1 transversal}} \label{sec : comparison} 
In this section, we show that existing structures for $(s,t)$-mincuts \cite{DBLP:journals/mp/PicardQ80, baswana2023minimum+} are just a special case of DAG ${\mathcal D}_{vit}(G)$ in Theorem \ref{thm : dag and 1 transversal}. Let $E_{min}\subseteq E_{vit}$ be the set consisting of all vital edges that are contributing to $(s,t)$-mincuts. For the set of all mincuts for each edge in $E_{min}$, the existing structures for $(s,t)$-mincuts \cite{DBLP:journals/mp/PicardQ80, baswana2023minimum+} are sufficient to store and characterize them in terms of $1$-transversal cuts. Suppose we construct DAG ${\mathcal D}_{vit}(G)$ in Theorem \ref{thm : dag and 1 transversal} for the set of vital edges $E_{min}$. We denote this graph by ${\mathcal D}_{vit}^{f^*}(G)$. We now show that not only DAG ${\mathcal D}_{vit}^{f^*}(G)$ stores and characterizes all mincuts for each edge in $E_{min}$ but also provides additional characteristics (refer to Table \ref{tab : Dvit vs Dpq and Dlambda}).  
 
 \begin{figure}[ht]
 \centering
    \includegraphics[width=\textwidth]{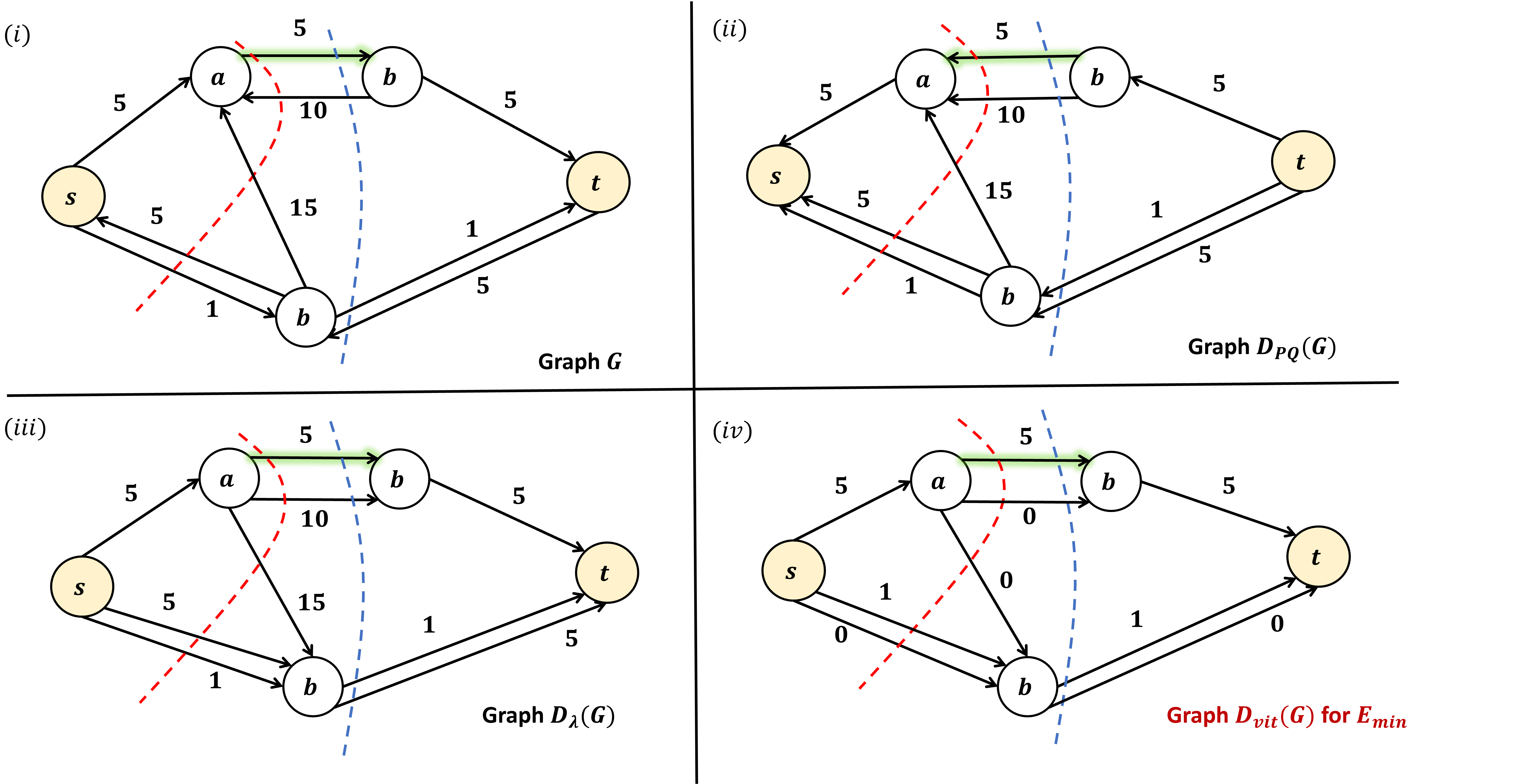} 
 \caption{For example graph $G$ in $(i)$, the existing structures for $(s,t)$-mincuts ${\mathcal D}_{PQ}(G)$  \cite{DBLP:journals/mp/PicardQ80} and ${\mathcal D}_{\lambda}(G)$ \cite{baswana2023minimum+} are represented in $(ii)$ and $(iii)$ respectively. For set $E_{min}$, structure ${\mathcal D}_{vit}(G)$ from Theorem \ref{thm : dag and 1 transversal} is represented in $(iv)$. Observe that the capacity of $(s,t)$-mincut in $G$ is $6$, which is preserved only in ${\mathcal D}_{vit}(G)$. Additionally, each $(s,t)$-mincut remains a relevant cut only in ${\mathcal D}_{vit}(G)$.}

  \label{fig : comparison}
\end{figure}

It follows from the construction of ${\mathcal D}_{PQ}(G)$ in \cite{DBLP:journals/mp/PicardQ80} that each edge $e\in E_{min}$ appears in the reverse direction in ${\mathcal D}_{PQ}(G)$. Moreover, each mincut $C$ for any vital edge in $E_{min}$ has capacity exactly zero because all edges are incoming to $C$, as shown in Figure \ref{fig : comparison}. Hence, any mincut for any vital edge in $E_{min}$ is no longer a relevant cut in the DAG ${\mathcal D}_{PQ}(G)$.  
It is shown in \cite{baswana2023minimum+} that the DAG, denoted by ${\mathcal D}_{\lambda}(G)$, obtained by reversing all edges of ${\mathcal D}_{PQ}(G)$ %showed that There is an alternate DAG structure, denoted by ${\mathcal D}_{\lambda}(G)$, given by Baswana, Bhanja, and Pandey \cite{DBLP:conf/icalp/BaswanaBP22} that
also stores all $(s,t)$-mincuts in $G$ and characterizes them as $1$-transversal cuts. ${\mathcal D}_{\lambda}(G)$ keeps the orientation of each edge in $E_{min}$ intact. However, it is not guaranteed that all the mincuts for a vital edge continue to remain relevant cuts for the vital edge in ${\mathcal D}_{\lambda}(G)$; for example,  mincut $\{s,a\}$ for edge $(a,b)$ in Figure \ref{fig : comparison}$(iii)$. Therefore, vital edges in $E_{min}$ might become nonvital edges in ${\mathcal D}_{\lambda}(G)$. %The following lemma is established in \cite{DBLP:conf/icalp/BaswanaBP22} for $(s,t)$-mincuts.
%\begin{lemma} \cite{DBLP:conf/icalp/BaswanaBP22} \label{lem : st mincut gamma edges}
    %For any pair of $(s,t)$-mincuts $C_1$ and $C_2$,

For the set of $(s,t)$-mincuts in graph $G$, there is no edge that belongs to set $\Gamma$-edges (shown in \cite{baswana2023minimum+}).    
%\end{lemma}
 %It follows from Lemma \ref{lem : st mincut gamma edges} that
 Therefore, in Step $1$ of the construction of ${\mathcal D}_{vit}(G)$, there are no edges from $\Gamma$-edges to be removed. In Step $2$ of the construction of ${\mathcal D}_{vit}(G)$, we replace each incoming edge $(u,v)$ of every $(s,t)$-mincut with an edge $(v,u)$ of capacity $0$. If we do not consider the capacity of an edge, then Step 2 actually flips the orientation of edge $(u,v)$. Therefore, it can be observed that the resulting graph ${\mathcal D}_{vit}^{f^*}(G)$, except edge capacities, is the same as DAG ${\mathcal D}_{\lambda}(G)$ in \cite{baswana2023minimum+}. Therefore, an $(s,t)$-cut $C$ in $G$ is a mincut for a vital edge in $E_{min}$ if and only if $C$ is a $1$-transversal cut in ${\mathcal D}_{vit}^{f^*}(G)$. Additionally, it follows from Theorem \ref{thm : dag and 1 transversal} that the capacity of $(s,t)$-mincut in ${\mathcal D}_{vit}^{f^*}(G)$ is $f^*$, each vital edge in $E_{min}$ is preserved, and each mincut for every vital edge in $E_{min}$ appears in ${\mathcal D}_{vit}^{f^*}(G)$ as a relevant cut.

%Interestingly, our DAG ${\mathcal D}_{vit}$ not only preserves all vital edges (instead of only $E_{min}$) but also ensures that every mincut for all vital edges appears as a relevant cut with a 1-transversality property.     

% {\color{red}{\textbf{Is explanation of Figure \ref{fig : comparison} required?}}}

% {\color{red}\textbf{How to use D vit for complete characterization of $(s,t)$-mincuts? explanation required?}}

% \begin{table}[ht]
%     \centering
%     \begin{tabular}{|c|c|c|c|c|c|}
%         \hline
%           \textbf{Structures for}  & Structure & Space & Characterization & $(s,t)$-mincut  & Is Mincut\\
%          \textbf{all $\mathbf{(s,t)}$-mincuts} &  &  & (Complete) & Capacity & a Relevant Cut?\\
%          \hline
%          ${\mathcal D}_{PQ}(G)$ \cite{DBLP:journals/mp/PicardQ80} & DAG & ${\mathcal O}(m)$ & $1$-transversal cuts & $0$ & No \\
%          \hline
%          ${\mathcal D}_{\lambda}(G)$ \cite{DBLP:conf/icalp/BaswanaBP22} & DAG & ${\mathcal O}(m)$ & $1$-transversal cuts & May exceed $f^*$ & No \\
%          \hline
%          $\mathbf{{\mathcal D}_{vit}(G)}$ \textbf{(New) }& \textbf{DAG} & ${\mathbf{\mathcal O}(m)}$ & $\mathbf{1}$\textbf{-transversal cuts} & $\mathbf{f^*}$  & \textbf{Yes}\\
%          \hline         
%     \end{tabular}    
%     \caption{ }
    
%     \label{tab : Dvit vs Dpq}
% \end{table}

\begin{table}[ht]
    \small
    \centering
    \begin{tabular}{|c|c|c|c|}
        \hline
          \textbf{Structures for all}  & ${\mathcal D}_{PQ}(G)$ & ${\mathcal D}_{\lambda}(G)$ & $\mathbf{{\mathcal D}}_{vit}^{f^*}(G)$\\
         \textbf{$\mathbf{(s,t)}$-mincuts vs $\mathbf{{\mathcal D}}_{vit}^{f^*}(G)$ } & \cite{DBLP:journals/mp/PicardQ80} & \cite{baswana2023minimum+} & \textbf{(New)}\\
         \hline
           &  &  & \\
          Structure & DAG & DAG & \textbf{DAG}\\
          \hline
           &  &  & \\
          Space & ${\mathcal O}(m)$ & ${\mathcal O}(m)$ & $\mathbf{{\mathcal O}(m)}$\\
          \hline
          $(s,t)$-mincut &  &  & \\
          Characterization & $1$-transversal cut & $1$-transversal cut  & $\mathbf{1}$-\textbf{transversal cut} \\
          \hline
          Is Capacity of & & & \\
          (s,t)-mincut Preserved?  & No & No & \textbf{Yes}\\
          \hline
          Is an $(s,t)$-mincut & & & \\
          a Relevant cut?  & No & No & \textbf{Yes}\\
           \hline  
          % {\color{red}Is all vital }& & & \\
          % edges Preserved?  & No & No & \textbf{Yes}\\
          % \hline
          % Is all Mincuts for & No, only & No, only & \textbf{Yes, for} \\
          % all vital edges Stored?  & for $E_{min}$ & for $E_{min}$ & $\mathbf{E_{vit}}$ \\
          % \hline
          % Is all Mincuts for &  &  & \textbf{Yes (Partially) by} \\
          % all vital edges Characterized?  & No & No & $1$\textbf{-transversal cut} \\
          % \hline
          % Is Mincut for edge & Not even & Not even & \\
          % $e\in E_{vit}$ a Relevant cut?  & for $E_{min}$ & for $E_{min}$ & \textbf{Yes}\\
          %  \hline         
    \end{tabular}    
    \caption{For the set of all $(s,t)$-mincuts, ${\mathcal D}_{vit}(G)$ in Theorem \ref{thm : dag and 1 transversal} not only contains all the characteristics of existing DAGs (\cite{DBLP:journals/mp/PicardQ80, baswana2023minimum+}) but also has two more additional properties given in last two rows of this table.  Each $(s,t)$-mincut $C$ is a relevant cut for each contributing edge of $C$ in $G$.}
    
    \label{tab : Dvit vs Dpq and Dlambda}
\end{table}

\begin{figure}[ht]
  \begin{center}
    \includegraphics[width=0.3\textwidth]{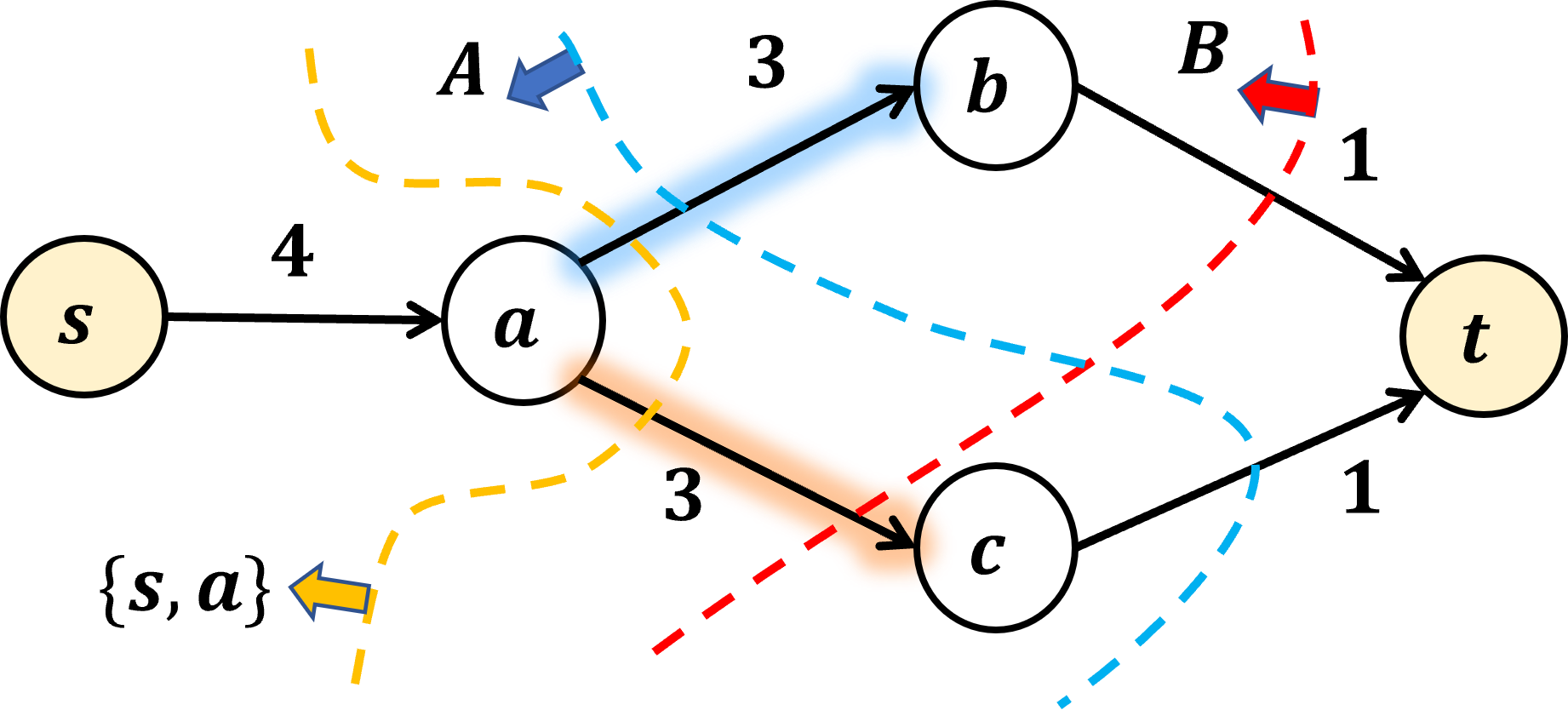}
  \end{center}
  \caption{${\mathcal D}_{vit}(G)$=$G$.  $\{s,a\}$ is a $1$-transversal cut, but not a mincut for any vital edge. 
  % $(iii)$ mincuts $A$ and $B$ for vital edges $(c,t)$ and $(d,t)$, respectively, are not closed under union. Moreover, edge $(b,a)$ (likewise edge $(c,d)$) contributes to $A$ (likewise to $B$) and is incoming to $B$ (likewise to $A$).
  }
  \label{fig : transversal cut not mincut}
\end{figure}
\subsection{An O(mn) Space Structure for Complete Characterization}
Although each mincut for every vital edge is a $1$-transversal cut in ${\mathcal D}_{vit}(G)$, it is not necessary that each $1$-transversal cut in ${\mathcal D}_{vit}(G)$ is a mincut for a vital edge (refer to Figure \ref{fig : transversal cut not mincut}). 
Hence, ${\mathcal D}_{vit}(G)$ could provide only a partial characterization. Let $e=(u,v)$ be any vital edge, and $G_{e}$ denote the graph obtained by adding two infinite weight edges, $(s,u)$ and $(v,t)$, in graph $G$. Observe that all mincuts for $(u,v)$ are compactly stored in the DAG for $(s,t)$-mincuts \cite{DBLP:journals/mp/PicardQ80} built on graph $G_{e}$. We denote this DAG by ${\mathcal D}_{PQ}(G_{e})$.
Recall that Theorem \ref{thm : n-1 cuts}  ensures that there is a set of at most $n-1$ $(s,t)$-cuts that stores a mincut for each vital edge. This set consists of mincuts for at most $n-1$ vital edges selected suitably. So, a promising approach is to store, for each selected edge $e$, the DAG ${\mathcal D}_{PQ}(G_e)$. 
 Overall, this structure occupies ${\mathcal O}(mn)$ space. Unfortunately, this structure fails to store all the mincuts for all vital edges (refer to Appendix \ref{app : limitation of mincut cover}).
Now, it is easy to observe that keeping ${\mathcal D}_{PQ}(G_e)$ for each vital edge $e$ serves as a compact structure for storing all mincuts for all vital edges and characterizing them in terms of $1$-transversal cuts.
 However, this structure may take ${\mathcal O}(m^2)$ space in the worst case. Interestingly, we present an ${\mathcal O}(mn)$ space structure $S_{vit}(G)$. This structure consists of ${\mathcal O}(n)$ DAGs for storing and characterizing all mincuts for all vital edges in terms of $1$-transversal cuts.

In Section \ref{sec: algorithm}, we classified all the vital edges into tight and loose edges. Theorem \ref{thm : n-1 loose edges} ensures that there can be at most $n-2$ loose edges. For each loose edge $e=(u,v)$, we keep a DAG ${\mathcal D}_{PQ}(G_{(u,v)})$. The set of all these $n-2$ DAGs stores all mincuts for all loose edges and provides a complete characterization in terms of 1-transversal cuts. However, it is also true that there can be $\Omega(n^2)$ tight edges in a graph (Theorem \ref{thm : n-1 loose edges}). Therefore, the main challenge arises in designing a compact structure that stores and characterizes all mincuts for all tight edges. Interestingly, it follows from Theorem \ref{thm : tight edge property} that, for a tight edge $(u,v)$, an $(s,t)$-cut of the least capacity that separates $u$ and $v$ is also a mincut for tight edge $(u,v)$. So, we take an approach of building a compact structure that, for every pair of vertices $\{u,v\}$ in $G$, can store and characterize all the $(s,t)$-cuts of the least capacity that separate $u$ and $v$. Therefore, a solution to the following problem would suffice to overcome the challenge.
\begin{problem} \label{prob : compact structure for all pairs}
    Given a directed weighted graph $G$ with a designated source $s$ and a designated sink $t$, build a compact structure that, for each pair of vertices $a$ and $b$ in $G$, stores and characterizes all the $(s,t)$-cuts of the least capacity that separate $a$ and $b$ in $G$.
\end{problem}
% This problem is formally stated in Problem \ref{prob : compact structure for all pairs}. %following problem comes naturally. {\color{red} Redundant}
% \begin{problem}
%     Given a directed weighted graph $G$ with a designated source $s$ and a designated sink $t$, build a compact structure that, for any pair of vertices $a$ and $b$ in $G$, can store and characterize all the $(s,t)$-cuts of the least capacity that separate $a$ and $b$ in $G$.
% \end{problem}
We, in the following, provide a solution to Problem \ref{prob : compact structure for all pairs} in three steps. First, we design a compact structure, denoted by ${\mathcal Y}(S)$, for storing all {\em Steiner $(s,t)$-mincuts} for a given {\em Steiner set $S$}. Thereafter, we design a hierarchy tree that stores a Steiner set $S(\mu)$ at each node $\mu$. Finally, we augment each node $\mu$ of the hierarchy tree with structures ${\mathcal Y}(S(\mu))$. %, which is sufficient for characterizing all mincuts for all tight edges. 

%{\color{red} We, in the following, present a compact structure, denoted by ${\mathcal S}_{vit}(G)$, that stores all mincuts for all vital edges and provides a complete characterization for each of them. Structure ${\mathcal S}_{vit}(G)$ consists of two components. Firstly, for all tight edges, a hierarchy tree that is suitably augmented with ${\mathcal O}(n)$ DAG structures {\color{red}from \cite{DBLP:journals/mp/PicardQ80} for ${\mathcal O}(n)$ different graphs obtained from $G$.} Secondly, for all loose edges, a set of at most $n-1$ DAG structures {\color{red} from \cite{DBLP:journals/mp/PicardQ80}}.}
\subsubsection{A Compact Structure for all Steiner (s,t)-mincuts}
Let $S\subseteq V$ be a Steiner set. An $(s,t)$-cut $C$ is called a {\em Steiner $(s,t)$-cut} if $C\cap S$ and $\overline{C}\cap S$ are nonempty. A Steiner $(s,t)$-cut with the least capacity is a {\em Steiner $(s,t)$-mincut}. Let ${\mathfrak C}_S$ be the set of all Steiner $(s,t)$-mincuts for $S$ and $\lambda_S$ denote the capacity of Steiner $(s,t)$-mincut in $G$. We begin by constructing a graph, denoted by $G(S)$,  that preserves all Steiner $(s,t)$-mincuts in ${\mathfrak C}_S$ as follows.

% \noindent
% \textbf{$G(S)$}: It is the quotient graph of $G$ for set of cuts ${\mathfrak C}_S$. \\
% \textbf{$G(S)$}: It is the quotient graph of $G$ obtained by contracting each maximal set of vertices that are not separated by any Steiner $(s,t)$-mincut in ${\mathfrak C}_S$ into single nodes. \\

 Graph $G(S)$ is a quotient graph of $G$ for the set of cuts ${\mathfrak C}_S$ (Theorem \ref{thm : quotient graph for cuts}). %Without causing any ambiguity, for graph $G(S)$, let us denote the node containing $s$ by $s$, similarly, the node containing $t$ by $t$. 
A node in $G(S)$ containing at least one vertex from Steiner set $S$ is called a Steiner node; otherwise, a non-Steiner node.  It is easy to observe that keeping a DAG ${\mathcal D}_{PQ}(G(S)_{(\mu,\nu)})$, for each pair of Steiner nodes $\{\mu,\nu\}$ in $G(S)$, is sufficient for storing and characterizing all Steiner $(s,t)$-mincuts in ${\mathfrak C}_S$. However, the space complexity of the resulting structure is ${\mathcal O}(|S|^2 m)$. Interestingly, we now show that keeping only ${\mathcal O}(|S|)$ DAGs is sufficient for storing and characterizing all Steiner $(s,t)$-mincuts for $S$. %present an ${\mathcal O}(|S| m)$ space structure for this purpose. 
%
%
%Since each of these DAGs occupies ${\mathcal O}(m)$ space, overall space complexity becomes ${\mathcal O}(|S|^2 m)$. We present in the following an ${\mathcal O}(|S| m)$ space structure for this purpose. 
%We begin with the following fact is immediate from the construction of $G(S)$.
%\begin{fact} \label{fact : separating capacity is lambdaS}

Let $\alpha$ and $\beta$ be any two  Steiner nodes in $G(S)$. It follows from the construction that $(s,t)$-cut of the least capacity that separates $\alpha$ and $\beta$ has capacity exactly $\lambda_S$. Let $C$ be any Steiner $(s,t)$-mincut that separates $\alpha$ and $\beta$, and keeps $\beta$ on the side of $t$. Suppose the source node $s$ in graph $G(S)$ turns out to be a Steiner node. So, the $(s,t)$-cut of the least capacity that separates source node $s$ and node $\beta$ has capacity exactly $\lambda_S$. Therefore, it can be observed that $C$ is preserved in ${\mathcal D}_{PQ}(G(S)_{(s,\beta)})$ and is characterized as a 1-transversal cut. %{\color{red} Let $\alpha$ be any node in $G(S)$ and ${\mathcal C}_{\alpha}$ be the set of all Steiner $(s,t)$-mincuts that keeps $\alpha$ on the side of $t$. It follows from Fact \ref{fact : separating capacity is lambdaS} that all Steiner $(s,t)$-mincuts from set ${\mathcal C}_{\alpha}$ can be stored and characterized using ${\mathcal D}_{PQ}(H_{s,\alpha})$ in terms of $1$-transversal cuts. } 
 Hence, for each Steiner node $\beta$ in $G(S)$, keeping a DAG ${\mathcal D}_{PQ}(G(S)_{(s,\beta)})$ suffices to store and characterize all Steiner $(s,t)$-mincuts in terms of $1$-transversal cuts. Since there are at most $|S|$ Steiner nodes in graph $G(S)$, the number of DAGs required is at most $|S|$. Therefore, the overall space occupied by the structure is ${\mathcal O}(|S|m)$.
%Unfortunately, it is not always true that the source node $s$ is a Steiner node in $G(S)$. %The property that source node $s$ is a Steiner node is crucial because it ensures that the $(s,t)$-cut of the least capacity that separates $s$ and any Steiner node is exactly $\lambda$. 
In order to materialize this above idea for the situation where $s$ is not necessarily a Steiner node, we make use of the following observation. %we construct a pair of graphs $G^s(S)$ and $G^t(S)$ from $G(S)$ using the following observation. %interesting 
%property. In graph $G^s(S)$, there is a Steiner node $\mu$ such that, for each Steiner $(s,t)$-mincut $C$ in $G^s(S)$, $\mu\in C$. Similarly, in graph $G^t(S)$, there is a Steiner node $\mu$ such that, for each Steiner $(s,t)$-mincut $C$ in $G^t(S)$, $\mu\in \overline{C}$. In order to construct the pair of graphs $G^s(S)$ and $G^t(S)$ with these properties, we crucially use the following simple observation.  
\begin{observation} \label{obs : Steiner mincut covering}
    Let $\mu$ be any Steiner node in $G(S)$. For any Steiner $(s,t)$-mincut $C\in {\mathfrak C}_S$, either $\mu\in C$ or $\mu\in \overline{C}$. 
\end{observation}
 Let $\mu$ be a Steiner node in $G(S)$ other than source $s$ and sink $t$. We construct a pair of graphs $G^s(S)$ and $G^t(S)$ from $G(S)$ as follows. Graph $G^s(S)$ is obtained from $G(S)$ by adding an infinite capacity edge from $s$ to node $\mu$. Similarly, graph $G^t(S)$ is obtained from $G(S)$ by adding an edge of infinite capacity from node $\mu$ to $t$. We call node $\mu$ as the {\em pivot} for the Steiner set $S$. It follows from Observation \ref{obs : Steiner mincut covering} that each Steiner $(s,t)$-mincut belonging to ${\mathfrak C}_{S}$ is a Steiner $(s,t)$-mincut either in $G^s(S)$ or in $G^t(S)$. This ensures that each Steiner $(s,t)$-mincut in ${\mathfrak C}_S$ is preserved in exactly one of the two graphs $G^s(S)$ or $G^t(S)$. This leads to the following lemma.
\begin{lemma} \label{lem : Steiner st-mincuts}
       %Let $x,y\in S$ such that $x,y$ are mapped to nodes $\alpha, \beta$ in graph $G(S)$ respectively. There exists a pair of DAGs ${\mathcal D}_{PQ}(H_{\mu,\beta})$ and ${\mathcal D}_{PQ}(H_{\alpha,\mu})$ that stores and characterizes every Steiner $(s,t)$-mincut $C$ with $x\in C$ and $y\in \overline C$ as follows. 
        Let $x,y\in S$ such that $x,y$ are mapped to nodes $\alpha, \beta$ in graph $G(S)$ respectively.
        Let $C$ be an $(s,t)$-cut with $x\in C$ and $y\in \overline C$. $C$ is an Steiner $(s,t)$-mincut if and only if $C$ is a $1$-transversal cut in exactly one of the two DAGs ${\mathcal D}_{PQ}(G^s(S)_{(\mu,\beta)})$ and ${\mathcal D}_{PQ}(G^t(S)_{(\alpha,\mu)})$ with $\alpha\in C$ and $\beta \in \overline{C}$.
        %, where $x$ and $y$ are mapped to nodes $\alpha$ and $\beta$ in $G(S)$ respectively.
       %
       % {\color{red} 
       % For any pair of vertices $x,y\in S$, an $(s,t)$-cut $C$ with $x\in C$ and $y\in \overline C$ is an Steiner $(s,t)$-mincut if and only if $C$ is a $1$-transversal cut in exactly one of the two DAGs ${\mathcal D}_{PQ}(H_{\mu,\beta})$ and ${\mathcal D}_{PQ}(H_{\alpha,\mu})$ with $\alpha\in C$ and $\beta \in \overline{C}$, where $x$ and $y$ are mapped to nodes $\alpha$ and $\beta$ in $G(S)$ respectively.}
   \end{lemma}

   Lemma \ref{lem : Steiner st-mincuts} and the preceding discussion 
lead to the following theorem. %that summarizes the result on the compact structure for storing and characterizing all Steiner $(s,t)$-mincuts for a given Steiner set.
   \begin{theorem} \label{thm : Steiner st-mincuts}
       Let $G=(V,E)$ be a directed weighted graph on $n$ vertices and $m$ edges with a designated source vertex $s$ and a designated sink vertex $t$. For any Steiner set $S\subseteq V$, there is an ${\mathcal O}(|S|m)$ space structure, denoted by ${\mathcal Y}(S)$, consisting of at most $2|S|-2$ DAGs that collectively store all Steiner $(s,t)$-mincuts for $S$ and characterizes them as follows. An $(s,t)$-cut $C$ is a Steiner $(s,t)$-mincut for Steiner set $S$ if and only if $C$ is a $1$-transversal cut in at least one of the $2|S|-2$ DAGs.    
   \end{theorem}
\begin{remark}
    Our technique for obtaining graphs $G^s(S)$ and $G^t(S)$ from $G(S)$ is similar to {\em covering} technique in \cite{baswana2023minimum+}. However, in \cite{baswana2023minimum+}, covering technique was applied to only those graphs that have at most two $(s,t)$-mincuts -- $\{s\}$ and $V\setminus \{t\}$. Moreover, pivot node can be any arbitrary vertex in the graph. In our case, graph $G(S)$ needs not have at most two $(s,t)$-mincut, and only a Steiner node can be a pivot node.   
\end{remark}

\subsubsection{An O(n) Space Hierarchy Tree for (s,t)-cuts}
 Let $H$ be an undirected weighted graph.  Katz, Katz, Korman, and Peleg \cite{DBLP:journals/siamcomp/KatzKKP04} designed an ${\mathcal O}(n)$ space hierarchy tree ${\tilde{\mathcal T}}_H$ that stores the capacity of the cut of the least capacity separating any pair of vertices in $H$. For any directed weighted graph $G$, by taking an approach similar to \cite{DBLP:journals/siamcomp/KatzKKP04}, we design an ${\mathcal O}(n)$ space hierarchy tree $T_{\mathfrak B}$ that stores the capacity of the $(s,t)$-cut of the least capacity separating any pair of vertices in $G$. It turns out that ${\tilde{\mathcal T}}_H$ is a special case of $T_{\mathfrak B}$ as stated in Observation \ref{obs : katz tree is a special case}.

 \begin{observation}\label{obs : katz tree is a special case}
 Let $H'$ be the graph obtained by adding two isolated vertices $s$ and $t$ to $H$. Observe that, for any pair of vertices $u$ and $v$, a cut of the least capacity separating $u$ and $v$ in $H$ is also an $(s,t)$-cut of the least capacity in $H'$ that separates $u$ and $v$. \end{observation}

 For a Steiner set $S\subseteq V$, let us consider the equivalence relation ${\mathcal R}_S({\mathfrak C_S})$ on the set $S$ and the set of all Steiner $(s,t)$-mincuts ${\mathfrak C}_S$ (Definition \ref{def : relation R}). %as follows. 
%Old
 % For a Steiner set $S\subseteq V$, we define a binary relation, denoted by ${\mathcal R}_S({\mathfrak C_S})$, on the set $S$ and the set of all Steiner $(s,t)$-mincuts ${\mathfrak C}_S$ as follows. 
% \begin{definition}[Relation ${\mathcal R}_S$]
%     For a pair of vertices $x,y\in S$, $x{\mathcal R}_S y$ if and only if there is no Steiner $(s,t)$-mincut $C\in {\mathfrak C}_S$ that separates $x$ and $y$. %such that $x\in C$ and $y\in \overline{C}$ or vice versa.       
% \end{definition}
% It is easy to prove that ${\mathcal R}_S({\mathfrak C}_S)$ is an equivalence relation.
Let ${\mathfrak B}({\mathfrak C}_S)$ denote the set of equivalence classes of $S$ formed by ${\mathcal R}_S({\mathfrak C}_S)$. So, any pair of vertices from the same equivalence class are not separated by any $(s,t)$-cut from ${\mathfrak C}_S$. We present in the following a hierarchy tree ${T_{\mathfrak B}}$ based on equivalence classes resulting from a set of relations ${\mathcal R}_S$.\\ %, where, for each set $S$, the capacity of Steiner $(s,t)$-mincut for $S$ is the capacity of $(s,t)$-cut of the least capacity separating a pair of vertices in $G$.\\

\noindent
\textbf{Notations:}\\
$r$: It is the root node of $T_{\mathfrak B}$.\\
$S({\mu})$:~Every node $\mu$ of the hierarchy tree $T_{\mathfrak B}$ is associated with a Steiner set $S({\mu})\subseteq V$.\\
$\lambda_{S(\mu)}$: It is the capacity of Steiner $(s,t)$-mincut for Steiner set $S(\mu)$ associated with a node $\mu$ in $T_{\mathfrak B}$.\\
\textsc{lca}$(\mu,\nu)$: It is the lowest common ancestor of node $\mu$ and node $\nu$ in $T_{\mathfrak B}$.\\
$\mu.cap$: It is a field at each internal node $\mu$ in $T_{\mathfrak B}$ that stores the capacity $\lambda_{S(\mu)}$.

\paragraph*{Initialization:} %Every node $\mu$ of the hierarchy tree is associated with a vertex set $U_{\mu}\subseteq V$. 
The root node $r$ of the hierarchy tree $T_{\mathfrak B}$ is associated with the vertex set $V$, that is, $S({r})=V$. %{\color{red} Let ${\mathfrak C}_{f^*}$ be the set of all $(s,t)$-mincuts of the graph $G$. For each equivalence class ${\mathfrak q}$ of ${\mathfrak B}(U_{r},{\mathfrak C}_{f^*})$, we create a new node $n_{\mathfrak q}$ and add it as a child of the root node $r$. The node $n_{\mathfrak q}$ is associated with the set ${\mathfrak q}$. }

\paragraph*{Construction:} We start from the root node $r$. Since $S(r)=V$, observe that ${\mathfrak C}_{S(r)}$ is the set of all $(s,t)$-mincuts of graph $G$. For each equivalence class ${\mathfrak q}$ of ${\mathfrak B}({\mathfrak C}_V)$, we create a new node $\alpha_{\mathfrak q}$ and add it as a child of the root node $r$. The node $\alpha_{\mathfrak q}$ is associated with the Steiner set ${\mathfrak q}$, that is $S(\alpha_{\mathfrak q})={\mathfrak q}$.\\
Let us now consider a node $\mu$ at any intermediate step of the construction. If for node $\mu$, $|S(\mu)|=1$, then the process terminates. Otherwise, it continues as follows. For each equivalence class ${\mathfrak q}$ belonging to ${\mathfrak B}({\mathfrak C}_{S(\mu)})$, a new node $\alpha_{\mathfrak q}$ is created and added as a child of ${\mu}$. The node $\alpha_{\mathfrak q}$ is associated with the Steiner set ${\mathfrak q}$, that is, $S({\alpha_{\mathfrak q}})={\mathfrak q}$. This process is recursively applied to each child of the node $\mu$ along the same line. \\ %It terminates when there is no leaf node $\nu$ that is associated with more than one vertices, that is $|S(\nu)|>1$.\\   

It follows from the construction that for an internal node $\mu$ and its child $\nu$ in $T_{\mathfrak B}$ , $S({\nu})\subset S({\mu})$. Moreover, for a pair of children $\nu$ and $\nu'$ of any node $\mu$ in $T_{\mathfrak B}$, $S(\nu)\cap S(\nu')$ is empty because equivalence classes are disjoint. These two facts ensure that the process terminates within $n$ steps. Now it is evident that there are exactly $n$ leaf nodes in ${T_{\mathfrak B}}$. Each leaf node is associated with a unique vertex in $G$. Moreover, it follows from the definition of ${\mathcal R}_S$ and construction of $T_{\mathfrak B}$ that each internal node has at least two children. Hence, the space occupied by ${T_{\mathfrak B}}$ is ${\mathcal{O}}(n)$. Let ${\mathcal L}(v)$ store the pointer to the leaf node where a vertex $v\in V$ is mapped in $T_{\mathfrak B}$. The following fact is immediate from the construction of ${T_{\mathfrak B}}$.
\begin{fact} \label{fact : lca contains the minimum cut}
    For a pair of vertices $u,v\in V$, let $\mu$ be the $\textsc{lca}$ of ${\mathcal L}(u)$ and ${\mathcal L}(v)$ in $T_{\mathfrak B}$. Then $\mu.cap$ stores the capacity of $(s,t)$-cut of the least capacity in $G$ that separates $u$ and $v$. 
\end{fact}
From the perspective of an internal node of hierarchy tree, we state the following fact.
\begin{fact} \label{fact : capacity at least lambda}
    Let $\mu$ be an internal node of $T_{\mathfrak B}$. For any pair of vertices $u,v\in S(\mu)$, the capacity of $(s,t)$-cut of the least capacity in $G$ that separates $u$ and $v$ is at least $\mu.cap$.
\end{fact}
The construction of $T_{\mathfrak B}$, Fact \ref{fact : lca contains the minimum cut}, and Fact \ref{fact : capacity at least lambda} establish the following theorem. 
\begin{theorem} \label{thm : tree with equivalence class}
    Let $G$ be any directed weighted graph $G=(V,E)$ on $n=|V|$ vertices and $m=|E|$ edges with a designated source vertex $s$ and designated sink vertex $t$. There is a hierarchy tree $T_{\mathfrak B}$ for $G$ occupying ${\mathcal O}(n)$ space such that, given any pair of vertices $\{u,v\}$ in $G$,
    tree $T_{\mathfrak B}$ can report the capacity of $(s,t)$-cut of the least capacity separating $u$ and $v$ in ${\mathcal O}(1)$ time. 
\end{theorem}

% \begin{remark} ({\color{red}Need to change here. Tree is defined above.})
%   Cheng and Hu \cite{DBLP:journals/anor/ChengH91} designed an ${\mathcal O}(n)$ space rooted tree that can also report the capacity of the $(s,t)$-cut of the least capacity separating a given pair of vertices.
%     However, it does not seem possible to augment it to design a structure that compactly stores and characterizes all mincuts for all tight edges, which is the main objective.
    
%    % Interestingly, our tree $T_{\mathfrak B}$ is useful for storing and characterizing all mincuts for all tight edges. This is because $T_{\mathfrak B}$ can be suitably augmented with the structure ${\mathcal Y}(S(\mu))$ at each internal node $\mu$ in $T_{\mathfrak B}$.
% \end{remark}

\begin{remark} 
  The ${\mathcal O}(n)$ space Ancestor tree ${\mathcal T}_{(s,t)}$ of Cheng and Hu \cite{DBLP:journals/anor/ChengH91} can also report the capacity of the $(s,t)$-cut of the least capacity separating a given pair of vertices.
    However, it does not seem possible to augment ${\mathcal T}_{(s,t)}$ to design a structure that compactly stores and characterizes all mincuts for all tight edges, which is the main objective.
\end{remark}
% \begin{remark} \label{rem : katz tree is a special case}
 %     Let $H'$ be the graph obtained by adding two isolated vertices $s$ and $t$ to $H$. Observe that, for any pair of vertices $u$ and $v$, a cut of the least capacity separating $u$ and $v$ in $H$ is also an $(s,t)$ cut of the least capacity in $H'$ that separates $u$ and $v$. 
 % \end{remark}
 We now construct a data structure, denoted by $\mathbf{T_{\mathfrak B}^A}$, by augmenting nodes of tree $\mathbf{T_{\mathfrak B}}$ from Theorem \ref{thm : tree with equivalence class} with the structure ${\mathcal Y}(S)$ from Theorem \ref{thm : Steiner st-mincuts}. This data structure provides a solution to Problem \ref{prob : compact structure for all pairs}.
\paragraph*{Construction of $\mathbf{T_{\mathfrak B}^A}$:} We obtain $T_{\mathfrak B}^A$ from $T_{\mathfrak B}$ by augmenting each internal node of $T_{\mathfrak B}$ as follows. Let $\mu$ be an internal node of $T_{\mathfrak B}$ with $k$ children $\mu_1,\mu_2,\ldots,\mu_k$. It follows from the construction of $T_{\mathfrak B}$ that all Steiner $(s,t)$-mincuts for Steiner set $S(\mu)$ partition $S(\mu)$ into $k$ disjoint subsets $S(\mu_1),S(\mu_2),\ldots,S(\mu_k)$. Therefore, for each set $S(\mu_i)$, $1\le i\le k$, there is a unique node in graph $G(S(\mu))$ to which each vertex of $S(\mu_i)$ is mapped. We select one of the nodes out of these $k$ nodes as the pivot node for $S(\mu)$ and augment internal node $\mu$ with the structure ${\mathcal Y}(S(\mu))$ (Theorem \ref{thm : Steiner st-mincuts}). \\ %Similarly, for every internal node $\nu$, we augment $\nu$ with ${\mathcal Y}(S(\nu))$. The obtained structure is denoted by $T_{\mathfrak B}^A$.
 
 We now analyze the space occupied by $T_{\mathfrak B}^A$. It follows from Theorem \ref{thm : tree with equivalence class} that the number of nodes in $T_{\mathfrak B}^A$ is ${\mathcal O}(n)$. Theorem \ref{thm : Steiner st-mincuts} ensures that each internal node $\mu$ of $T_{\mathfrak B}^A$ occupies ${\mathcal O}(km)$ space, where $k$ is the number of children of $\mu$. It is easy to observe that the total number of children for all internal nodes in $T_{\mathfrak B}^A$ is bounded by ${\mathcal O}(n)$. Therefore, overall space occupied by $T_{\mathfrak B}^A$ is ${\mathcal O}(mn)$. Thus, the following theorem is immediate.
 \begin{theorem} \label{thm : complete characterization for separating cuts}
     For a directed weighted graph $G$ on $n$ vertices and $m$ edges with a designated source vertex $s$ and a designated sink vertex $t$, there exists an ${\mathcal O}(mn)$ space structure $T_{\mathfrak B}^A$ consisting of ${\mathcal O}(n)$ DAGs that, for each pair of vertices $\{u,v\}$ in $G$,  stores all $(s,t)$-cuts of the least capacity that separate $\{u,v\}$ and characterize them as follows. An $(s,t)$-cut $C$ is an $(s,t)$-cut of the least capacity separating $u$ and $v$ if and only if $C$ is a $1$-transversal cut in at least one of ${\mathcal O}(n)$ DAGs.  
 \end{theorem}

 Exploiting the result of Theorem \ref{thm : complete characterization for separating cuts}, We present the following structure ${\mathcal S}_{vit}(G)$ that stores and characterizes all mincuts for all vital edges in graph $G$.

\paragraph*{Structure ${\mathcal S}_{vit}(G)$:} It consists of the following two structures. 
\begin{enumerate}
    \item For tight edges, structure $T_{\mathfrak B}^A$ from Theorem \ref{thm : complete characterization for separating cuts}.
    \item For loose edges, DAG ${\mathcal D}_{PQ}(G_{(u,v)})$ for each loose edge $(u,v)$.
\end{enumerate}
This completes the proof of Theorem \ref{thm : vital complete characterization}.
\begin{note}
    The previously existing data structure for Problem \ref{prob : compact structure for all pairs} is only for unweighted graphs, and it occupies ${\mathcal O}(mn)$ space \cite{baswana2023minimum+}. Moreover, it is only for those pair of vertices for which $(s,t)$-cut of the least capacity separating them has capacity minimum+1. %An ${\mathcal O}(mn)$ space structure is designed in \cite{baswana2023minimum+} for this special case of Problem \ref{prob : compact structure for all pairs}. 
    Hence, Theorem \ref{thm : vital complete characterization} provides a generalization of the result in \cite{baswana2023minimum+} for weighted graphs while matching the space bound. An important application of structure ${\mathcal S}_{vit}(G)$ is Theorem \ref{thm : main result}(3).
    %in sensitivity analysis for $(s,t)$-mincut is provided in the following section.
\end{note}
\section{An Application of Compact Structure to Sensitivity Oracle} \label{sec : application sensitivity oracle}
In this section, we establish Theorem \ref{thm : main result}(3). In particular, we address the problem of reporting DAG ${\mathcal D}_{PQ}(G')$ where $G'$ is the graph obtained from $G$ after reducing the capacity of any given edge $e$ by a given value $\Delta>0$. To solve this problem, we design a compact data structure, denoted by ${\mathcal X}_{vit}(G)$, for graph $G$ that can efficiently answer the following query. \\ %various queries on mincuts for vital edges. Queries $\textsc{AllCuts}$ and $\textsc{AllCutsAllEdges}$ are answered using data structure ${\mathcal X}_{vit}(G)$ and Query $\textsc{IfSubCut}$ is answered using data structure ${\mathcal F}_{vit}(G)$. All the queries are defined in Section \ref{sec : data structure overview}. 
 %{\color{red} We present a compact data structure ${\mathcal X}_{vit}$ obtained from the structure ${\mathcal S}_{vit}(G)$ in Theorem \ref{thm : vital complete characterization} for answering the queries.}
 
\noindent
$\textsc{AllCuts}(e):$ Returns the DAG ${\mathcal D}_{PQ}(G_e)$ that stores and characterizes all mincuts for vital edge $e$.\\

For graph $G$, we have DAG ${\mathcal D}_{PQ}(G)$ for storing and characterizing all $(s,t)$-mincuts \cite{DBLP:journals/mp/PicardQ80}. 
 DAG ${\mathcal D}_{PQ}(G)$ is obtained by contracting every strongly connected component into a single node in the residual graph corresponding to a maximum $(s,t)$-flow \cite{DBLP:journals/mp/PicardQ80}. Let $\tau$ be any topological ordering of the nodes of ${\mathcal D}_{PQ}(G)$. 
For any vertex $x$ in $G$, let $\phi(x)$ be the node in ${\mathcal D}_{PQ}(G)$ to which $x$ is mapped. Likewise, for any topological ordering $\tau$ of the nodes of ${\mathcal D}_{PQ}(G)$, 
$\tau(x)$ refers to the topological number of $\phi(x)$. Observe that $\phi(t)$ appears before $\phi(s)$ in any topological ordering of ${\mathcal D}_{PQ}(G)$. The following is an interesting property of DAG ${\mathcal D}_{PQ}(G)$, established in \cite{baswana2023minimum+}.
\begin{lemma} [Lemma 4.4 and Theorem 4.5 in \cite{baswana2023minimum+}] \label{lem : topological order}
    For any topological ordering $\tau$ of ${\mathcal D}_{PQ}(G)$, every suffix of $\tau$ is an $(s,t)$-mincut in $G$.
\end{lemma}
 Given the set of edges $E$ of graph $G$ and topological ordering $\tau$, by exploiting Lemma \ref{lem : topological order}, we present a procedure, denoted by $\textsc{Construct}\_{\mathcal D}_{PQ}$, to construct DAG ${\mathcal D}_{PQ}(G)$ in ${\mathcal O}(m)$ time as follows. This procedure is crucially used to answer Query \textsc{AllMincut} later in this section.\\

\noindent
$\textsc{Construct}\_{\mathcal D}_{PQ}(E,\tau)$: Initially, let ${\mathcal D}$ be a graph with a node set the same as the node set of $\tau$. Consider an edge $(x,y)\in E$. It follows from Lemma \ref{lem : topological order} that if for edge $(x,y)$, $\tau(x)<\tau(y)$ then $(x,y)$ is an incoming edge to the $(s,t)$-mincut defined by the smalles suffix of $\tau$ containing node $\phi(y)$. %prefix of $\phi(y)$in $\tau$. 
Therefore, it follows from Theorem \ref{thm : a special assignment of flow}$(1)$ that $(x,y)$ is a nonvital edge. So, we flip the orientation of edge $(x,y)$ and add it to ${\mathcal D}$. Similarly, if for edge $(x,y)$ we have $\tau(x)>\tau(y)$, then it follows from Lemma \ref{lem : topological order} that $(x,y)$ is a contributing edge of $(s,t)$-mincut defined by the smallest suffix of $\tau$ containing node $\phi(x)$. So, edge $(x,y)$ must appear in ${\mathcal D}_{PQ}(G)$, and hence we add $(x,y)$ to ${\mathcal D}$. Otherwise, we have $\phi(x)=\phi(y)$. In this case, edge $(x,y)$ is not appearing in ${\mathcal D}_{PQ}(G)$ as there is no $(s,t)$-mincut that separates $x$ and $y$. We repeat this procedure for every edge in $E$. It is shown in Theorem B.3 of \cite{baswana2023minimum+} that flipping each edge of ${\mathcal D}$, we arrive at DAG ${\mathcal D}_{PQ}(G)$.\\

% \noindent
% $\textsc{Construct}\_{\mathcal D}_{PQ}(E,\tau)$: Initially, let ${\mathcal D}$ be a graph with a node set the same as the node set of $\tau$. Let $(x,y)$ be an edge in graph $G$. It follows from Lemma \ref{lem : topological order} that if for edge $(x,y)$, $\tau(x)<\tau(y)$ then $(x,y)$ is an incoming edge to the $(s,t)$-mincut defined by the suffix of $\phi(y)$ %prefix of $\phi(y)$
% in $\tau$. Therefore, it follows from Theorem \ref{thm : a special assignment of flow}$(1)$ that $(x,y)$ is a nonvital edge. So, we flip the orientation of edge $(x,y)$ and add it to ${\mathcal D}$. Similarly, if for edge $(x,y)$ we have $\tau(x)>\tau(y)$, then it follows from Lemma \ref{lem : topological order} that $(x,y)$ is a contributing edge of $(s,t)$-mincut defined by the suffix of $\phi(x)$ in $\tau$. So edge $(x,y)$ must appear in ${\mathcal D}_{PQ}(G)$, and hence we add it to ${\mathcal D}$. Otherwise, edge $(x,y)$ is not appearing in ${\mathcal D}_{PQ}(G)$ as there is no $(s,t)$-mincut that separates $x$ and $y$. We repeat this procedure for every edge in $E$. It is shown in \cite{baswana2023minimum+} that flipping each edge of ${\mathcal D}$, we arrive at DAG ${\mathcal D}_{PQ}(G)$.\\

We now analyze the running time of the procedure $\textsc{Construct}\_{\mathcal D}_{PQ}$. Creating the nodes of ${\mathcal D}_{PQ}(G)$ requires ${\mathcal O}(n)$ time. The overall time is ${\mathcal O}(m)$ because, for each edge, $\textsc{Construct}\_{\mathcal D}_{PQ}$ executes only a constant number of operations. 

\subsection*{A Data Structure for Answering AllMincut}  Recall that for each vital edge $e$, ${\mathcal S}_{vit}(G)$ stores at most two DAGs for storing and characterizing all mincuts for edge $e$ as follows. For each loose edge $(u,v)$, a DAG ${\mathcal D}_{PQ}(G_{(u,v)})$ is stored. For every tight edge $(u,v)$, a pair of DAGs are stored. Let $\mu$ be the \textsc{lca} of $u$ and $v$ in $T_{\mathfrak B}^A$. The pair of DAGs are ${\mathcal D}_{PQ}(G^t(S(\mu))_{(\alpha,\nu)})$ and ${\mathcal D}_{PQ}(G(S(\mu))^s_{(\nu,\beta)})$, where $\nu$ is a pivot node in $G(S(\mu))$ and vertices $u$,  $v$ are mapped to nodes $\alpha$, $\beta$ respectively in graph $G(S(\mu))$. For simplicity, we denote ${\mathcal D}_{PQ}(G^t(S(\mu))_{(\alpha,\nu)})$ by ${\mathcal D}_{PQ}(H_{(\alpha,\nu)})$ and ${\mathcal D}_{PQ}(G^s(S(\mu))_{(\nu,\beta)})$ by ${\mathcal D}_{PQ}(H_{(\nu,\beta)})$. We now construct our data structure.

\paragraph*{Description of ${\mathcal X}_{vit}(G)$:} It consists of the following three components.
\begin{enumerate}
    \item  For each DAG ${\mathcal D}$ that belongs to ${\mathcal S}_{vit}(G)$, we store only a topological ordering of the nodes of ${\mathcal D}$ instead of ${\mathcal D}$.
    \item  Set $E$ of all edges of graph $G$.
    \item Set $E_{loose}$ of all loose edges of graph $G$.\\
\end{enumerate}

\noindent
Set $E$ and explicit marking of each vital edge as loose and tight edge occupies ${\mathcal O}(m)$ space.  ${\mathcal S}_{vit}(G)$ consists of ${\mathcal O}(n)$ DAGs. A topological ordering of the nodes of each of the ${\mathcal O}(n)$ DAGs occupies ${\mathcal O}(n)$ space. Therefore, the space occupied by data structure ${\mathcal X}_{vit}(G)$ is ${\mathcal O}(n^2)$.

 We now answer Query $\textsc{AllCuts}(e=(u,v))$ for a given vital edge $e$ using ${\mathcal X}_{vit}(G)$. Suppose $e$ is a loose edge. In ${\mathcal X}_{vit}(G)$, a topological ordering of ${\mathcal D}_{PQ}(G_{(u,v)})$ is stored. Therefore, using the procedure $\textsc{Construct}\_{\mathcal D}_{PQ}$, we can construct 
 DAG ${\mathcal D}_{PQ}(G_{(u,v)})$ in ${\mathcal O}(m)$ time. 
 Suppose $e$ is a tight edge. 
It follows from the construction of ${\mathcal X}_{vit}(G)$ that a pair of topological orderings $\tau_s$ and $\tau_t$ for the nodes of DAGs ${\mathcal D}_{PQ}(H_{(\alpha,\nu)})$ and ${\mathcal D}_{PQ}(H_{(\nu,\beta)})$, respectively, are stored in ${\mathcal X}_{vit}(G)$. 
For topological ordering $\tau_s$ and $\tau_t$, we use procedure $\textsc{Construct}\_{\mathcal D}_{PQ}$ to construct DAGs ${\mathcal D}_{PQ}(H_{(\alpha,\nu)})$ and ${\mathcal D}_{PQ}(H_{(\nu,\beta)})$ in ${\mathcal O}(m)$ time. Recall that all $(s,t)$-cuts of the least capacity that keep $u$ on the side of $s$ and $v$ on the side of $t$ are stored in the two graphs. However, it is quite possible that the node containing vertex $v$ does not belong to the sink node in ${\mathcal D}_{PQ}(H_{(\alpha,\nu)})$. Similarly, the node containing vertex $u$ might not belong to the source node in ${\mathcal D}_{PQ}(H_{(\nu,\beta)})$. But, we are interested only in those $(s,t)$-cuts that separate $u$ and $v$. For this purpose, we construct a pair of DAGs from ${\mathcal D}_{PQ}(H_{(\alpha,\nu)})$ and ${\mathcal D}_{PQ}(H_{(\nu,\beta)})$ using the following lemma, established in \cite{DBLP:journals/mp/PicardQ80}. 
\begin{lemma} [\cite{DBLP:journals/mp/PicardQ80}] \label{lem : uv dag using pq}
    Let $\{x,y\}$ be any pair of vertices in $G$ and there exists an $(s,t)$-mincut $C$ such that $x\in C$ and $y\in \overline{C}$. There is an  algorithm that, given ${\mathcal D}_{PQ}(G)$, can construct DAG ${\mathcal D}_{PQ}(G_{(x,y)})$ in ${\mathcal O}(m)$ time. 
\end{lemma}
Let ${\mathcal D}_1$ and ${\mathcal D}_2$ be the DAGs obtained from ${\mathcal D}_{PQ}(H_{(\alpha,\nu)})$ and ${\mathcal D}_{PQ}(H_{(\nu,\beta)})$ using Lemma \ref{lem : uv dag using pq} that store  the set of all mincuts for edge $(u,v)$. The final step for answering Query $\textsc{AllCuts}(e)$ is the design of DAG ${\mathcal D}_{PQ}(G_{(u,v)})$ from $G$ using the two DAGs ${\mathcal D}_1$ and ${\mathcal D}_2$. \\
Let $(x,y)$ be an edge in $G$. Note that if edge $(x,y)$ does not appear in any of the two DAGs ${\mathcal D}_1$ and ${\mathcal D}_2$, then $(x,y)$ edge must not appear in ${\mathcal D}_{PQ}(G_{(u,v)})$ as there is no mincut for edge $(u,v)$ that separates $x$ and $y$. So, we contract both endpoints $x$ and $y$ into a single node. Suppose $(x,y)$ appears in exactly one of the two DAGs ${\mathcal D}_1$ and ${\mathcal D}_2$. Then, we keep the orientation of edge $(x,y)$ as the orientation it has in ${\mathcal D}_1$ or ${\mathcal D}_2$. Now suppose $(x,y)$ appears in both the DAGs ${\mathcal D}_1$ and ${\mathcal D}_2$. The following lemma ensures that the orientation of edge $(x,y)$ either remains the same or is flipped in both the DAGs. 
\begin{lemma} \label{lem : gamma 0 for all mincuts for an edge}
    For any pair of mincuts $\{C,C'\}$ for a vital edge $e=(p,q)$, there is no edge that lies between $C\setminus C'$ and $C'\setminus C$. 
\end{lemma}
\begin{proof}
     Let $H_{(p,q)}$ be the graph obtained from $G$ by adding a pair of infinite capacity edges -- $(s,p)$ and $(q,t)$. Observe that an $(s,t)$-cut $A$ in $G$ is an $(s,t)$-mincut in $H_{(p,q)}$ if and only if $A$ is a mincut for edge $(p,q)$ in $G$. 
     Suppose there is an edge $(w,z)$ that lies between $C\setminus C'$ and $C'\setminus C$. Without loss of generality, assume that $(w,z)$ is a contributing edge of $C$. Since $C$ is an $(s,t)$-mincut in $H_{(p,q)}$, it follows from \textsc{FlowCut} property and Lemma \ref{lem : vital edge in every maximum flow} that each edge that contributes to $C$ is a vital edge in $H_{(p,q)}$. So, edge $(w,z)$ is a vital edge in $H_{(p,q)}$. Edge $(w,z)$ is an incoming edge of $C'$, and it is given that $C'$ is a mincut for vital edge $(p,q)$ in $G$. Therefore, it follows from Lemma \ref{lem:incoming-edge-irrelevant} that $(w,z)$ is a nonvital edge in $H_{(p,q)}$, a contradiction. 
\end{proof}
The following corollary is immediate from Lemma \ref{lem : gamma 0 for all mincuts for an edge}.
\begin{corollary}  \label{cor : same orientation}
      Let $(x,y)$ be any edge in $G$. For any pair of mincuts $\{C,C'\}$ for a vital edge that separate $x$ and $y$, edge $(x,y)$ is contributing to $C$ if and only if $(x,y)$ is contributing to $C'$.  
\end{corollary}
It follows from Corollary \ref{cor : same orientation} that the orientation of edge $(x,y)$ is consistent in both the DAGs  ${\mathcal D}_1$ and ${\mathcal D}_2$. So, we keep the orientation of edge $(x,y)$ as the orientation it has in ${\mathcal D}_1$ and ${\mathcal D}_2$. Repeating this process for every edge in $G$, we arrive at the DAG ${\mathcal D}_{PQ}(G_{(u,v)})$. It is easy to observe that the overall time taken by the procedure is ${\mathcal O}(m)$. Hence, Query $\textsc{AllCuts}(e)$ is answered using ${\mathcal X}_{vit}(G)$ in ${\mathcal O}(m)$ time, and the following theorem is immediate.
\begin{theorem} \label{thm : Query AllCuts}
    For any directed weighted graph $G$ on $n=|V|$ vertices with a designated source vertex $s$ and a designated sink vertex $t$, there is an ${\mathcal O}(n^2)$ space data structure ${\mathcal X}_{vit}(G)$ that, given any vital edge $(u,v)$ in $G$, can report the DAG ${\mathcal D}_{PQ}(G_{(u,v)})$ in ${\mathcal O}(m)$ time.   
\end{theorem}
The existing algorithm to compute DAG ${\mathcal D}_{PQ}(G_{(u,v)})$ from scratch invokes computing a maximum  $(s,t)$-flow in graph $G_{(u,v)}$, as stated in \cite{DBLP:journals/mp/PicardQ80}. 
Refer to Table
\ref{tab : Comparison between maxflow and query algorithm} for the state-of-the-art algorithms for computing a maximum $(s,t)$-flow. It can be observed that,
compared to the best deterministic algorithms for computing ${\mathcal D}_{PQ}(G_{(u,v)})$ from scratch, our data structure (Theorem \ref{thm : Query AllCuts}) can report ${\mathcal D}_{PQ}(G_{(u,v)})$ in time faster by polynomial factors (refer to Table \ref{tab : Comparison between maxflow and query algorithm}).

\begin{table}[ht]
\small
    \centering
    \begin{tabular}{|c|c|}
        \hline
          % \textbf{}  &  \\
         \textbf{} & Running Time \\
         \hline
          \hline
           Deterministic \& &     \\
         Strongly Polynomial \cite{DBLP:journals/jal/KingRT94}& ${\mathcal O}(mn\log_{m/n \log{n}}{n})$  \\
          %  Deterministic \& & ${\mathcal O}(\Lambda m\log~\frac{n^2}{m}\log~W)$,  \\
          % Weakly Polynomial \cite{DBLP:journals/jacm/GoldbergR98} & $\Lambda=min\{n^{3/2},m^{1/2}\}$\\
          % \hline
          %  Deterministic \& & ${\mathcal O}(mn)$ if $m={\Omega}(n^{1+\epsilon})$ \cite{DBLP:journals/jal/KingRT94}    \\
          % {\color{red}CHECK!} Strongly Polynomial \cite{DBLP:journals/jal/KingRT94, DBLP:conf/stoc/Orlin13}& ${\mathcal O}(mn)$ if $m={\mathcal O}(n^{\frac{16}{15}-\epsilon})$ \cite{DBLP:conf/stoc/Orlin13} \\
          \hline
           Deterministic \& &    \\
           Weakly Polynomial \cite{chen2023almost, van2023deterministic} & $m^{1+o(1)}\log{W}$  \\
          \hline 
\end{tabular} 
 \caption{State-of-the-art algorithms for computing a maximum $(s,t)$-flow. Here, $W$ is the maximum capacity of any edge in $G$.}
    
\label{tab : Comparison between maxflow and query algorithm}
\end{table}
\paragraph*{Answering Query $\textsc{AllMincut}$:} Given any edge $e=(u,v)$ and any value $\Delta>0$, we can verify in ${\mathcal O}(1)$ time, using tree ${\mathcal T}_{vit}(G)$ (Theorem \ref{thm : reporting value}) and $w(e)$, whether the capacity of  $(s,t)$-mincut reduces after reducing the capacity of edge $e$ by $\Delta$. Suppose the capacity of $(s,t)$-mincut reduces; otherwise report DAG ${\mathcal D}_{PQ}(G)$. It is easy to observe that each mincut for edge $e$ becomes an $(s,t)$-mincut. So, we report in ${\mathcal O}(m)$ time the DAG ${\mathcal D}_{PQ}(G_{(u,v)})$ using data structure ${\mathcal X}_{vit}(G)$ in Theorem \ref{thm : Query AllCuts}. Finally, in ${\mathcal D}_{PQ}(G_{(u,v)})$, we update the capacity of edge $e$ to $w(e)-\Delta$. This completes the proof of Theorem \ref{thm : main result}(3).

\section{Lower Bounds} \label{sec : lower bounds}
In this section, we first provide a lower bound on the worst case size of mincut cover for all edges. Later, we give a lower bound on the space for any data structure that can report the capacity of $(s,t)$-mincut after the failure of any edge in both undirected and directed graphs.
\subsection{Mincut Cover for All Edges}\label{sec : su vt lower bound}
In this section, we establish the following theorem.
\begin{theorem} \label{thm : su vt lower bound}
     There exists a directed weighted graph $G=(V,E)$ on $n=|V|$ vertices such that the number of mincuts for all edges in $G$ is $\Omega(n^2)$.    
 \end{theorem}
 In order to prove Theorem \ref{thm : su vt lower bound}, we construct a graph ${\mathcal G}$ on $2n+2$ vertices with a set of $\Omega(n^2)$ edges $E'$ that has the following property. For each pair of edges $e,e'\in E'$, the capacity of mincut for edge $e$ is different from the capacity of mincut for edge $e'$ (refer to Figure \ref{fig : n2 values}$(i)$). 

  \begin{figure}
 \centering
    \includegraphics[width=\textwidth]{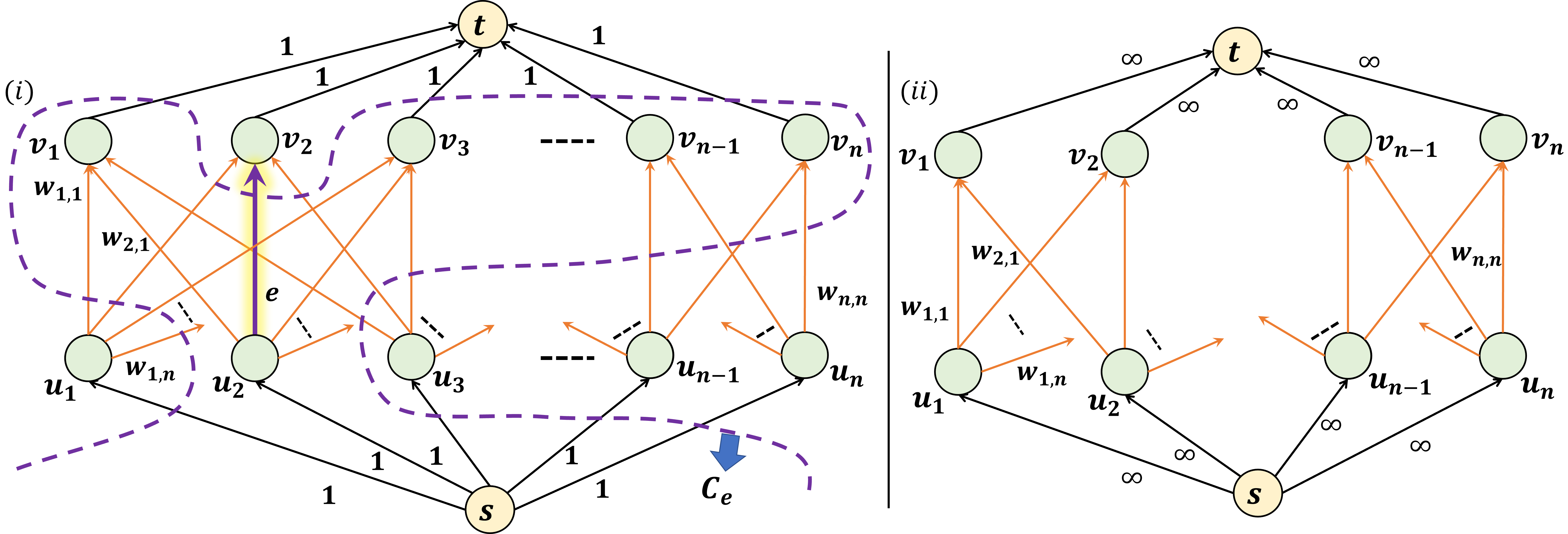} 
  \caption{($i$) Graph ${\mathcal G}$. The dashed curve represents the mincut for edge $(u_2,v_2)$. ($ii$) Graph $G(M)$. Each edge $(u_i,v_j)$, $1\le i,j \le n$, is a vital as well as tight edge.}   \label{fig : n2 values}
\end{figure}

\paragraph*{Construction of ${\mathcal G}$:} The vertex set consists of two disjoint sets $A=\{ u_1,u_2,\ldots,u_n\}$ and $B=\{ v_1,v_2,\ldots,v_n\}$ of $n$ vertices each, along with a source vertex $s$ and a sink vertex $t$. 
The edges are defined as follows. There is an edge with unit capacity from vertex $s$ to each vertex in $A$. Likewise, there is a unit capacity edge from each vertex in $B$ to sink $t$. 
For each $u\in A$ and $v\in B$, there is an edge $(u,v)$.
Let $E'$ be the set of all edges from $A$ to $B$. 
Each edge in $E'$ is assigned a unique capacity from 
$[n^2,2n^2]$. %It can be observed that
Hence, $w(e_1)\ne w(e_2)$
for every pair of edges $e_1,e_2\in E'$.
%{\color{blue} For each $u\in A$ and $v\in B$, there is an edge $(u,v)$ having $w((u,v))\ge n^2$. Let $E'$ be the set of all edges from $A$ to $B$. For each pair of edges $e_1,e_2\in E'$, $w(e_1)\ne w(e_2)$.}

 \begin{lemma} \label{lem : n2 values cuts}
     For every pair of edges $e_1$ and $e_2$ from $E'$, the capacity of mincut for $e_1$ is different from the capacity of mincut for $e_2$.
 \end{lemma}
 \begin{proof}
     Let $e=(u,v)$ be an edge from $E'$. Let us define an $(s,t)$-cut $C_e$ in which edge $e$ contributes. 
     $C_e$ keeps each vertex $x\in A\setminus \{u\}$ on the side of $t$. Similarly, $C_e$ keeps each vertex $x\in B\setminus \{v\}$ on the side of $s$. Therefore, the contributing edges of $(s,t)$-cut $C_e$ are edge $e$, the edges from $s$ to each vertex $x\in A\setminus \{u\}$, and the edges from each vertex $x\in B\setminus \{v\}$ to $t$. Hence the capacity of $C_e$ is $w(e)+(n-1)+(n-1)=w(e)+2n-2$.
     Any other $(s,t)$-cut in which edge $e$ contributes must have at least one more contributing edge $e'$ from $E'$. 
     %in addition to edge $e$. 
     Since $w(e')$ is at least $n^2$, therefore, $C_e$ is the mincut for edge $e$. 
     For any other edge $e''$ from $E'$, in a similar way, the capacity of the mincut for $e''$ is $w(e'')+2n-2$, which is different from the capacity of $C_e$ since $w(e)\ne w(e'')$. This completes the proof.
 \end{proof}
Since $|E'|=\Omega(n^2)$, it follows from Lemma \ref{lem : n2 values cuts} that graph ${\mathcal G}$ has $\Omega(n^2)$ different capacities of mincuts for edges in $E'$. This completes the proof of Theorem \ref{thm : su vt lower bound}.

\subsection{Fault-tolerant (s,t)-mincut} 
In this section, we provide a lower bound on the space for any data structure that can report the capacity of $(s,t)$-mincut after the failure of any edge.
Let $M$ be a $n\times n$ matrix where for any $i,j\in [n]$, $M[i,j]$ stores an integer in the range $[1,n^c]$ for some constant $c>0$. Given any instance of the matrix ${M}$, we construct the following undirected graph $G(M)$ (refer to Figure \ref{fig : n2 values}($ii$) ignoring the edge directions).

\paragraph*{Description of $G(M)$:}  The vertex set is defined as follows. There is a source vertex $s$ and a sink vertex $t$. For all the $n$ rows of matrix $M$, there is a set $R$ of $n$ vertices $\{u_1,u_2,\ldots,u_n\}$. Similarly, For all the $n$ columns of matrix $M$, there is a set $C$ of $n$ vertices $\{v_1,v_2,\ldots,v_n\}$. \\
For each $i,j\in [n]$, there is an edge $e_{ij}$ of capacity $w_{i,j}=M[i,j]$ between vertex $u_i$ and $v_j$. For each $i\in [n]$, there is an edge of infinite capacity between source $s$ and vertex $u_i$ and there is an edge of infinite capacity between vertex $v_i$ and sink $t$. Let $\lambda$ be the capacity of $(s,t)$-mincut in $G(M)$.\\

 We now establish a relation between the graph $G(M)$ and matrix $M$ in the following lemma.
 
\begin{lemma} \label{lem : lower bound}
   For any $i,j\in [n]$, the capacity of $(s,t)$-mincut in $G(M)$ after the failure of edge $(u_i,v_j)$ is $\lambda-M[i,j]$.
\end{lemma}
\begin{proof}
     In graph $G(M)$, $C=\{s\}\cup R$ is the only $(s,t)$-cut of finite capacity. Therefore, $C$ is the $(s,t)$-mincut  in graph $G(M)$. Observe that only edges $(u_i,v_j)$ for any $0<i,j\le n$ contribute to $C$. Hence $\lambda=\sum_{i,j \in [n]} w((u_i,v_j))$. Therefore, upon failure of edge $e_{ij}=(u_i,v_j)$, the reduction in the capacity of $(s,t)$-mincut is the same as $w(e_{ij})$, which is $M[i,j]$ as follows from the construction of $G(M)$ described above.
\end{proof}
Let $F(G(M))$ be any data structure that can report the $(s,t)$-mincut after the failure of an edge in the graph $G(M)$. We use the data structure $F(G(M))$ and Lemma \ref{lem : lower bound} to report $M[i,j]$ for any $0<i,j\le n$ as follows. 

    { \textit{return $(\lambda-\lambda')$ where $\lambda'$ is the value returned by $F(G(M))$ after failure of edge $(u_i,v_j)$.} }

Note that there are $2^{n^2c\log{n}}$ different instances of the matrix $M$. 
It follows from Lemma \ref{lem : lower bound} that for any pair of distinct instances of matrix $M$, the encoding of the corresponding data structures must differ.  
Therefore, there is an instance of matrix $M$ for which the data structure $F(G(M))$ must occupy $\Omega(n^2\log{n})$ bits of space. 

For directed graphs, a lower bound of $\Omega(n^2\log{n})$ bits of space can be achieved by orienting each undirected edge of $G(M)$ as follows. For each $i,j\in [n]$, orient edge $(u_i,v_j)$ from vertex $u_i$ to vertex $v_j$. For each $i\in [n]$, orient edge $(s,u_i)$ from source $s$ to vertex $u_i$. Similarly, for each $j\in [n]$, orient edge $(v_j,t)$ from vertex $v_j$ to sink $t$. The rest of the proof is the same as the case of undirected graphs. This completes the proof of Theorem \ref{thm : lower bound}.

By exploiting Theorem \ref{thm : su vt lower bound}, the following data structure lower bound can also be established along similar lines to the proof of Theorem \ref{thm : lower bound}.
\begin{theorem} \label{thm : n^2 space for su - vt}
    Any data structure that, given a pair of vertices $\{u,v\}$, can report the capacity of an $(s,t)$-cut $C$ of the least capacity such that $u\in C$ and $v\in \overline{C}$ in a directed weighted graph on $n$ vertices must occupy $\Omega(n^2\log{n})$ bits of space. 
\end{theorem}
%%%%%%%%%%%%%%%%%%%%%%%%%%%%SINGLE FILE ENDS%%%%%%%%%%%%%%%%%%%%%%%%%%%%%%%%%%%%

% \import{sections/}{abstract}
% \import{sections/}{introduction-2}
\bibliography{main}

\appendix

\section{Limitation of Mincut Cover for Complete Characterization} \label{app : limitation of mincut cover}
In this section, we establish the following theorem.
 \begin{theorem} \label{thm : limitation of mincut cover}
    Let $G$ be an undirected weighted graph. There exists a subset $E_{cov}$ of vital edges such that keeping a mincut for each edge in $E_{cov}$ provides a mincut cover for $G$ but the set of DAGs, ${\mathcal D}_{PQ}(G_{(x,y)})$ for each edge $(x,y)\in E_{cov}$, fails to store and characterize all mincuts for all vital edges of $G$.  
 %  
%    There exists an undirected graph $G$ with a mincut cover ${\mathcal M}$ constructed by selecting a mincut for each edge in the set of vital edges $E_{cov}\subseteq E_{vit}$ such that the set of DAGs, ${\mathcal D}_{PQ}(H_e)$ for each edge $e\in E_{cov}$, fails to store all mincuts for every vital edge in $G$.  
\end{theorem}
\begin{figure}[ht]
 \centering
    \includegraphics[width=0.4\textwidth]{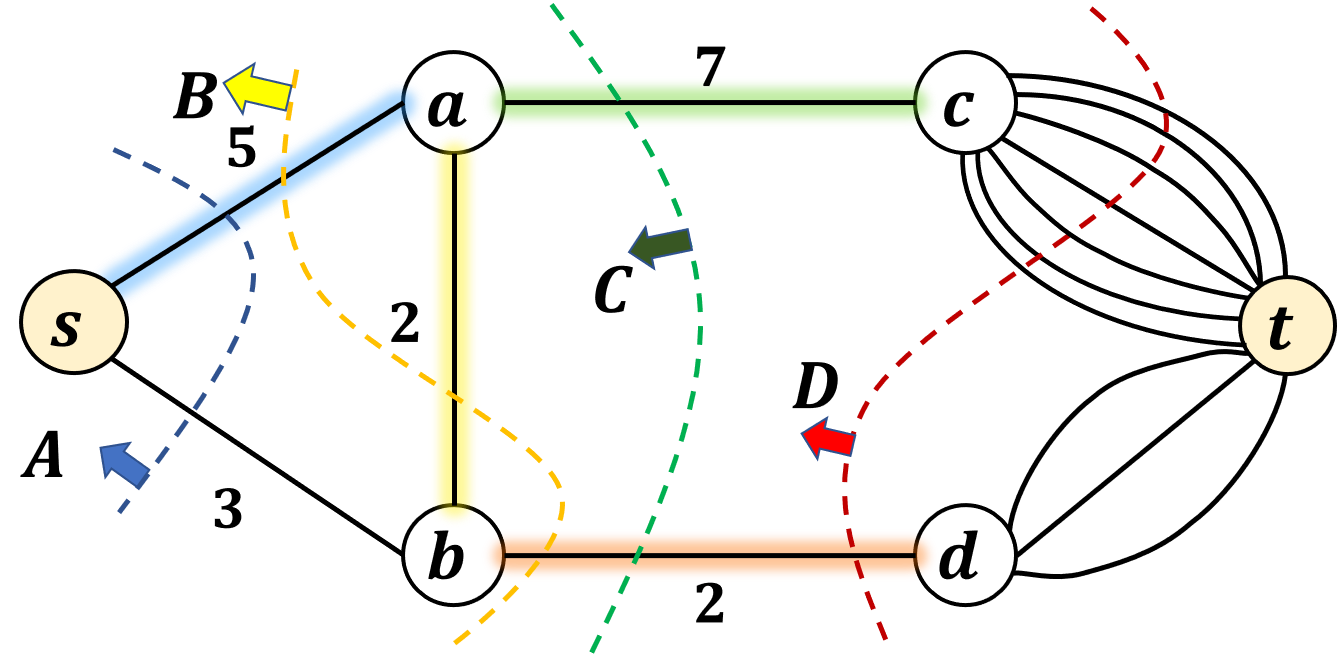} 
  \caption{Graph $H$. A vital edge $e$ and a mincut for edge $e$ appear in the same color.}
  \label{fig : limitation of mincut cover}
\end{figure}
We construct the following undirected weighted graph $H$ to establish Theorem \ref{thm : limitation of mincut cover}. We refer to Figure \ref{fig : limitation of mincut cover} for better understanding. \\%Let ${\mathcal M}$ be a smallest mincut cover for $H$. In oder to construct ${\mathcal M}$, let $E_{cov}$ be the set of edges selected by the recurrence in Section \ref{sec : cover and cardinality}. Let $e$ be an edge from $E_{cov}$. Storing a DAG ${\mathcal D}_{PQ}(H_e)$ for each edge $e\in E_{cov}$ is not sufficient for storing and characterizing all mincuts for all vital edges. 

\noindent
\textbf{Construction of $H$:}  It consists of four vertices $a,b,c,$ and $d$ other than the source vertex $s$ and the sink vertex $t$. The edges of $H$ are defined as follows. There are $7$ edges of capacity $1$ between $c$ and $t$ and $3$ edges of capacity $1$ between $d$ and $t$. Other edges are $(s,a),(s,b),(a,b),(a,c),(b,d)$ with capacities $5,3,2,7,2$ respectively. 

\begin{lemma}
    Let us consider three cuts $A=\{s\}$, $B=\{s,b\},$ and $C=\{s,a,b\}$. The set ${\mathcal M}=\{A,B,C\}$ forms a mincut cover for $H$ and it is constructed using the set of edges ${E_{cov}}=\{(s,a),(a,b),(a,c)\}$.
\end{lemma}
\begin{proof}
    The capacity of $(s,t)$-mincut in $H$ is $8$. The set of vital edges of $H$ is $\{ (s,a), (s,b), (a,b), (a,c), \\ (b,d) \}$. We arrange the set of all vital edges in the increasing order of their capacities of mincuts as follows -- $\{(s,a), (s,b), (a,b), (a,c), (b,d)\}$. Let us select the following mincuts for each edge in $E_{cov}$ according to the order -- $A$ is a mincut for edge $(s,a)$, $B$ is a mincut for edge $(a,b)$, and $C$ is a mincut for edge $(a,c)$. It is easy to observe that $A$ is also a mincut for edge $(s,b)$ and $C$ is also a mincut for edge $(b,d)$.  Therefore, it follows from the construction of mincut cover (Theorem \ref{thm : n-1 cuts}) in Section \ref{sec: mincut-cover} that ${\mathcal M}=\{A,B,C\}$ is a  mincut cover for $H$.   
\end{proof}
It is easy to observe that the set of DAGs, ${\mathcal D}_{PQ}(H_e)$ for each edge $e\in E_{cov}$, does not store the mincut $D=\{s,a,b,c,d\}$ for edge $(b,d)$. This completes the proof of Theorem \ref{thm : limitation of mincut cover}

\section{A Path Intersecting an (s,t)-cut At Least n Times} \label{app : unbounded transversality}
% {\color{blue}\textbf{A Query for reviewers:} 
%      We believe that the result established here in Appendix \ref{app : unbounded transversality} is not very hard to achieve. However, for the sake of completeness, we are keeping it in the appendix. Kindly let us know whether to keep it here or in the main body.} \\
     
\begin{figure}[ht]
 \centering
    \includegraphics[width=0.6\textwidth]{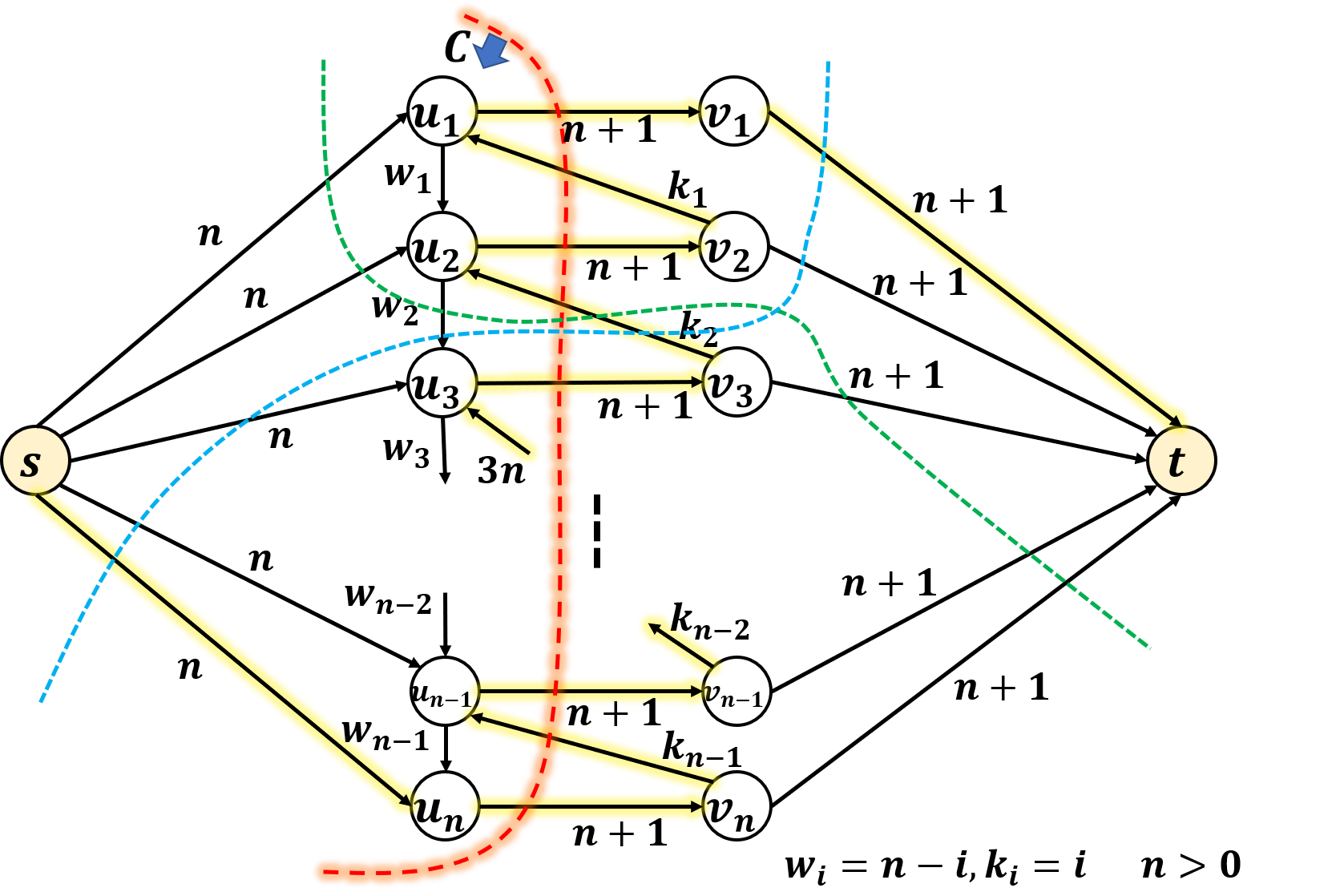} 
  \caption{  Graph $H$ is the same as the graph $Q'(H)$. In $Q'(H)$, mincut $C$ for vital edge $(u_1,v_1)$ intersects the $s$ to $t$ yellow path $2n-1$ times. }   \label{fig : n transversal with gamma non zero} 
\end{figure}

In order to establish $\Omega(n)$-transversality of mincuts for vital edges in $Q'(G)$, we construct the following graph on $2n+2$ vertices with a designated source vertex $s$ and a designated sink vertex $t$. \\

\noindent
\textbf{Construction of ${H}$ : }The vertex set contains two disjoint subsets $A=\{u_1,u_2,\ldots,u_n\}$ and $B=\{v_1,v_2,\ldots,v_n\}$ of $n$ vertices each with a source vertex $s$ and a sink vertex $t$. \\
The edge set consists of the following edges. There is an edge $(s,u)$ for each $u\in A$ of capacity $n$. Similarly, there is an edge $(v,t)$ for each $v\in B$ of capacity $n+1$. For each pair $\{u_i,v_i\}$, $1\le i \le n$ there is an edge $(u_i,v_i)$ of capacity $n+1$. For each consecutive pair $\{u_i,u_{i+1}\}$, $1\le i \le n-1$ there is an edge $(u_i,u_{i+1})$ of capacity $n-i$. Finally, for each pair $\{v_{i+1},u_i\}$, $1\le i \le n-1$, there is an edge $(v_{i+1},u_i)$ of capacity $i$. For better understanding we refer to Figure \ref{fig : n transversal with gamma non zero}.

\begin{lemma} \label{lem : graph is quotient graph}
    For each edge $e$ in $H$, there is a mincut for a vital edge in which edge $e$ is contributing.
\end{lemma}
\begin{proof}
    The $(s,t)$-mincut is $\{s\}$ and has capacity $n^2$. All outgoing edges of source $s$ are contributing to the $(s,t)$-mincut. Hence, they are vital as well. 

    Let us consider an edge $(v_i,t)$ where $1\le i \le n$. There is an $(s,t)$-cut, say $C_1^i$, for $(v_i,t)$ (except $(v_n,t)$) that keeps all vertices $u_j$ and $v_j$, $1\le j \le i$ on the side of $s$. There is another $(s,t)$-cut, say $C_2^i$, for $(v_i,t)$ (except edge $(v_1,t)$), that keeps all vertices $u_j$ and $v_j$, $i\le j \le n$ on the side of $s$. The capacity of $C_1^i$ is $(n+1)i+(n-i)+n(n-i)=n^2+n$. Similarly the capacity of $C_2^i$ is $ni+i+(n-i)(n+1)=n^2+n$. Every other $(s,t)$-cut in which edge $(v_i,t)$ is contributing has capacity at least $n^2+n$. Hence both $C_1^i$ and $C_2^i$ are mincuts for edge $(v_i,t)$. Moreover, failure of edge $(v_i,t)$ reduces capacity of $(s,t)$-mincut to $n^2+n-n-1=n^2-1$. Therefore, each edge $(v_i,t)$, $1\le i \le n$ is a vital edge. 
    For an edge $(u_i,u_{i+1})$, $1\le i \le n-1$, $C_1^i$ is the mincut for vital edge $(v_i,t)$ in which the edge is contributing. Similarly, for an edge $(v_{i+1},u_i)$, $1\le i\le n-1$, $C_2^{i+1}$ is the mincut for edge $(v_{i+1},t)$ in which the edge is contributing.

    We now consider the set $E_{m}$ of all edges $(u_i,v_i)$, $1\le i \le n$. Each edge from $E_{m}$ contributes to the $(s,t)$-cut $C$ that keeps set $A$ on the side of $s$. $C$ has a capacity $n^2+n$. Removal of each edge in $E_{m}$ reduces capacity of $(s,t)$-mincut to $n^2-1$. Hence all the edge of $E_m$ are vital edges. Moreover, $C$ is a mincut for each edge in $E_m$.
 \end{proof}
 It follows from Lemma \ref{lem : graph is quotient graph} that the quotient graph of $H$ is the graph $H$. Moreover, after flipping the orientation of all edges that are only incoming to a mincut for a vital edge from $H$, $H$ is the resulting graph.  Therefore, $Q'(H)$ is the same graph as $H$. Again, it follows from the proof of Lemma \ref{lem : graph is quotient graph} that $C$ is a mincut for vital edge $(u_1,v_1)$. However, there is a path $\langle s,u_n,v_n,u_{n-1},v_{n-1},\ldots,u_2,v_2,u_1,v_1,t\rangle$ that intersectes edge-set of $C$ exactly $2n-1$ times. Hence $C$ appears as a $2n-1$-transversal cut. This establishes Theorem \ref{thm : n transversality}.

% \begin{question}[For Reviewers] we are not keeping  in the main body because it is obstructing the readability. We  
% \end{question}

\end{document}